\documentclass[letterpaper,11pt]{article}
\usepackage{amsmath, amsthm, amssymb}
\usepackage[usenames, dvipsnames]{color}
\usepackage[normalem]{ulem} 
\usepackage{fullpage}
\usepackage[numbers, sort]{natbib}
\usepackage[noend]{algorithmic}
\usepackage{tikz, pgfplots}
\usetikzlibrary{shapes.gates.logic.US,trees,positioning,arrows}
\usepackage{enumerate}
\usepackage{hyperref}
\usepackage{xspace,color}
\usepackage{graphicx}
\usepackage{caption}
\usepackage{subcaption}
\usepackage{enumitem,linegoal}
\usepackage[linesnumbered,ruled,vlined]{algorithm2e}
\usepackage{comment}	
\usepackage{ctable}

\newtheorem{observation}{Observation}[section]
\usepackage{mathtools}
\usepackage[flushleft]{threeparttable}
\usepackage{verbatim}

\usepackage{footnote}
\makesavenoteenv{tabular}
\makesavenoteenv{table}

\usepackage[margin=1in]{geometry}

\definecolor{crimsonglory}{rgb}{0,0,0}

\definecolor{commentcolor}{rgb}{0.75, 0.0, 0.2}

 \newtheorem{theorem}{Theorem}[section]
 \newtheorem{lemma}[theorem]{Lemma}
 
 \newtheorem{corollary}[theorem]{Corollary}
 
 \newtheorem{definition}[theorem]{Definition}
 
 \newcommand{\etal}{\textit{et al.}}
 \newcommand{\eg}{\textit{e.g.}}
 \newcommand{\ie}{\textit{i.e.}}

\newif\ifqed
 
\makeatletter
\def\GrabProofArgument[#1]{ #1: \egroup\ignorespaces}
\def\proof{\noindent\textbf\bgroup Proof%
	\@ifnextchar[{\GrabProofArgument}{. \egroup\ignorespaces}\qedtrue}

\def\qedhere{\tag*{\qed}\global\qedfalse}
\makeatother

\usetikzlibrary{arrows,shapes,snakes,automata,backgrounds,petri,calc}
\usepackage[latin1]{inputenc}

\usepackage{relsize}
\usepackage{xspace}

\newcommand{\eps}{\epsilon}
\newcommand{\fa}{\mathcal{F}_a}
\newcommand{\fb}{\mathcal{F}_b}
\newcommand{\zarib}{\mathsf{C}}
\newcommand{\polylog}{\mathsf{polylog}}

\newcommand{\Problem}[1]{\textsf{#1}\xspace}
\newcommand{\FProblem}[1]{\color{blue}\textsf{#1}\color{black}\xspace}

\newcommand{\pr}{\mathsf{Pr}}

\newcommand{\ii}{\mathsf{I}}
\newcommand{\smax}{\mathsf{s_{\max}}}
\newcommand{\best}{\mathsf{H}}
\newcommand{\vmax}{\mathsf{v_{\max}}}
\newcommand{\emax}{\mathsf{e_{\max}}}
\newcommand{\wmax}{\mathsf{w_{\max}}}
\newcommand{\dmax}{\mathsf{d_{\max}}}

\newcommand{\tildorder}{\widetilde O}
\newcommand{\tildomega}{\widetilde \Omega}
\newcommand{\ttimes}{\star}
\newcommand{\projection}{\mathcal{P}}
\renewcommand{\k}{\mathsf{k}}

\newcommand{\classicknapsackalgorithm}{\mathsf{ClassicKnapsack}}
\newcommand{\infinitymultiplicitiesalgorithm}{\mathsf{KnapsackWithInfiniteMultiplicities}}
\newcommand{\givenmultiplicitiesalgorithm}{\mathsf{KnapsackWithGivenMultiplicities}}
\newcommand{\knapsackforsmallsizes}{\mathsf{KnapsackForSmallSizes}}

\makeatletter
\def\hlinewd#1{%
\noalign{\ifnum0=`}\fi\hrule \@height #1 %
\futurelet\reserved@a\@xhline}
\makeatother

\newcommand*\samethanks[1][\value{footnote}]{\footnotemark[#1]}

\newcounter{proccnt}

\newcommand{\konote}[1]{}

\title{Fast Algorithms for Knapsack via Convolution and Prediction\footnote{The conference version of the paper is appeared in STOC'18.}}

\author{
	MohammadHossein Bateni\thanks{Google Inc., New York}
	\and MohammadTaghi HajiAghayi\thanks{University of Maryland, College Park}
	\thanks{Supported in part by NSF CAREER award CCF-1053605,  NSF BIGDATA grant IIS-1546108, NSF AF:Medium grant CCF-1161365,   
		DARPA GRAPHS/AFOSR grant FA9550-12-1-0423, and another DARPA SIMPLEX grant.}
	\and Saeed Seddighin\samethanks[1]\samethanks[2]
	\and Clifford Stein\thanks{Columbia University, New York, NY 10027, USA} \thanks{Research supported in
		part by NSF grants CCF-1421161 and CCF-1714818. Some research done while visiting
		Google}
}

\begin{document}
	\newcommand{\ignore}[1]{}
\sloppy
\date{}

\maketitle

\thispagestyle{empty}

\begin{abstract}
The \Problem{knapsack} problem is a fundamental problem in combinatorial optimization.
It has been studied extensively from theoretical as well as practical perspectives
as 
it is one of the most well-known NP-hard problems.
The goal is to pack a knapsack of size $t$ with the maximum value from a collection
of $n$ items with given sizes and values.  

Recent evidence suggests that a classic $O(nt)$ dynamic-programming
solution for the \Problem{knapsack} problem might be the fastest in
the worst case.  In fact, solving the \Problem{knapsack} problem was
shown to be computationally equivalent to the \Problem{$(\min, +)$ convolution}
problem, which is thought to be facing a
quadratic-time barrier.  This hardness is 
in 
contrast to the more famous \Problem{$(+, \cdot)$ convolution}
(generally known as \Problem{polynomial multiplication}), that has an
$O(n\log n)$-time solution via Fast Fourier Transform.

Our main results are algorithms with near-linear running times (in terms of the size of the knapsack and the number of items)
for the \Problem{knapsack} problem,
if either the values or sizes of items are small integers.  
More specifically, if item sizes are integers bounded by $\smax$, the running time of
our algorithm is $\tilde O((n+t)\smax)$.
If the item values are integers bounded by $\vmax$, our algorithm runs in time
$\tilde O(n+t\vmax)$.
Best previously known running times were $O(nt)$, $O(n^2\smax)$ and
$O(n\smax\vmax)$ (Pisinger, J. of Alg., 1999).

At the core of our algorithms lies the \emph{prediction technique}:
Roughly speaking, this new technique enables us to compute
the convolution of two vectors in time $\tildorder(n\emax)$ 
when an approximation of the solution within an additive error
of $\emax$ is available. 

Our results also improve the best known strongly polynomial time solutions for knapsack. In the limited size setting, when the items have multiplicities, the fastest strongly polynomial time algorithms for knapsack run in time $O(n^2 \smax^2)$ and $O(n^3 \smax^2)$ for the cases of infinite and given multiplicities, respectively. Our results improve both running times by a factor of $\tildomega(n \max \{1, n/\smax\})$.



\end{abstract}
\section{Introduction}

The \FProblem{knapsack} problem is a fundamental problem in combinatorial optimization.
It has been studied extensively from theoretical as well as practical perspectives
(\!\eg, \cite{Bellman:1957,HS74,MT90:book,Chvatal80,Pisinger}), as
it is one of the most well-known NP-hard problems~\cite{GJ90:book}.
The goal is to pack a knapsack of size $t$ with the maximum value from a collection
of $n$ items with given sizes and values.  
More formally, item $i$ has size $s_i$ and
value $v_i$, and we want to maximize $\sum_{i\in S} v_i$ such that $S\subseteq [n]$
and $\sum_{i\in S} s_i \leq t$.

Recent evidence suggests that a classic $O(nt)$ dynamic-programming solution for
the \Problem{knapsack} problem~\cite{Bellman:1957} might be the fastest in the worst case.  In fact, solving the
\Problem{knapsack} problem was shown to be equivalent to the \Problem{$(\min, +)$ convolution}
problem~\cite{cygan2017problems}, which is thought to be facing a quadratic-time barrier.  
The two-dimensional extension, called the \Problem{$(\min, +)$ matrix product} problem,
appears in several conditional hardness results.
These hardness results for $(\min, +)$ matrix product and equivalently $(\max, +)$ matrix product are in contrast to the
more famous \Problem{$(+, \cdot)$ convolution} (generally known as \Problem{polynomial
multiplication}), that has an $O(n\log n)$-time solution via Fast Fourier Transform 
(FFT)~\cite{thomas2001introduction}.

Before moving forward, we present the general form of \FProblem{convolution} problems.
Consider two vectors $a = (a_0, a_1, \dots, a_{m-1})$ and $b = (b_0, b_1, \dots, b_{n-1})$.
We use the notations $|a| = m$ and $|b| = n$ to denote the size of the vectors.
For two associative binary operations $\oplus$ and $\otimes$,
the \FProblem{$(\oplus, \otimes)$ convolution} of $a$ and $b$ is a vector $c = (c_0, c_1, \ldots, c_{2n-1})$, defined as follows.
\begin{align*}
c_i = \raisebox{-1ex}{$\substack{{\mbox{\relscale{1.8}$\oplus$}}\\
                   \mathsmaller{j: 0\leq j<m}\\
                   \mathsmaller{0\leq i-j<n}}$}
\{ a_j \otimes b_{i-j} \},
\qquad \mbox{ for } 0 \leq i < m+n-1.
\end{align*}
\setcounter{page}{1}
The past few years have seen increased attention towards several variants of \Problem{convolution} problems
(\!\eg, \cite{BCDEHILPT12,chan2015clustered,bringmann2017near,koiliaris2017faster,%
backurs2017better,cygan2017problems,KPS17}).
Most importantly, many problems, such as 
\Problem{tree sparsity}, \Problem{subset sum}, and \Problem{$3$-sum},
 have been shown to have  conditional lower bounds on their running time
via  their intimate connection with \Problem{$(\min, +)$ convolution}.

In particular, previous studies have shown that \Problem{(max,+)
  convolution}, \Problem{knapsack}, and \Problem{tree sparsity} are
computationally (almost) equivalent~\cite{cygan2017problems}. 
However, these hardness results are obtained by constructing instances
with arbitrarily high item values (in the case of \Problem{knapsack})
or vertex weights (in the case of \Problem{tree sparsity}). A fast
algorithm can solve \Problem{$(\min,+)$ convolution} in almost linear time when
the vector elements are bounded. This raises the question of
whether moderate instances of \Problem{knapsack} or \Problem{tree
  sparsity} can be solved in subquadratic time. The recent breakthrough of Chan and
Lewenstein~\cite{chan2015clustered} implicitly suggests that \Problem{knapsack}
and \Problem{tree sparsity} may be solved in barely subquadratic time
$O(n^{1.859})$ when the values or weights are small\footnote{It follows from
  the reduction of~\cite{cygan2017problems} that any subquadratic algorithm for
\Problem{convolution} yields a subquadratic algorithm for \Problem{knapsack}.}. 

Our main results are algorithms with near-linear running times
for the \Problem{knapsack} problem,
if either the values or sizes of items are small integers.  
More specifically, if item sizes are integers bounded by $\smax$, the running time of
our algorithm is $\tilde O((n+t)\smax)$.
If the item values are integers bounded by $\vmax$, our algorithm runs in time
$\tilde O(n+t\vmax)$.
Best previously known running times were $O(nt)$, $O(n^2\smax)$ and
$O(n\smax\vmax)$~\cite{Pisinger}.
As with prior work, we focus on two special cases of
\Problem{$0/1$ knapsack} (each item may be used at most once)
and \Problem{unbounded knapsack} (each item can be used many
times), but unlike previous work we present near linear-time exact
algorithms for these problems.

Our results are similar in spirit to the work of
Zwick~\cite{zwick2002all}  wherein the author obtains a
subcubic time algorithm for the \Problem{all pairs shortest paths
  problem} (\FProblem{APSP}) where the edge weights are small
integers. Similar to \Problem{knapsack} and \Problem{$(\max,+)$
  convolution}, there is a belief that \Problem{APSP} cannot be solved
in truly subcubic time. We obtain our results through new
sophisticated algorithms for improving the running time of \Problem{convolution}
in certain settings whereas Zwick uses the known convolution
techniques as black box and develops randomized algorithms to improve
the running time of \Problem{APSP}.

We emphasize that our work  does not improve the complexity of the general
\Problem{$(\min, +)$ convolution} problem, for which no strongly subquadratic-time
algorithm is known to exist.
Nevertheless, our techniques provide almost linear running time for the parameterized
case of \Problem{$(\min, +)$ convolution} when the input numbers are bounded by the parameters.

A summary of the previous known results along with our new results is shown 
in Table~\ref{table:runningtimes}. Notice that in \Problem{0/1 knapsack}, 
$t$ is always bounded by $n \, \smax$ and thus our results improve the 
previously known algorithms even when $t$ appears in the running time.
\begin{table}[h]
\caption{$n$ and $t$ denote the number of items and the knapsack size respectively. $\vmax$ and $\smax$ denote the maximum value and size of the items. Notice that when the knapsack problem does not have multiplicity, $t$ is always bounded by $n \smax$ and thus our running times are always better than the previously known algorithms. Theorems~\ref{theorem:limitedsize},~\ref{theorem:strong1}, and~\ref{theorem:strong2}, as well as Corollary~\ref{corollary:unboundedlimitedsize} are randomized and output a correct solution with probability at least $1-n^{-10}$.}\label{table:runningtimes}
\def\outerrule{\hlinewd{1pt}}
\def\innerrule{\hlinewd{0.3pt}}
\centerline{
\begin{tabular}{lcc}
\outerrule
\bf setting & \bf running time & \bf our improvement\\
\outerrule
general setting & $O(nt)$~\cite{thomas2001introduction} & -\\
\innerrule
limited size knapsack & $O(n^2 \smax)$~\cite{Pisinger} & $\tildorder((n+t)\smax)$\\
& &  \color{red}(Theorem~\ref{theorem:limitedsize})\color{black}\\
\innerrule
limited size knapsack, unlimited multiplicity & $O(n^2 \smax ^2)$~\cite{tamir2009new} & $\tildorder(n \smax + \smax^2 \min\{n, \smax\})$\\
& &  \color{red}(Theorem~\ref{theorem:strong1})\color{black}\\
& &  $\tildorder((n+t) \smax)$\\
& &  \color{red}(Corollary~\ref{corollary:unboundedlimitedsize})\color{black}\\
\innerrule
limited size knapsack, given multiplicity & $O(n^3 \smax ^2)$~\cite{tamir2009new} & $\tildorder(n \smax^2 \min\{n, \smax\})$\\
& &  \color{red}(Theorem~\ref{theorem:strong2})\color{black}\\
\innerrule
limited value knapsack & - & $\tildorder(n+t \vmax)$\\
& &  \color{red}(Theorem~\ref{theorem:knapsack})\color{black}\\
\innerrule
limited value knapsack, unlimited multiplicity & - & $\tildorder(n + t \vmax)$\\
& &  \color{red}(Theorem~\ref{theorem:unboundedknapsack})\color{black}\\
\innerrule
limited value and size & $O(n \smax \vmax)$~\cite{Pisinger}& $\tildorder((n+t) \min\{\vmax,\smax\})$\\
& &  \color{red}(Theorems~\ref{theorem:knapsack} and~\ref{theorem:limitedsize})\color{black}\\
\outerrule
\end{tabular}}
\end{table}







\section{Our Contribution}
\subsection{Our Technique}
Recall that the  \Problem{$(+,\cdot)$ convolution} is indeed \Problem{polynomial multiplication}.
In this work, we are mostly concerned with \Problem{$(\max, +)$
  convolution} (which is computationally equivalent to
\Problem{minimum convolution}).  We may drop all qualifiers and simply
call it \Problem{convolution}. 
We use the notation \FProblem{$a\star b$} for
\Problem{$(\max, +)$ convolution} and \FProblem{$a\times b$} for
\Problem{polynomial multiplication} of two vectors $a$ and $b$. Also we denote by $a^{\ttimes k}$  
the $k$'th power of $a$ in the $(\max,+)$ setting, 
that is $\underbrace{a \ttimes a \ttimes \ldots \ttimes a}_{k \text{ times }}$.

If there is no size or value constraint, it has been shown that \Problem{knapsack}
and \Problem{$(\max, +)$ convolution} are computationally equivalent with respect
to subquadratic algorithms~\cite{cygan2017problems}. In other words, 
any subquadratic solution for \Problem{knapsack} yields a subquadratic solution for
\Problem{$(\max, +)$ convolution} and vice versa. Following this intuition,
our algorithms are closely related to algorithms for computing \Problem{$(\max, +)$ convolution}
in restricted settings. The main contribution of this work is a technique for computing the
\Problem{$(\max, +)$ convolution} of two vectors, namely \textit{the prediction technique}.
Roughly speaking, the prediction technique enables us to compute the convolution of two vectors
in time $\tildorder(n\emax)$ when an approximation of the solution within an additive error of $\emax$ is given.
As we show in Sections~\ref{sec:knapsack} and~\ref{sec:power}, this method can be applied to the
\Problem{0/1 knapsack} and \Problem{unbounded knapsack} problems to solve them
in $\tildorder(n\,\emax)$ time (\!\eg, if $\emax \geq \vmax$). In Section~\ref{sec:maxplusconvolution}, we explain the \textbf{prediction technique}
in three steps:
\begin{enumerate}
\item \textbf{Reduction to polynomial multiplication:} We make use of a classic reduction to compute $a \ttimes b$ in time $\tildorder(\emax(|a|+|b|))$ when all values of $a$ and $b$ are integers in the range $[0,\emax]$. This reduction has been used in many previous works (\!\eg, \cite{zwick2002all,chan2015clustered,bringmann2017near,backurs2017better,zwick1998all}). In addition to this, we show that when the values are not necessarily integral, an approximation solution with additive error $1$ can be found in time $\tildorder(\emax(|a|+|b|))$. We give a detailed explanation of this in Section~\ref{sec:maxplusconvolution:reduction}.

\item \textbf{Small distortion case:} Recall that \Problem{$a \star b$} denotes the \Problem{$(\max,+)$ convolution} of vectors $a$ and $b$.
In the second step, we define the ``small distortion'' case where
$a_i+b_j \geq (a\star b)_{i+j} - \emax$ for all $i$ and $j$.
Notice that the case where all input values are in the range $[0,\emax]$ is a special case of the small distortion case.
Given such a constraint, we show that $a\star b$ can be computed in time $\tilde O(\emax n)$
using the reduction to polynomial multiplication described in the first step. 
We obtain this result via two observations:
\begin{enumerate}
\item If we add a constant value $C$ to each component of either $a$ or $b$,
each component of their ``product'' $a\star b$ increases by the same amount $C$.
\item For a given constant $C$, adding a quantity $iC$ to every element $a_i$ and $b_i$ of the vectors $a$ and $b$, for all $i$,
results in an increase of $iC$ in  $(a \ttimes b)_i$ for every $0 \leq i < |a \ttimes b|$ (here $|a \ttimes b|$ denotes the size of vector $a \ttimes b$).
\end{enumerate}
These two operations help us transform the vectors $a$ and $b$ such that all elements fall in the range $[0, O(\emax)]$.
Next, we approximate the convolution of the transformed vectors via the results of the first step, and eventually
compute $a \ttimes b$ in time $\tildorder(\emax n)$.
We give more details in Section~\ref{sec:maxplusconvolution:solutionrange}.

\item \textbf{Prediction}: We state the {\it prediction technique} in Section~\ref{sec:maxplusconvolution:prediction}.
Roughly speaking, when an estimate of each component of the convolution is available,
with additive error $\emax$,
this method lets us compute the convolution in time $\tilde O(\emax n)$.
More precisely, in the prediction technique, we are given two integer vectors $a$ and $b$,
as well as $|a|$ intervals $[x_i, y_i]$.
We are guaranteed that (1) for every $0 \leq i < |a|$ and $x_i \leq j \leq y_i$,
the difference between $(a \ttimes b)_{i+j}$ and $a_i + b_j$ is at most $\emax$;
(2)  for every $0  \leq i < |a \ttimes b|$ we know that for at least one $j$
we have $a_j + b_{i-j} = (a \ttimes b)_i$ and $x_i \leq j \leq y_i$; and
(3) if $i < j$, then both $x_i \leq x_j$ and $y_i \leq y_j$ hold.
We refer to the intervals as an ``uncertain solution'' for $a \ttimes b$ within an error of $\emax$.
\end{enumerate}

The reason we call such a data structure an uncertain solution is that given such a structure, one can approximate the solution in almost linear time by iterating over the indices of the resulting vector and for every index $i$ find one $j$ such that $x_j \leq i-j \leq y_j$ and approximate $(a \ttimes b)_i$ by $a_j + b_{i-j}$. Such a $j$ can be found in time $O(\log n)$ via binary search since the boundaries of the intervals are monotone. 
In the prediction technique, we show that an uncertain solution within an additive error of $\emax$ suffices to compute the convolution of two vectors in time $\tildorder(\emax n)$. We obtain this result by breaking the problem into many subproblems with the small distortion property and applying the result of the second step to compute the solution of each subproblem in time $\tildorder(\emax n)$. We show that all the subproblems can be solved in time $\tildorder(\emax n)$ in total, and based on these solutions, $a \ttimes b$ can be computed in time $\tildorder(\emax n)$.
We give more details in Section~\ref{sec:maxplusconvolution:prediction}.

\vspace{0.2cm}
{\noindent \textbf{Theorem}~\ref{theorem:prediction} [restated informally]. \textit{Given two integer vectors $a$ and $b$ and an uncertain solution for $a \ttimes b$ within an error of $\emax$, one can compute $a \ttimes b$ in time $\tildorder(\emax n)$.\\}}

Notice that in Theorem~\ref{theorem:prediction}, there is no assumption on the range of the values in the input vectors and the running time depends linearly on the accuracy of the uncertain solution. 

\subsection{Main Results}
We show in Section~\ref{sec:knapsack} that the prediction technique enables us to solve the
\Problem{0/1 knapsack} problem in time $\tildorder(\vmax t + n)$. To this end, we define the
\emph{knapsack convolution} as follows: given vectors $a$ and $b$ corresponding to the solutions of
two \Problem{knapsack} problems $\k_a$ and $\k_b$, the goal is to compute $a \ttimes b$.  If a vector $a$ is the solution of a knapsack problem, $a_i$ denotes the maximum total value of the items that can be placed in a knapsack of size $i$. The only difference between knapsack convolution and $(\max, +)$ convolution is that in the knapsack convolution both vectors adhere to knapsack structures, whereas in the $(\max,+)$ convolution there is no assumption on the values of the vectors.
We show that if in the \Problem{knapsack} problems, the values of the items are
integers bounded by $\vmax$, then an uncertain solution for $a \ttimes b$
within an error of $\vmax$ can be computed in almost linear time.
The key observation here is that one can approximate the solution of the
\Problem{knapsack} problem within an additive error of $\vmax$ as follows:
sort the items in descending order of $v_i/s_i$ and put the items in the knapsack one by one
until either we run out of items or the remaining space of the knapsack is too small for the next item.
Based on this algorithm, we compute an uncertain solution for the knapsack convolution in almost linear time
and via Theorem~\ref{theorem:prediction} compute $a \ttimes b$ in time $\tildorder(\vmax n)$.
Finally, we use the recent technique of~\cite{cygan2017problems} to reduce the
\Problem{$0/1$ knapsack} problem to the knapsack convolution.
This yields an $\tildorder(\vmax t+ n)$ time algorithm for solving the \Problem{$0/1$ knapsack} problem
when the item values are bounded by $\vmax$.

\vspace{0.2cm}
{\noindent \textbf{Theorem}~\ref{theorem:knapsack} [restated].
\textit{The \Problem{$0/1$ knapsack} problem can be solved in time $\tildorder(\vmax t+n)$ when the item values are integer numbers in the range $[0,\vmax]$.\\}}

As another application of the prediction technique,
we present an algorithm that receives a vector $a$ and an integer $k$ as input and computes
$a^{\ttimes k}$.
We show that if the values of the input vector are integers in the range $[0, \emax]$,
the total running time of the algorithm is $\tildorder(\emax|a^{\ttimes k}|)$. This improves upon the trivial $\tildorder(\emax^2 |a^{\ttimes k}|)$.
Similar to what we do in Section~\ref{sec:knapsack},
we again show that the convolution of two powers of $a$ can be approximated within a small additive error.
We use this intuition to compute an uncertain solution within an additive error of $O(\emax)$
and apply the prediction technique to compute the exact solution in time $\tildorder(\emax|a^{\ttimes k}|)$. 

\vspace{0.2cm}
{\noindent \textbf{Theorem}~\ref{theorem:power} [restated]. 
\textit{Let $a$ be an integer vector with values in the range $[0,\emax]$. For any integer $k \geq 1$, one can compute $a^{\ttimes k}$  in time $\tildorder(\emax |a^{\ttimes k}|)$.\\}}
 
As a consequence of Theorem~\ref{theorem:power}, we show that the unbounded knapsack problem can be solved in time $\tildorder(n+\vmax t)$.

\vspace{0.2cm}
{\noindent \textbf{Theorem}~\ref{theorem:unboundedknapsack} [restated].
\textit{The \Problem{unbounded knapsack} problem can be solved in time $\tildorder(n+\vmax t)$ when the item values are integers in the range $[0,\vmax]$.\\}}

To complement our results, we also study the \Problem{knapsack} problem
when the item values are unbounded real values, but the sizes are integers in the range $[1,\smax]$. 
For this case, we present a randomized algorithm that solves the problem w.h.p.\footnote{With probability at least $1-n^{-10}$.}\ in time $\tildorder(\smax(n+t))$.
The idea is to first put the items into $t/\smax$ buckets uniformly at random.
Next, we solve the problem for each bucket separately, up to a knapsack size $\tildorder(\smax)$.
We use the Bernstein's inequality to show that w.h.p.,
only a certain interval of the solution vectors are important and we can neglect the rest of the values, thereby enabling
 us to merge the solutions of the buckets efficiently.
Based on this, we present an algorithm to merge the solutions of the buckets in time $\tildorder(\smax (n+t))$, 
yielding a randomized algorithm for solving the \Problem{knapsack} problem in time $\tildorder(\smax (n+t))$ w.h.p.\
when the sizes of the items are bounded by $\smax$.

\vspace{0.2cm}
{\noindent \textbf{Theorem}~\ref{theorem:limitedsize} [restated].
\textit{There exists a randomized algorithm that correctly computes the solution of the
\Problem{knapsack} problem in time $\tildorder(\smax(n+t))$ w.h.p., when the item sizes are integers in the range $[1,\smax]$.}}

%
\subsection{Implication to Strongly Polynomial Time Algorithms}
When we parameterize the \Problem{0/1 knapsack} problem by $\max \{s_i\} \leq \smax$, one can set $t':=\min(t,n \smax)$ and solve the problem with knapsack size $t'$ in time $\tildorder((t'+n) \smax) = \tildorder(n \smax^2)$. This yields a strongly polynomial time solution for the \Problem{knapsack} problem. However, this only works when we are allowed to use each item only once. In Section~\ref{sec:knapsack-multiplicity},  we further extend this solution to the case where each item $(s_i, v_i)$ has a given multiplicity $m_i$. For this case, our algorithm runs in time $\tildorder(n \smax ^2 \min\{n, \smax\})$ when $m_i$'s are arbitrary and solves the problem in time $\tildorder(n \smax \min\{n, \smax\} )$ when $m_i = \infty$ for all $i$. Both results improve the algorithms of~\cite{tamir2009new} by a factor of $\tildomega(\max\{n, \smax\})$ in the running time. These results are all implied by Theorem~\ref{theorem:limitedsize}. 

\subsection{Further Results}

It has been previously shown that \Problem{tree sparsity}, \Problem{knapsack}, and \Problem{convolution} problems are equivalent with respect to their computational complexity. However, these reductions do not hold for the case of small integer inputs. In Sections~\ref{sec:treeseparability} and~\ref{sec:01sparsity}, we show some reductions between these problems in the small input setting. In addition to this, we introduce the \FProblem{tree separability} problem and explain its connection to the rest of the problems in both general and small integer settings. We also present a linear time algorithm for \Problem{tree separability} when the degrees of the vertices and edge weights are all small integers.
\begin{figure}[h]
\centerline{%
\begin{tikzpicture}[
    and/.style={and gate US,thick,draw,fill=red!60,rotate=90,
		anchor=east,xshift=-1mm},
    or/.style={or gate US,thick,draw,fill=blue!60,rotate=90,
		anchor=east,xshift=-1mm},
    be/.style={circle,thick,draw,fill=green!60,anchor=north,
		minimum width=0.7cm},
    tr/.style={buffer gate US,thick,draw,fill=purple!60,rotate=90,
		anchor=east,minimum width=0.8cm},
    label distance=3mm,
    every label/.style={blue},
    event/.style={rectangle,thick,draw,fill=yellow!20,text width=2cm,
		text centered,font=\sffamily,anchor=north},
    edge from parent/.style={very thick,draw=black!70},
    edge from parent path={(\tikzparentnode.south) -- ++(0,-1.05cm)
			-| (\tikzchildnode.north)},
    level 1/.style={sibling distance=7cm,level distance=1.4cm,
			growth parent anchor=south,nodes=event},
    level 2/.style={sibling distance=7cm},
    level 3/.style={sibling distance=6cm},
    level 4/.style={sibling distance=3cm}
    ]
%
   \begin{scope}[xshift=-7.5cm,yshift=-5cm,very thick,
		node distance=1.6cm,on grid,>=stealth',
		block/.style={rectangle,draw,fill=cyan!20},
		comp/.style={thin,rectangle,draw,fill=white!40}]
   \node [block] (re)					{bounded tree sparsity};
    \node[comp] (re')[below=of re, yshift=0.5cm] {$\tildorder(\wmax n)$};
   \node [block] (re2)				[right= of re, xshift=3cm]	{bounded tree separability};
    \node[comp] (re'2)[below=of re2, yshift=0.5cm] {$\tildorder(\wmax n)$};
   \node [block] (re3)				[right= of re2, xshift=2.8cm]	{bounded knapsack};
    \node[comp] (re'3)[below=of re3, yshift=0.5cm] {$\tildorder(\vmax n)$};
   \node [block] (re4)				[right= of re3, xshift=2.5cm]	{bounded convolution};
   \node[comp] (re'4)[below=of re4, yshift=0.5cm] {$\tildorder(\emax n)$};
   \node [block]	 (ca1)	[above=of re2,xshift=1.8cm]	{$\dmax$-distance bounded convolution} edge [->] (re) edge [->] (re2) edge [->] (re3) edge [->] (re4);
   \node[comp] (ca'1)[right=of ca1, xshift=3cm] {$\tildorder(\dmax n)$};
   \node [block] (s1)	[above=of ca1]		{$0/1$ tree sparsity} edge [->] (ca1);
   \node[comp] (s'1)[right=of s1, xshift=1cm] {$\tildorder(n)$};
   \draw[very thick,->]  (re.north west) -- (s1.west);
   \draw[dotted,thin] (s1) -- (s'1);
   \draw[dotted,thin] (ca1) -- (ca'1);
    \draw[dotted,thin] (re) -- (re');
        \draw[dotted,thin] (re2) -- (re'2);
            \draw[dotted,thin] (re3) -- (re'3);
                \draw[dotted,thin] (re4) -- (re'4);
   \end{scope}
\end{tikzpicture}}
\caption{Desired running times are specified in the white boxes.
  Here $a \rightarrow b$ means that an efficient algorithm for $a$ yields an efficient algorithm for $b$.}\label{fig:reductions}
\end{figure}
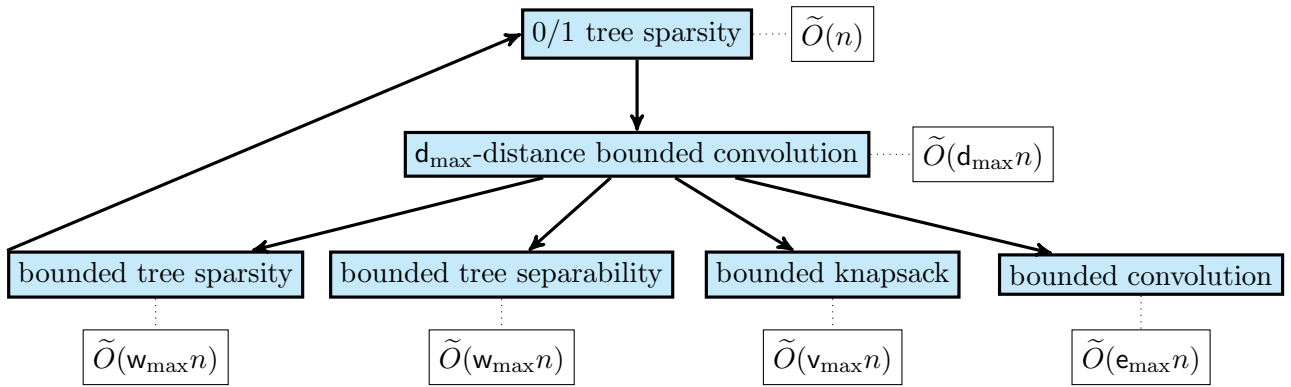


\section{The Prediction Technique for $(\max,+)$ Convolution}\label{sec:maxplusconvolution}
In this section, we present several algorithms for computing the $(\max,+)$ convolution (computationally equivalent to $(\min, +)$ convolution) of two vectors. Recall that in this problem, two vectors $a$ and $b$ are given as input and the goal is to compute a vector $c$ of length $|a|+|b|-1$ such that 
$$c_i = \max_{j=0}^{i} [a_j + b_{i-j}].$$
For this definition only, we assume that each vector $a$ or $b$ is padded on the right with sufficiently many $-\infty$ components: \ie, $a_i = -\infty$ for $i \geq |a|$ and $b_j = -\infty$ for $j \geq |b|$.

Assuming $|a| + |b| = n$, a trivial algorithm to compute $c$ from $a$ and $b$ is to iterate over all pairs of indices and compute $c$ in time $O(n^2)$. Despite the simplicity of this solution, thus far, it has remained one of the most efficient algorithms for computing the $(\max,+)$ convolution of two vectors. However, for special cases, more efficient algorithms compute the result in subquadratic time. For instance, as we show in Section~\ref{sec:maxplusconvolution:reduction}, if the values of the vectors are integers in the range $[0,\emax ]$, one can compute the $(\max,+)$ convolution of two vectors in time $\tildorder(\emax n)$.

In this section, we present several novel techniques for multiplying vectors in the $(\max, +)$ setting in truly subquadratic time under different assumptions. The main result of this section is \textit{the prediction technique} explained in Section~\ref{sec:maxplusconvolution:prediction}. Roughly speaking, we define the notion of \textit{uncertain solution} and show that if an uncertain solution of two integer vectors with an error of $\emax $ is given, then it is possible to compute the $(\max,+)$ convolution of the vectors in time $\tildorder(\emax n)$. Later in Sections~\ref{sec:knapsack} and~\ref{sec:power} we use this technique to improve the running time of the \Problem{knapsack} and other problems.

In our algorithm, we subsequently make use of a classic reduction from \Problem{$(\max,+)$ convolution} to \Problem{polynomial multiplication}. In the interest of space, we skip this part here and explain it in Section~\ref{sec:maxplusconvolution:reduction}. The same reduction has been used as a blackbox in many recent works~\cite{zwick2002all,chan2015clustered,bringmann2017near,backurs2017better,zwick1998all}. Based on this reduction, we show that an $\tildorder(\emax n)$ time algorithm can compute the convolution of two integer vectors whose values are in the range $[0,\emax ]$. We further explain in Section~\ref{sec:maxplusconvolution:reduction} that even if the values of the vectors are real but in the range $[0,\emax ]$, one can approximate the solution within an additive error less than $1$. These results hold even if the input values can be either in the interval $[0, \emax]$ or in the set $\{-\infty, \infty\}$. We use this technique in Section~\ref{sec:maxplusconvolution:solutionrange} to compute the $(\max, +)$ convolution of two integer vectors in time $\tildorder(\emax n)$ when for every $i$ and $j$ we have $|a_i + b_j - (a \ttimes b)_{i+j}| \leq \emax $. Finally, in Section~\ref{sec:maxplusconvolution:prediction} we use these results to present the prediction technique for computing the $(\max, +)$ convolution of two vectors in time $\tildorder(\emax n)$.

\subsection{An $\tildorder(\emax n)$ Time Algorithm for the Case of Small Distortion}\label{sec:maxplusconvolution:solutionrange}
In this section we study a variant of the \Problem{$(\max, +)$ convolution} problem where every $a_i + b_j$ differs from $(a \ttimes b)_{i+j}$ by at most $\emax $. Indeed this condition is strictly weaker than the condition studied in Section~\ref{sec:maxplusconvolution:reduction}. Nonetheless we show that still an $\tildorder(\emax n)$ time algorithm can compute $a \ttimes b$ if the values of the vectors are integers but not necessarily in the range $[0,\emax ]$. In the interest of space, we omit the proofs of Lemmas~\ref{lemma:rangegood},~\ref{lemma:range3}, and~\ref{lemma:range4}  and include them in Section~\ref{sec:step2-omitted}.

We first assume that both vectors $a$ and $b$ have size $n$. Moreover, since the case of $n=1$ is trivial, we assume w.l.o.g.\ that $n > 1$. In order to compute $a \ttimes b$ for two vectors $a$ and $b$, we transform them into two vectors $a'$ and $b'$ via two operations. In the first operation, we add a constant $C$ to every element of a vector. In the second operation, we fix a constant $C$ and add $iC$ to every element $i$ of \textbf{both} vectors. We delicately perform these operations on the vectors to make sure the resulting vectors $a'$ and $b'$ have small values. This enables us to approximate (and not compute since the values of $a'$ and $b'$ are no longer integers) the solution of $a' \ttimes b'$ in time $\tildorder(\emax n)$. Finally, we show how to derive the solution of $a \ttimes b$ from an approximation for $a' \ttimes b'$.  We begin by observing a property of the vectors.

\begin{lemma}\label{lemma:rangegood}
	Let $a$ and $b$ be two vectors of size $n$ such that for all $0 \leq i,j < n$ we have $(a \ttimes b)_{i+j} - a_i - b_j \leq \emax $. Then,
	\begin{itemize}
		\item for every $0 \leq i,j < n$, we have  $|(a_i - b_i) - (a_j - b_j)| \leq \emax$; and
		\item for every $0 \leq i \leq j \leq k < n$ such that $j-i = k-j$, we have $|a_j - (a_i + a_k)/2| \leq \emax $.
	\end{itemize}
\end{lemma}

%
%

Note that since there is no particular assumption on vector $a$, the condition of Lemma~\ref{lemma:rangegood} carries over to vector $b$ as well. Next we use Lemma~\ref{lemma:rangegood} to present an $\tildorder(\emax n)$ time algorithm for computing $a \ttimes b$.  
The two operations explained above help us transform the vectors $a$ and $b$ such that all elements fall in the range $[0, O(\emax)]$.
Next, we approximate the convolution of the transformed vectors via the results of Section~\ref{sec:maxplusconvolution:reduction}, and eventually
compute $a \ttimes b$ in time $\tildorder(\emax n)$.

\begin{lemma}\label{lemma:range3}
	Let $a$ and $b$ be two integer vectors of size $n$ such that for all $0 \leq i,j < n$ we have $(a \ttimes b)_{i+j} - a_i - b_j \leq \emax $. One can compute $a \ttimes b$ in time $\tildorder(\emax n)$.
\end{lemma}

All that remains is to extend our algorithm to the case where we no longer have $|a| = |b|$. We assume w.l.o.g.\ that $|b| \geq |a|$ and divide $|b|$ into $\lceil |b|/|a| \rceil$ vectors of length $|a|$ such each $b_i$ appears in at least one of these vectors. Then, in time $O(\emax |a|)$ we compute the $(\max,+)$ convolution of $a$ and each of the smaller intervals, and finally use the results to compute $a \ttimes b$ in time $\tildorder(\emax  (|a| + |b|))$.

\begin{lemma}\label{lemma:range4}
Let $a$ and $b$ be two integer vectors such that for all $0 \leq i < |a|$ and $0 \leq j < |b|$ we have $(a \ttimes b)_{i+j} - a_i - b_j \leq \emax $. One can compute $a \ttimes b$ in time $\tildorder(\emax  (|a| + |b|))$.
\end{lemma}


\subsection{Prediction}\label{sec:maxplusconvolution:prediction}
In this section, we explain the prediction technique and show how it can be used to improve the running time of classic problems when the input values are small. Roughly speaking, we show that in some cases an approximation algorithm with an additive error of $\emax $ can be used to compute the exact solution of a $(\max,+)$ convolution in time $\tildorder(\emax n)$. In general, an additive approximation of $\emax $ does not suffice to compute the $(\max,+)$ convolution in time $\tildorder(\emax n)$. However, we show that under some mild assumptions, an additive approximation yields a faster exact solution. We call this the prediction technique.

Suppose for two integer vectors $a$ and $b$ of size $n$, we wish to compute $a \ttimes b$. The values of the elements of $a$ and $b$ range over a potentially large  (say $O(n)$) interval and thus Algorithm~\ref{alg:mult1} doesn't improve the $O(n^2)$ running time of the trivial solution. However, in some cases  we can predict which $a_i$'s and $b_j$'s are far away from $(a \ttimes b)_{i+j}$. For instance, if $a$ and $b$ correspond to the solutions of two knapsack problems whose item weights are bounded by $\emax $, a well-known greedy algorithm can approximate $a \ttimes b$ within an additive error of $\emax$ ($a_i$ and $b_i$ denote the solutions of the knapsack problem for size $i$). The crux of the argument is that if we sort the items with respect to the ratio of weight over size in descending order and fill the knapsack in this order until we run out of space, we always get a solution of at most $\emax $ away from the optimal. 
Now, if $a_i + b_j$ is less than the estimated value for $(a \ttimes b)_{i+j}$ for some $i$ and $j$, then there is no way that the pair $(a_i,b_j)$ contributes to the solution of $a \ttimes b$. With a more involved argument, one could observe that whenever $a_i + b_j$ is at least $\emax $ smaller than the estimated solution for $(a \ttimes b)_{i+j}$, then $a_i + b_k < (a \ttimes b)_{i+k}$ for either all $k$'s in $[j,n-1]$ or all $k$'s in $[0,j]$. We explain this in more details in Section~\ref{sec:knapsack}.

This observation shows that in many cases, $(a_i,b_j)$ pairs that are far from $(a \ttimes b)_{i+j}$ can be trivially detected and ignored. Therefore, the main challenge is to handle the $(a_i,b_j)$ options that are close to $(a \ttimes b)_{i+j}$. Our prediction technique states that such instances can also be solved in subquadratic time. 
To this end, suppose that $a$ and $b$ are two integer vectors of size $n$,
and for every $0 \leq i < |a|$ we have an interval $[x_i, y_i]$,
and we are guaranteed that $a_i + b_j$ is at most $\emax $ away from $(a \ttimes b)_{i+j}$ for all $j \in [x_i,y_i]$. 
Also, we know that for any $0 \leq i < |a \ttimes b|$ there exists a $j$ such that $a_j + b_{i-j} = (a \ttimes b)_{i}$ and $x_j \leq i-j \leq y_j$. We call such data \textit{an uncertain solution}. We show in Theorem~\ref{theorem:prediction} that if an uncertain solution is given, then one can compute $a \ttimes b$ in time $\tildorder(\emax n)$. For empty intervals only, $y_i$ is allowed to be smaller than $x_i$.

\begin{theorem}\label{theorem:prediction}
    Let $a$ and $b$ be two integer vectors and assume we have $|a|$ intervals $[x_i, y_i]$ such that 
    \begin{itemize}
        \item $a_i + b_j \geq (a \ttimes b)_{i+j} - \emax $ for all $0 \leq i < |a|$ and $j \in [x_i,y_i]$;
        \item for all $0 \leq i < |a \ttimes b|$, there exists an index $j$ such that $a_j + b_{i-j} = (a \ttimes b)_i$ and $x_j \leq i-j \leq y_j$; and
        \item $0 \leq x_i, y_i < |b|$ for all intervals and $x_i \leq x_j $ and $y_i \leq y_j$ hold for all $0 \leq i < j < |a|$.
    \end{itemize}
   Then, one can compute $a \ttimes b$ from $a$, $b$, and the intervals in time $\tildorder(\emax  (|a| + |b|))$.
\end{theorem}
\begin{proof}
 One can set $n$ equal to the smallest power of two greater than $\max\{|a|, |b|\}$ and add extra $-\infty$'s to the end of the vectors to make sure $|a| = |b| = n$. Next, for every newly added element of $a$ we set its corresponding interval $[x_i, y_i]$ to $(n-q,n-q-1)$ (that is, an empty interval) where $q$ is the number of newly added $-\infty$'s to the end of $b$. This way, all conditions of the theorem are met and $|a|+|b|$ is multiplied by at most a constant factor. Therefore, from now on, we assume $|a| = |b| = n$ and that $n$ is a power of two. Keep in mind that for every $i$ with property $x_i \leq y_i$, none of $\{a_i,b_{x_i},b_{x_i+1},\ldots,b_{y_i}\}$ is equal to $-\infty$.

Our algorithm runs in $\log n + 1$ rounds. In every round we split $b$ into several intervals. For an interval $[\alpha,\beta]$ of $b$ we call the projection of $[\alpha,\beta]$ the set of all indices $i$ of $a$ that satisfy both $x_i \leq \alpha$ and $y_i \geq \beta$. We denote the projection of an interval $[\alpha,\beta]$ by $\projection(\alpha,\beta)$. We first show that for every $0 \leq \alpha \leq \beta < n$, $\projection(\alpha,\beta)$ corresponds to an interval of $a$. We defer the proof of Observation~\ref{observation:simple1} to Appendix~\ref{sec:step3-omitted}.

\begin{observation}\label{observation:simple1}
    For every $0 \leq \alpha \leq \beta < n$, $\projection(\alpha,\beta)$ is an interval of $a$.
\end{observation}

Furthermore, for any pair of disjoint intervals $[\alpha_1, \beta_1]$ and $[\alpha_2, \beta_2]$,
we observe that $\projection(\alpha_1, \beta_1) \setminus \projection(\alpha_2, \beta_2)$ is always an interval. Similar to Observation~\ref{observation:simple1}, we include the proof of Observation~\ref{observation:simple2} in Appendix~\ref{sec:step3-omitted}.

\begin{observation}\label{observation:simple2}
    For $0 \leq \alpha_1 \leq \beta_1 < \alpha_2 \leq \beta_2 < n$, both $\projection(\alpha_1, \beta_1) \setminus \projection(\alpha_2, \beta_2)$ and $\projection(\alpha_2, \beta_2) \setminus \projection(\alpha_1, \beta_1)$ are intervals of the indices of $a$.
\end{observation}
The proof for $\projection(\alpha_2, \beta_2) \setminus \projection(\alpha_1, \beta_1)$ being an interval follows from symmetry.

Before we start the algorithm, we construct a vector $c$ of size $2n-1$ and set all its indices equal to $-\infty$. In Round 1 of our algorithm, we only have a single interval $[\alpha_1, \beta_1] = [0,n-1]$ for $b$. Therefore, we compute $\projection(0,n-1) = [\gamma_1, \delta_1]$  and construct a vector $a^1$ of size $\delta_1 - \gamma_1 +1$ and set $a^1_i = a_{i+\gamma}$. Similarly, we construct a vector $b^1$ of size $\beta_1 - \alpha_1 + 1$ and set $b^1_i = b_{i+\alpha}$. Next, we compute $c^1 = a^1 \ttimes b^1$ using Lemma~\ref{lemma:range4}, and then based on that we set $c_{i+\alpha+\gamma} \leftarrow \max\{c_{i+\alpha+\gamma}, c'_i\}$ for all $0 \leq i < |c^1|$. 

The second round is similar to Round 1, except that this time we split $b$ into two intervals $[\alpha_1, \beta_1]$ and $[\alpha_2,\beta_2]$  where $\alpha_1 = 0, \beta_1 = n/2-1, \alpha_2 = n/2$, and $\beta_2 = n-1$. For interval $[\alpha_1, \beta_1]$ of $b$ we compute $[\gamma_1, \delta_1] = \projection(\alpha_1, \beta_1) \setminus \projection(\alpha_2, \beta_2)$ and similarly for the second interval of $b$ we compute $[\gamma_2, \delta_2] = \projection(\alpha_2, \beta_2) \setminus \projection(\alpha_1, \beta_1)$. Similar to Round 1, we construct $a^1, a^2, b^1, b^2$ from $a$ and $b$ with respect to the intervals and compute $c^1 = a^1 \ttimes b^1$ and $c^2 = a^2 \ttimes b^2$. Finally we update the solution based on $c^1$ and $c^2$.

More precisely, in Step $s+1$ of the algorithm, we split $b$ into $2^s$ intervals $[\alpha_i, \beta_i]$ where $\alpha_i = (i-1) 2^{(\log n)-s}$ and $\beta_i = i 2^{(\log n)-s} -1$. For odd intervals we compute $[\gamma_{2i+1}, \delta_{2i+1}] = \projection(\alpha_{2i+1}, \beta_{2i+1}) \setminus \projection(\alpha_{2i}, \beta_{2i})$ and for even intervals we compute $[\gamma_{2i}, \delta_{2i}] = \projection(\alpha_{2i}, \beta_{2i}) \setminus \projection(\alpha_{2i+1}, \beta_{2i+1})$. Next, we construct vectors $a^1, a^2, \ldots, a^{2^s}$ and $b^1, b^2, \ldots, b^{2^s}$ from $a$ and $b$ and compute $c^i = a^i \ttimes b^i$ using Lemma~\ref{lemma:range4} for every $1 \leq i \leq 2^s$. Finally, for every $1 \leq i \leq 2^s$ and $0 \leq j < |c^i|$, we set $c_{\alpha_i + \gamma_i + j} = \max\{c_{\alpha_i + \gamma_i + j},c^i_j\}$.\\

    \begin{algorithm}[h!]

    \KwData{Two integer vectors $a$ and $b$ of size $n$, intervals $[x_i, y_i]$ for $0 \leq i < n$ meeting the conditions of Theorem~\ref{theorem:prediction}}
    \KwResult{$a \ttimes b$}
    
    $c \leftarrow $\text{ a vector of size $2n -1 $ with indices set to $\infty$ initially}\; \label{line:prediction:1}
    
    \For {$s \in [0, \log n]$}{
        \For {$i \in [1, 2^s]$}{
            $\alpha_i \leftarrow (i-1)2^{(\log n)-s}$\; \label{line:prediction:4}
            $\beta_i \leftarrow i 2^{(\log n)-s}-1$\; \label{line:prediction:5}
        }
        
        \For {$i \in [1, 2^s]$}{
            \If {$s = 0$}{
                $[\gamma_i, \delta_i] \leftarrow \projection(\alpha_i, \beta_i)$\; \label{line:prediction:8}
            }\Else{
                \If {$i$ is odd}{
                    $[\gamma_i, \delta_i] \leftarrow \projection(\alpha_i, \beta_i) \setminus \projection(\alpha_{i+1}, \beta_{i+1})$\; \label{line:prediction:11}
                }\Else{
                    $[\gamma_i, \delta_i] \leftarrow \projection(\alpha_i, \beta_i) \setminus \projection(\alpha_{i-1}, \beta_{i-1})$\; \label{line:prediction:12}
                }
            }
            
            $a^i \leftarrow $\text{ a vector of size $\delta_i - \gamma_i + 1$ s.t. }$a^i_j = a_{\gamma_i+j}$\; \label{line:prediction:13}
            $b^i \leftarrow $\text{ a vector of size $2^{(\log n)-s}$ s.t. }$b^i_j = b_{\alpha_i+j}$\; \label{line:prediction:14}
            $c^i \leftarrow \textsf{DistortedConvolution}(a^i,  b^i, \emax )$\; \label{line:prediction:15}
            
            \For {$j \in [1,|c^i|]$}{
                $c_{\alpha_i+\gamma_i+j} \leftarrow \max\{c_{\alpha_i+\gamma_i+j}, c^i_j\}$\; \label{line:prediction:17}
            }
            
        }
    }    
    \textbf{Return } $c$\;        
    \caption{\textsf{ConvolutionViaPredictionMethod}($a,  b, \emax , x_i\text{'s}, y_i\text{'s}$)}    \label{alg:prediction}
\end{algorithm}
We show that (i) Algorithm~\ref{alg:prediction} finds a correct solution for $a \ttimes b$, and (ii) its running time is $\tildorder(\emax (|a| + |b|))$. Observe that Line~\ref{line:prediction:1} of Algorithm~\ref{alg:prediction} runs in time $O(n)$ and all basic operations (\!\eg, Lines~\ref{line:prediction:4} and~\ref{line:prediction:5}) run in time $O(1)$ and thus all these lines  in total take time $O(n \log n) = \tildorder(n)$. Moreover, for any $[\alpha, \beta]$, $\projection(\alpha,\beta)$ can be found in time $O(\log n)$ by binary searching the indices of $a$. More precisely, in order to find $\projection(\alpha,\beta)$ we need to find an index $\gamma$ of $a$ such that $x_\gamma \leq \alpha$ and an index $\delta$ such that $y_\delta \geq \beta$. Since both $x$ and $y$ are non-decreasing, we can find such indices in time $O(\log n)$. Therefore, the total running times of Lines~\ref{line:prediction:8},~\ref{line:prediction:11}, and~\ref{line:prediction:12} is $O(n \log^2 n) = \tildorder(n)$. The running time of the rest of the operations (Lines~\ref{line:prediction:13},~\ref{line:prediction:14},~\ref{line:prediction:15}, and~\ref{line:prediction:17}) depend on the length of the intervals $[\alpha_i, \beta_i]$ and $[\gamma_i, \delta_i]$. 
For a Round $s+1$, let $\ell_a = |a^1| + |a^2| + \ldots, |a^{2^s}|$ be the total length of the intervals $[\gamma_i, \delta_i]$. Similarly, define $\ell_b = |b^1| + |b^2| + \ldots + |b^{2^s}|$ and $\ell_c = |c^1| + |c^2| + \ldots + |c^{2^s}|$ as the total length of the intervals $[\alpha_i, \beta_i]$ and vectors $c^i$. It follow from the algorithm that in Round $s+1$, the running time of Lines~\ref{line:prediction:13},~\ref{line:prediction:14}, and~\ref{line:prediction:17} is $\tildorder(\ell_c)$ and the running time of Line~\ref{line:prediction:15} is $\tildorder(\emax \ell_c)$. Therefore, it only suffices to show that $\ell_c = O(n)$ to prove Algorithm~\ref{alg:prediction} runs in time $\tildorder(\emax n)$.

Notice that in every Round $s+1$ we have $|b_i| = 2^{\log n - s}$ and thus $\ell_b = 2^s 2^{\log n - s} = n$. Moreover, for every $c^i$ we have $c^i = a^i \ttimes b^i$ and thus $|c^i| \leq |a^i| + |b^i|$. Therefore, $\ell_c \leq \ell_a + \ell_b  = \ell_a + n$. Thus, in order to show $\ell_c = O(n)$, we need to prove that $\ell_a = O(n)$. To this end, we argue that for every $0 \leq i < n$, the $i$'th element of $a$ appears in at most two intervals of $[\gamma_i, \delta_i]$. Suppose for the sake of contradiction that for $0 \leq \alpha_{j_1} < \beta_{j_1} < \alpha_{j_2} < \beta_{j_2} < \alpha_{j_3} < \beta_{j_3}$ we have $i \in [\gamma_{j_1}, \delta_{j_1}] \cap [\gamma_{j_2}, \delta_{j_2}] \cap [\gamma_{j_3}, \delta_{j_3}]$. Recall that depending on the parity of $j_2$, $[\gamma_{j_2}, \delta_{j_2}]$ is either equal to $\projection(\alpha_{j_2},\beta_{j_2}) \setminus \projection(\alpha_{j_2+1},\beta_{j_2+1})$ or $\projection(\alpha_{j_2},\beta_{j_2}) \setminus \projection(\alpha_{j_2-1},\beta_{j_2-1})$ and since $i \in [\gamma_{j_2}, \delta_{j_2}]$ then either of $i \notin \projection(\alpha_{j_2-1},\beta_{j_2-1})$ or $i \notin \projection(\alpha_{j_2+1},\beta_{j_2+1})$ hold. This implies that either $y_i < \beta_{j_2+1}$ or $x_i > \alpha_{j_2-1}$ which imply either $i \notin \projection(\alpha_{j_1}, \beta_{j_1})$ or $i \notin \projection(\alpha_{j_3}, \beta_{j_3})$ which is a contradiction. Thus, $\ell_a \leq 2n$ and therefore $\ell_c \leq 3n$. This shows that Algorithm~\ref{alg:prediction} runs in time $\tildorder(\emax n)$.

To prove correctness, we show that (i) every $a^i$ and $b^i$ meet the condition of Lemma~\ref{lemma:range4}, and (ii) for every $a_i$ and $b_j$ such that $j \in [x_i, y_i]$ in some round of the algorithm and for some $k$, $a^k$ contains $a_i$ and $b^k$ contains $b_j$.

We start with the former. Due to our algorithm, in every round for every $[\alpha_i, \beta_i]$ we have $[\gamma_i, \delta_i] \subseteq \projection(\alpha_i, \beta_i)$. This implies that for every $i' \in [\gamma_i, \delta_i]$ and every $j' \in [\alpha_i, \beta_i]$ we have $$a^i_{i'-\gamma_i} + b^i_{j'-\alpha_i} - \emax  = a_{i'} + b_{j'} - \emax  \geq (a \ttimes b)_{i'+j'} \geq (a^i \ttimes b^i)_{i'+j'-\gamma_{i'}-\alpha_{j'}}.$$ Thus, the condition of Lemma~\ref{lemma:range4} holds for every $a^i$ and $b^i$.

We finally show that for every $0 \leq i < n$ and every $0 \leq j < n$ such that $j \in [x_i, y_i]$, in some round of the algorithm we have $j \in [\alpha_k, \beta_k]$ and $i \in [\gamma_k, \delta_k]$ for some $k$. To this end, consider the first Round $s+1$ in which $i \in \projection(\alpha_{\lceil j/2^{\log n-s}\rceil}, \beta_{\lceil j/2^{\log n-s}\rceil})$. We know that this eventually happens in some round since in Round $\log n+1$ we have $i \in \projection(\alpha_{\lceil j/2^{\log n-\log n}\rceil}, \beta_{\lceil j/2^{\log n-\log n}\rceil}) = \projection(j,j)$. Round $s+1$ is the first round that $i \in \projection(\alpha_{\lceil j/2^{\log n-s}\rceil}, \beta_{\lceil j/2^{\log n-s}\rceil})$ happens and thus either $s = 0$  or $s > 0$. The former completes the proof since it yields $i \in [\alpha_{\lceil j/2^{\log n-s}\rceil}, \beta_{\lceil j/2^{\log n-s}\rceil}]$. The latter implies that $i \notin [\alpha_{\lceil j/2^{\log n-s+1}\rceil}, \beta_{\lceil j/2^{\log n-s+1}\rceil}]$ and thus $i \in [\alpha_{\lceil j/2^{\log n-s}\rceil}, \beta_{\lceil j/2^{\log n-s}\rceil}]$. Thus, in Round $s+1$ we have $i \in [\gamma_k, \delta_k]$ and $j \in [\alpha_k, \beta_k]$ for $k = \lceil j/2^{\log n-s}\rceil$.
\end{proof}


\bibliographystyle{abbrv}	
\bibliography{knapsack}

\begin{thebibliography}{10}

\bibitem{backurs2017better}
A.~Backurs, P.~Indyk, and L.~Schmidt.
\newblock Better approximations for tree sparsity in nearly-linear time.
\newblock In {\em SODA}, pages 2215--2229, 2017.

\bibitem{Bellman:1957}
R.~Bellman.
\newblock {\em Dynamic Programming}.
\newblock Princeton University Press, Princeton, NJ, USA, first edition, 1957.

\bibitem{BCDEHILPT12}
D.~Bremner, T.~M. Chan, E.~D. Demaine, J.~Erickson, F.~Hurtado, J.~Iacono,
  S.~Langerman, M.~Patrascu, and P.~Taslakian.
\newblock Necklaces, convolutions, and {X+Y}.
\newblock In {\em ESA}, pages 160--171, 2006.

\bibitem{bringmann2017near}
K.~Bringmann.
\newblock A near-linear pseudopolynomial time algorithm for subset sum.
\newblock In {\em SODA}, pages 1073--1084, 2017.

\bibitem{BGSW16}
K.~Bringmann, F.~Grandoni, B.~Saha, and V.~V. Williams.
\newblock Truly sub-cubic algorithms for language edit distance and
  {RNA}-folding via fast bounded-difference min-plus product.
\newblock In {\em FOCS}, pages 375--384, 2016.

\bibitem{Bussieck:1994}
M.~Bussieck, H.~Hassler, G.~J. Woeginger, and U.~T. Zimmermann.
\newblock Fast algorithms for the maximum convolution problem.
\newblock {\em Oper. Res. Lett.}, 15(3):133--141, Apr. 1994.

\bibitem{CGIMPS16}
M.~L. Carmosino, J.~Gao, R.~Impagliazzo, I.~Mihajlin, R.~Paturi, and
  S.~Schneider.
\newblock Nondeterministic extensions of the strong exponential time hypothesis
  and consequences for non-reducibility.
\newblock In {\em ITCS}, pages 261--270, 2016.

\bibitem{chan2015clustered}
T.~M. Chan and M.~Lewenstein.
\newblock Clustered integer {3SUM} via additive combinatorics.
\newblock In {\em STOC}, pages 31--40, 2015.

\bibitem{Chvatal80}
V.~Chvatal.
\newblock Hard knapsack problems.
\newblock {\em Operations Research}, 28:1402--1411, 1980.

\bibitem{thomas2001introduction}
T.~H. Cormen, C.~E. Leiserson, R.~L. Rivest, and C.~Stein.
\newblock {\em Introduction to algorithms}.
\newblock MIT press Cambridge, 2001.

\bibitem{cygan2017problems}
M.~Cygan, M.~Mucha, K.~Wegrzycki, and M.~Wlodarczyk.
\newblock On problems equivalent to $(\min, +)$-convolution.
\newblock In {\em ICALP}, pages 22:1--22:15, 2017.

\bibitem{GJ90:book}
M.~R. Garey and D.~S. Johnson.
\newblock {\em Computers and Intractability; A Guide to the Theory of
  NP-Completeness}.
\newblock W. H. Freeman \& Co., New York, NY, USA, 1990.

\bibitem{HS74}
E.~Horowitz and S.~Sahni.
\newblock Computing partitions with applications to the knapsack problem.
\newblock {\em J. ACM}, 21(2):277--292, April 1974.

\bibitem{koiliaris2017faster}
K.~Koiliaris and C.~Xu.
\newblock A faster pseudopolynomial time algorithm for subset sum.
\newblock In {\em SODA}, pages 1062--1072, 2017.

\bibitem{KPS17}
M.~K{\"{u}}nnemann, R.~Paturi, and S.~Schneider.
\newblock On the fine-grained complexity of one-dimensional dynamic
  programming.
\newblock In {\em ICALP}, pages 21:1--21:15, 2017.

\bibitem{maragos:book}
P.~Maragos.
\newblock Differential morphology.
\newblock In S.~Mitra and G.~Sicuranza, editors, {\em Nonlinear Image
  Processing}, chapter~10, pages 289--329. Academic Press, 2000.

\bibitem{MT90:book}
S.~Martello and P.~Toth.
\newblock {\em Knapsack Problems: Algorithms and Computer Implementations}.
\newblock John Wiley \& Sons, Inc., New York, NY, USA, 1990.

\bibitem{Pisinger}
D.~Pisinger.
\newblock Linear time algorithms for knapsack problems with bounded weights.
\newblock {\em Journal of Algorithms}, 33:1--14, 1999.

\bibitem{sleator1983data}
D.~D. Sleator and R.~E. Tarjan.
\newblock A data structure for dynamic trees.
\newblock In {\em STOC}, pages 114--122, 1981.

\bibitem{tamir2009new}
A.~Tamir.
\newblock New pseudopolynomial complexity bounds for the bounded and other
  integer knapsack related problems.
\newblock {\em Operations Research Letters}, 37(5):303--306, 2009.

\bibitem{uspensky1937introduction}
J.~V. Uspensky.
\newblock {\em Introduction to mathematical probability}.
\newblock McGraw-Hill, 1937.

\bibitem{west2001introduction}
D.~B. West.
\newblock {\em Introduction to graph theory}.
\newblock Prentice Hall, second edition, 2001.

\bibitem{zwick1998all}
U.~Zwick.
\newblock All pairs shortest paths in weighted directed graphs-exact and almost
  exact algorithms.
\newblock In {\em FOCS}, pages 310--319, 1998.

\bibitem{zwick2002all}
U.~Zwick.
\newblock All pairs shortest paths using bridging sets and rectangular matrix
  multiplication.
\newblock {\em JACM}, 49(3):289--317, 2002.

\end{thebibliography}

\appendix
 \section{The Knapsack Problem}\label{sec:knapsack}
In this section, we consider the knapsack problem and present a fast algorithm that can solve this problem for small values. In particular, when the maximum value of the items is constant, our algorithm runs in linear time. In this problem, we have a knapsack of size $t$ and $n$ items each associated to a size $s_i$ and value $v_i$. The goal is to place a subset of the items into the knapsack with maximum total value subject to their total size being limited by $t$. In the 0/1 knapsack problem, we are allowed to use each item at most once whereas in the unbounded knapsack problem, we can use each item several times. From now on, every time we use the term knapsack problem we mean the 0/1 knapsack problem unless stated otherwise.

A classic dynamic programming algorithm yields a running time of $O(nt)$~\cite{thomas2001introduction} for the knapsack problem. On the negative side, recently it has been shown that both the 0/1 and unbounded knapsack problems are as hard as $(\max,+)$ convolution and thus it is unlikely to solve any of these problems in time $O((n+t)^{2-\epsilon})$ for any $\epsilon > 0$~\cite{cygan2017problems}. However, there is no assumption on the values of the items in these reductions and thus the hardness results don't carry over to the case of small values. In particular, a barely subquadratic time ($O(t^{1.859}+n)$) algorithm follows from the work of~\cite{chan2015clustered} when the item values are constant integer numbers. In what follows, we show that we can indeed solve the problem in truly subquadratic time when the input values are small.  We assume throughout this section that the values of the items are integers in range $[0, \vmax ]$. Using the prediction technique we present an $\tildorder(\vmax t+n)$ time algorithm for the knapsack problem.

We begin by defining a knapsack variant of the $(\max,+)$ convolution in Section~\ref{sec:knapsack:convolution} and show that if the corresponding knapsack problems have non-negative integer values bounded by $\vmax $, then one can compute the $(\max, +)$ convolution of two vectors in time $\tildorder(\vmax n)$. It follows from the recent technique of~\cite{cygan2017problems}  that using this type of $(\max, +)$ convolution, one can solve the knapsack problem in time $\tildorder(\vmax n)$. However, for the sake of completeness, we include a formal proof in Appendix~\ref{appendix:knapsack}.

\subsection{Knapsack Convolution}\label{sec:knapsack:convolution}
Let $a$ and $b$ be two vectors that correspond to the solutions of two knapsack instances $\k_a$ and $\k_b$. More precisely, $a_i$ is the maximum value of the items in knapsack problem $\k_a$ with a total size of at most $i$. Similarly, $b_i$ is the maximum value of the items in knapsack problem $\k_b$ with a total size of at most $i$. We show that if the values of the items in $\k_a$ and $\k_b$ are non-negative integers bounded by $\vmax $, then one can compute $a \ttimes b$ in time $\tildorder(\vmax (|a| + |b|) + n)$ where $n$ is the total number of items in $\k_a$ and $\k_b$. 

The sketch of the algorithm is as follows: We first define the fractional variants of both the knapsack problem and the knapsack convolution. We show that both problems can be efficiently solved in time $O(n \log n)$ where $n$ is the total number of items in each knapsack problem. Next, we observe that any solution of the fractional knapsack problem can be turned into a solution for the knapsack problem with an additive error of at most $\vmax $. Similarly, any solution for the fractional knapsack convolution is always at most $2\vmax $ away from the solution of the knapsack convolution. We then show that both the solution of the fractional knapsack problem and the solution of fractional knapsack convolution have some properties. We explore these properties and show that they enable us to find an uncertain solution for the knapsack convolution in time $\tildorder(t+n)$. This yields an $\tildorder(\vmax n)$ time solution for knapsack convolution via Theorem~\ref{theorem:prediction}. In the interest of space, we omit some of the proofs of this section and include them in Appendix~\ref{appendix:knapsack:convolution}.

We define the fractional variant of the knapsack problem as follows. In the fractional knapsack problem, we are also allowed to divide the items into smaller pieces and the value of each piece is proportional to the size of that piece. More formally, the fractional knapsack problem is defined as follows:

\begin{definition}
    Given a knapsack of size $t$ and $n$ items with sizes $s_1, s_2, \ldots, s_n$ and values $v_1, v_2, \ldots, v_n$, the fractional knapsack problem is to find $n$ non-negative real values $f_1, f_2, \ldots, f_n$ such that $\sum s_i f_i \leq t$, $f_i \leq 1$ for all $i$, and $\sum f_i v_i$ is maximized.
\end{definition}

One well-known observation is that the fractional variant of the knapsack problem
can be solved exactly via a greedy algorithm: sort the items according
to the ratio of value over size and put the items into the knapsack
accordingly. If at some point there is not enough space for the next
item, we break it into a smaller piece that completely fills the
knapsack. We stop when either we run out of items or the knapsack is
full. We call this the greedy knapsack algorithm and name this simple observation.

    
    
    

\begin{observation}\label{observation:knapsack1}
     The greedy algorithm solves the fractional knapsack problem in time $O(n \log n)$.
\end{observation}

It is easy to see that in any solution of
the greedy algorithm for knapsack, there is at most one item in the
knapsack which is broken into a smaller piece. Therefore, one can turn
any solution of the fractional knapsack problem into a solution of the
knapsack problem by removing the only item from the knapsack with $f_i
< 1$ (if any). If all the values are bounded by $\vmax $, this hurts
the solution by at most an additive factor of $\vmax $. Moreover, the
solution of the fractional knapsack problem is always no less than the
solution of the integral knapsack problem. Thus, any solution for the
fractional knapsack problem can be turned into a solution for the
knapsack problem with an additive error of at most $\vmax $.

Based on a similar idea, we define the fractional knapsack convolution of two vectors as follows:
\begin{definition}
    Let $a$ and $b$ be two vectors corresponding to two knapsack problems $\k_a$ and $\k_b$ with knapsack sizes $t_a$ and $t_b$. For a real value $t$, we define the fractional knapsack convolution of $a$ and $b$ with respect to $t$ as the solution of the knapsack problem with knapsack size $t$ and the union of items  of $\k_a$ and $\k_b$ subject to the following two additional constraints:
    \begin{itemize}
        \item The total size of the items of $\k_a$ in the solution is bounded by $t_a$.
        \item The total size of the items of $\k_b$ in the solution is bounded by $t_b$.
    \end{itemize}
\end{definition}
One can modify the greedy algorithm to compute the solution of the fractional knapsack convolution as well. The only difference is that once the total size of the items of either knapsack instances in the solution reaches the size of that knapsack we ignore the rest of the items from that knapsack. A similar argument to what we stated for Observation~\ref{observation:knapsack1} proves the correctness of this algorithm.

\begin{algorithm}
    \KwData{$t$, $a$, $b$, and two knapsack instances $\k_a$ and $\k_b$ corresponding to $a$ and $b$.}
    \KwResult{The solution of fractional knapsack convolution of $a$ and $b$ with respect to $t$}
    
    Let \textsf{items} be a sequence of size $n$ containing all items of $\k_a$ and $\k_b$\;
    Sort the items of \textsf{items} according to $v_i / s_i$ in non-increasing order.
    
    $Answer \leftarrow 0$\;
    \For{$i \in [1,n]$}{
        \If {$(s_i, v_i)$ belongs to $\k_a$}{
            $cut \leftarrow \min\{s_i, t, t_a\}$\;
            $Answer \leftarrow Answer + v_i cut/s_i$\;
            $t \leftarrow t - cut$\;
            $t_a \leftarrow t_a - cut$\;
        }\Else{
            $cut \leftarrow \min\{s_i, t, t_b\}$\;
            $Answer \leftarrow Answer + v_i cut/s_i$\;
            $t \leftarrow t - cut$\;
            $t_b \leftarrow t_b - cut$\;
        }
    }
    
    \textbf{Return } $Answer$\;        
    \caption{\textsf{GreedyAlgorithmForFractionalKnapsackConvolution}($a,b,\k_a,\k_b$)}\label{alg:greedyalgorithmforfractionalknapsackconvolution}
\end{algorithm}

\begin{observation}\label{observation:knapsack2}
    Algorithm~\ref{alg:greedyalgorithmforfractionalknapsackconvolution} solves the fractional knapsack convolution in time $O(n \log n)$.
\end{observation}
Again, one can observe that in any solution of the greedy algorithm for fractional knapsack convolution, there are at most two items that are fractionally included in the solution (at most one for each knapsack instance). Thus, we can get a solution with an additive error of at most $2\vmax $ for the knapsack convolution problem from the solution of the fractional knapsack convolution.

We explore several properties of the fractional solutions for the knapsack problems and the knapsack convolution and based on them, we present an algorithm to compute an uncertain solution for the knapsack convolution within an error of $O(\vmax )$.
Define $a':[0,t_a] \rightarrow \mathbb{R}$ and $b':[0,t_b] \rightarrow \mathbb{R}$  as the solutions of the fractional knapsack problems for $\k_a$ and $\k_b$, respectively. Therefore, for any real value $x$ in the domain of the functions, $a'(x)$ and $b'(x)$ denote the solution of each fractional knapsack problem for knapsack size $x$. Moreover, we define a function $c:[0,t_a+t_b]\rightarrow \mathbb{R}$ where $c'(x)$ is the solution of the fractional knapsack convolution of $a$ and $b$ with respect to $x$. Note that for all $a'$, $b'$, and $c'$, parameter $x$ may be a real value. The following observations follow from the greedy solutions for $a'$, $b'$, and $c'$.

\begin{observation}\label{observation:fas}
    There exist non-decreasing functions $\fa:[0,t_a+t_b]\rightarrow [0,t_a]$ and $\fb:[0,t_a+t_b]\rightarrow [0,t_b]$ such that $c'(x) = a'(\fa(x)) + b'(\fb(x))$.
\end{observation}

Since in Algorithm~\ref{alg:greedyalgorithmforfractionalknapsackconvolution} we put the items greedily in the knapsack, for every $0 \leq x \leq t_a$, there exists a $0 \leq y \leq t_a + t_b$ such that $\fa(y) = x$. Similarly, for every $0 \leq x \leq t_b$, there exists a $0 \leq y \leq t_a+t_b$ such that $\fb(y) = x$. We define $\fa^{-1}(x)$ as the smallest $y$ such that $\fa(y) = x$. Moreover, $\fb^{-1}(x)$ is equal to the smallest $y$ such that $\fb(y) = x$.

\begin{observation}\label{observation:inc}
    For an $0 \leq x \leq t_a$, $y$, and $y'$ such that $0 \leq y < y'  \leq \fa^{-1}(x)-x$ we have $c'(x+y) - a'(x)+b'(y) \geq c'(x+y') - a'(x) - b'(y')$.
\end{observation}

\begin{observation}\label{observation:dec}
    For an $0 \leq x \leq t_a$, $y$, and $y'$ such that $\fa^{-1}(x)-x \leq y < y'  \leq t_b$ we have $c'(x+y) - a'(x)+b'(y) \leq c'(x+y') - a'(x) - b'(y')$.
\end{observation}

Observations~\ref{observation:inc} and~\ref{observation:dec} show that for any $0 \leq x \leq t_a$, if we define $g_x(y) = c'(x+y) - a'(x) - b'(y)$ then $g_x$ is non-decreasing in the interval $[0, \fa^{-1}(x)-x]$ and non-increasing in the interval $[\fa^{-1}(x)-x, t_b]$. Now, for every integer $i \in [0, t_a]$ define $\alpha'_i$ to be the smallest number in range $[0, \fa^{-1}(x)-x]$ such that $g_i(\alpha'_i) \leq 2\vmax $. Similarly, define $\beta'_i$ to be the largest number in range $[\fa^{-1}(x)-x, t_b]$ such that $g_i(\beta'_i) \leq 2\vmax $. It follows from Observations~\ref{observation:inc} and~\ref{observation:dec} that $g_i(x) \leq 2\vmax $ holds in the interval $[\alpha'_i, \beta'_i]$ and $g_i(x) > 2\vmax $ holds for any $x$ outside this range. Moreover, Observations~\ref{observation:mon1} and~\ref{observation:mon2} imply that $[\alpha'_i, \beta'_i]$'s are monotonic.

\begin{observation}\label{observation:mon1}
    Let $x, x'$, and $y$ be three real values such that $0 \leq x < x' \leq t_a$ and $0 \leq y \leq \fa^{-1}(x)-x$. Then $c'(x+y) - a'(x) - b'(y) \leq c'(x'+y) - a'(x') - b'(y)$. 
\end{observation}

\begin{observation}\label{observation:mon2}
    Let $x, x'$, and $y$ be three real values such that $0 \leq x < x' \leq t_a$ and $0 \leq \fa^{-1}(x')-x' \leq y$. Then $c'(x+y) - a'(x) - b'(y) \geq c'(x'+y) - a'(x') - b'(y)$. 
\end{observation}
Notice that for every pair of integers $i$ and $j$ such that $\alpha'_i \leq j \leq \beta'_i$ we have $c'(i+j) - a'(i) - b'(j) \leq 2\vmax $. Recall that $a'$ and $b'$ are the solutions of the fractional knapsack problems and thus $a'(i) - a_i$ and $b'(j) - b_j$ are bounded by $\vmax $. Moreover, since $c'(i+j)$ is always at least as large as $c_{i+j}$, we have $c_{i+j} - a_i - b_j \leq 4\vmax $ for all $\alpha'_i \leq j \leq \beta'_i$. Furthermore, for every integer $j \in [0,t_b] \setminus [\alpha'_i, \beta'_i]$ we have $c'(i+j) - a'(i) - b'(j) > 2\vmax $. Similarly, one can argue that $c'(i+j) \leq c_{i+j}+2\vmax $, $a'(i) \geq a_i$, and $b'(j) \geq b_j$ and thus $c_{i+j} - a_i - b_j > 0$ which means that intervals $[\alpha'_i, \beta'_i]$ make an uncertain solution for $a \ttimes b$ within an error of $4\vmax $. To make the intervals integer, we set $\alpha_i = \lceil \alpha'_i \rceil$ and $\beta_i = \lfloor \beta'_i \rfloor$. Since $[\alpha_i, \beta_i]$ is also an uncertain solution within an error of $4\vmax $ we can compute $a \ttimes b$ in time $\tildorder(\vmax (|a| + |b|) + n)$.

\begin{theorem}\label{theorem:knapsackconvolution}
    Let $\k_a$ and $\k_b$ be two knapsack problems with knapsack sizes $t_a$ and $t_b$  and $n$ items in total. Moreover, let the item values in $\k_a$ and $\k_b$ be integer values bounded by $\vmax $ and $a$ and $b$ be the solutions of these knapsack problems. There exists an $\tildorder(\vmax (t_a + t_b)+n)$ time algorithm for computing $a \ttimes b$ using $a$, $b$, $\k_a$, and $\k_b$. 
\end{theorem}
\begin{proof}
        Let $t = t_a + t_b$ be the largest index of $a \ttimes b$. As shown earlier, intervals $[\alpha_i, \beta_i]$ formulated above make an uncertain solution for $a \ttimes b$ within an error of $4\vmax $. Thus, it only suffices to compute these intervals and then using Theorem~\ref{theorem:prediction} we can compute $a \ttimes b$ in time $\tildorder(\vmax (|a|+|b|))$. In order to determine the intervals, we first compute three arrays $a'$, $b'$, and $c'$ with ranges $[0, t_a]$, $[0, t_b]$, and $[0, t_a + t_b]$, respectively. Then for every $i$ in range $[0, t_a]$ we compute $a'_i$ to be the solution to the fractional knapsack problem of $\k_a$ with knapsack size $i$. This can be done in time $O(n \log n + t)$, since we can use the greedy algorithm to determine these values. Similarly, we compute $b'_i$ equal to the solution to the fractional knapsack problem for $\k_b$ and $c_i$ equal to the solution of the fractional knapsack convolution for $a \ttimes b$. This step of the algorithm takes a total time of $O(n \log n + t) = \tildorder(n + t)$.
    
    Along with the construction of array $c'$, we also compute two arrays $\fa$ and $\fb$ in time $O(t)$ where $c'_i = a'_{{\fa}_i} + b'_{{\fb}_i}$. More precisely, every time we compute $c'_i$ for some integer $i$ we also keep track of the total size of the solution corresponding to each knapsack and store these values in arrays $\fa$ and $\fb$. Also one can compute an array $\fa^{-1}$ from $\fa$ in time $O(t)$. Next, we iterate over all integers $i$ in range $[0,t_a]$ and for each $i$ compute $\alpha_i$ and $\beta_i$ in time $O(\log n)$. Recall that $\alpha_i$ ($\beta_i$) is the  smallest (largest) integer $j$ in range $[0, {\fa^{-1}}_i - i]$ ($[{\fa^{-1}}_i - i, t_b]$) such that $c'_{i+j} - a'_i - b'_{j} \leq 2\vmax $. Moreover, $c'_{i+j} - a'_i - b'_{j}$ is monotonic in both ranges $[0, {\fa^{-1}}_i - i]$ and $[{\fa^{-1}}_i - i, t_b]$ and hence $\alpha_i$ and $\beta_i$ can be computed in time $O(\log t_b)$ for every $i$. This makes a total running time of $O(t_a \log t_b) = \tildorder(t)$. Finally, since intervals $[\alpha_i, \beta_i]$ make an uncertain solution for $a \ttimes b$ within an error of $4\vmax $, we can compute $a \ttimes b$ in time $\tildorder(\vmax (t_a + t_b)+n)$ (Theorem~\ref{theorem:prediction}).
\end{proof}

Based on the proof of Theorem \ref{theorem:knapsackconvolution} one can imply that even if the knapsack solutions are subject to an additive error of $\vmax$, the convolution can be still computed in time $\tildorder(\vmax (t_a + t_b)+n)$. In Appendix~\ref{appendix:knapsack} we show that Theorem~\ref{theorem:knapsackconvolution} yields a solution for the 0/1 knapsack problem in time $\tildorder(\vmax t + n)$. The algorithm follows from the reduction of~\cite{cygan2017problems} from 0/1 knapsack to knapsack convolution.

\begin{theorem}[a corollary of Theorem~\ref{theorem:knapsackconvolution} and the reduction of~\cite{cygan2017problems} from 0/1 knapsack to $(\max,+)$ convolution]\label{theorem:knapsack}
    The 0/1 knapsack problem can be solved in time $O(\vmax t+n)$ when the item values are integer numbers in range $[0,\vmax ]$.
\end{theorem}
\newpage
\section{Computing $a^{\ttimes k}$ and Application to Unbounded Knapsack}\label{sec:power}
Throughout this section, any time we mention $a^{\ttimes k}$ we mean $\overbrace{ a \ttimes a \ttimes \ldots \ttimes a }^{k \text{ times}}$. In this section, we present another application of the prediction technique for computing the $k$'th power of a vector in the $(\max,+)$ setting. The classic algorithm for this problem runs in time $\tildorder(n^2)$ and thus far, there has not been any substantial improvement for this problem. We consider the case where the input values are integers in range $[0,\emax ]$, nonetheless, this result carries over to any range of integer numbers within an interval of size $\emax $.\footnote{It only suffices to add a constant $C$ to every element of the vector to move the numbers to the interval $[0,\emax ]$. After computing the solution, we may move the solution back to the original space.} Using known FFT-based techniques, one can compute $a \ttimes a$ in time $\tildorder(\emax |a|)$ (see Section~\ref{sec:maxplusconvolution:reduction}). However, the values of the elements of $a^{\ttimes 2}$ no longer lie in range $[0,\emax ]$ and thus computing $a^{\ttimes 2} \ttimes a^{\ttimes 2}$ requires more computation than $a \ttimes a$. In particular, the values of the elements of $a^{\ttimes k/2}$ are in range $[0, \emax k/2]$ and thus computing $a^{\ttimes k/2} \ttimes a^{\ttimes k/2}$  via the known techniques requires a running time of $\tildorder(\emax k|a^{\ttimes k}|)$. The main result of this section is an algorithm for computing $a^{\ttimes k}$ in time $\tildorder(\emax |a^{\ttimes k}|)$. Moreover, we show that any prefix of size $n$ of $a^{\ttimes k}$ can be similarly computed  in time $\tildorder(\emax n)$. We later make a connection between this problem and the unbounded knapsack problem and show this results in an $\tildorder(\emax t+n)$ time algorithm for the unbounded knapsack problem when the item values are integers in range $[0,\emax ]$. Our algorithm is based on the prediction technique explained in Section~\ref{sec:maxplusconvolution}.

We first explore some observations about the powers of a vector in the $(\max,+)$ setting. We begin by showing that if $b = a^{\ttimes k}$ for some positive integer $i$, then the elements of $b$ are (weakly) monotone.

\begin{lemma}\label{lemma:power:1}
	Let $a$ be a vector whose  values are in range $[0,\emax ]$ and $b = a^{\ttimes k}$ for a positive integer $k$. Then, for $0 \leq i < j < |a^{\ttimes k}|$ we have $b_j \geq b_i - \emax $.
\end{lemma}
\begin{proof}
	If $k = 1$, the lemma follows from the fact that all values of the elements of $a$ are in range $[0,\emax ]$. For $k > 1$, we define $c = a^{\ttimes k-1}$ and let $l$ be an index of $c$ with the maximum $c_i$ subject to $l \leq j$. Since $b = c \ttimes a$ we have $b_i  = c_{i'} + a_{i-i'}$ for some $i'$. Note that $c_{i'} \leq c_l$ and also all values of the indices of $c$ are bounded by $\emax $. Thus we have $b_i \leq c_l + \emax $. In addition to this, since $l \leq j$ we have $b_j \geq c_l + a_{j-i}$. Notice that all values of the indices of $a$ are non-negative and therefore $b_j \geq c_l$. This along with the fact that $b_i \leq c_l + \emax $ implies that $b_j \geq b_i - \emax $.
\end{proof}

Another observation that we make is that if for a $k$ and a $k'$ we have $|k-k'| \leq 1$, then $a^{\ttimes k} \ttimes a^{\ttimes k'}$ can be computed by just considering a few $(i,j)$ pairs of the vectors with close values.

\begin{lemma}\label{lemma:power:2}
	Let $k$ and $k'$ be two positive integer exponents such that $|k-k'| \leq 1$. Moreover, let $a$ be an integer vector whose  elements' values lie in range $[0,\emax ]$. Then, for every $0 \leq i \leq |a^{\ttimes k}|$, there exist two indices $j$ and $i-j$ such that (i) $a^{\ttimes k+k'}_i = a^{k}_j + a^{\ttimes k'}_{i-j}$ and (ii) $|a^{k}_j - a^{\ttimes k'}_{i-j}| \leq \emax $. 
\end{lemma}
\begin{proof}
	By definition $a^{\ttimes k+k'} = \overbrace{ a \ttimes a \ttimes \ldots \ttimes a }^{k+k' \text{ times}}$. Therefore, for every $0 \leq i < |a^{\ttimes k+k'}|$, there exist $k+k'$ indices $i_1, i_2, \ldots, i_{k+k'}$ such that $a^{\ttimes k+k'}_i = a_{i_1} + a_{i_2} + \ldots + a_{i_{k+k'}}$ and $i_1 + i_2 + \ldots + i_{k+k'} = i$. We assume w.l.o.g. that $a_{i_1} \leq a_{i_2} \leq \ldots \leq a_{i_{k+k'}}$. We separate the odd and even indices of $i$ to form two sequences $i_1, i_3, \ldots$ and $i_2, i_4, \ldots$. Notice that since $|k-k'| \leq 1$, the size of one of such sequences is $k$ and the size of the other one is $k'$. We assume w.l.o.g. that the size of the odd sequence is $k$ and the size of the even sequence is $k'$. We now define $j = i_1 + i_3 + \ldots$ and $j' = i_2 + i_4 + \ldots$. Since $i_1 + i_2 + \ldots = i$ then $j' = i-j$ holds. Since $a^{\ttimes k+k'}_i = a_{i_1} + a_{i_2} + \ldots + a_{i_{k+k'}}$ we also have $a^{\ttimes k}_j = a_{i_1} + a_{i_3} + \ldots$, $a^{\ttimes k'}_{j'} = a_{i_2} + a_{i_4} + \ldots$, and also $a^{\ttimes k+k'}_i = a^{\ttimes k}_j + a^{\ttimes k'}_{j'}$. To complete the proof, it only suffices to show that $|a^{\ttimes k}_j - a^{\ttimes k'}_{j'}| \leq \emax $. This follows from the fact that the value of all indices of $a$ are in range $[0,\emax ]$ and that $a_{i_1} \leq a_{i_2} \leq a_{i_3} \leq \ldots \leq a_{i_{k+k'}}$.
\end{proof}

	What Lemma~\ref{lemma:power:2} implies is that when computing $a^{\ttimes k} = a^{\ttimes \lceil k/2 \rceil} \ttimes a^{\ttimes \lfloor k/2 \rfloor}$, it only suffices to take into account $(i,j)$ pairs such that $|a^{\ttimes \lceil k/2 \rceil}_i - a^{\ttimes \lfloor k/2 \rfloor}_j| \leq \emax $. This observation enables us to compute $a^{\ttimes k} = a^{\ttimes \lceil k/2 \rceil} \ttimes a^{\ttimes \lfloor k/2 \rfloor}$ in time $\tildorder(\emax |a^{\ttimes k}|)$ via the prediction technique. Suppose $a$ is an integer vector with values in range $[0,\emax ]$. In addition to this, assume that $\hat{a} = a^{\ttimes \lceil k/2 \rceil}$ and $\bar{a} = a^{\ttimes \lfloor k/2 \rfloor}$. We propose an algorithm that receives $\hat{a}$ and $\bar{a}$ as input and computes $a^{\ttimes k} = \hat{a} \ttimes \bar{a}$ as output. The running time of our algorithm is $\tildorder(\emax |a^{\ttimes k}|)$.

We define two integer vectors $\hat{b}$ and $\bar{b}$ where $\hat{b}_i = \max_{j \leq i}\hat{a}_j$. Similarly, $\bar{b}_i = \max_{j \leq i}\bar{a}_j$. By definition, both vectors $\hat{b}$ and $\bar{b}$ are non-decreasing. Now, for every index $i$ of $\hat{b}$ we find an interval $[x_i, y_i]$ of $\bar{b}$ such that $\hat{b}_i - 2\emax  \leq \bar{b}_j \leq \hat{b}_i + 2\emax $ for any $j$ within $[x_i, y_i]$. Since both vectors $\hat{b}$ and $\bar{b}$ are non-decreasing, computing each interval takes time $O(\log n)$ via binary search. Finally, we provide these intervals to the prediction technique and compute $a^{\ttimes k} = \hat{a} \ttimes \bar{a}$ in time $\tildorder(\emax |a^{\ttimes k}|)$. In Lemma~\ref{lemma:power:main}, we prove that the intervals adhere to the conditions of the prediction technique and thus Algorithm~\ref{algorithm:power:lemma} correctly computes $a^{\ttimes k}$ from $\hat{a}$ and $\bar{a}$ in time $\tildorder(\emax |a^{\ttimes k}|)$.

\begin{algorithm}[H]
	\label{algorithm:power:lemma}
	\KwData{Two vectors $\hat{a}$ and $\bar{a}$ s.t. $\hat{a} = a^{\ttimes \lceil k/2 \rceil}$ and $\bar{a} = a^{\ttimes \lfloor k/2 \rfloor}$ for some $a$ and $k$.}
	\KwResult{$\hat{a} \ttimes \bar{a}$}
	
	Let $\hat{b}, \bar{b}$ be two vectors of size $|a^{\ttimes \lceil k/2 \rceil}|$ and $|a^{\ttimes \lfloor k/2 \rfloor}|$ respectively.\;
	$\hat{b}_0 \leftarrow \hat{a}_1$\;
	$\bar{b}_0 \leftarrow \bar{a}_1$\;
	\For{$i \in [1,|\hat{a}|-1]$}{
		$\hat{b}_i \leftarrow \max\{\hat{b}_{i-1},\hat{a}_i\}$\;
	}
	\For{$i \in [1,|\bar{a}|-1]$}{
		$\bar{b}_i \leftarrow \max\{\bar{b}_{i-1},\bar{a}_i\}$\;
	}
	
	\For{$i \in [1,|a^{\ttimes k}|-1]$}{
		$x_i \leftarrow $ the smallest $j$ such that $\bar{b}_j \geq \bar{b}_i - 2\emax $\;
		$y_i \leftarrow $ the largest $j$ such that $\bar{b}_j \leq \bar{b}_i + 2\emax $\;
	}
	
	$c = \textsf{PolynomialMultiplicationViaPredictionMethod}(\hat{a},\bar{a},5\emax ,x_i\text{'s}, y_i{'s})$\;
	\textbf{Return } $c$\;        
	\caption{\textsf{FastPower}($\hat{a},  \bar{a}, \emax $)}
\end{algorithm}

\begin{lemma}\label{lemma:power:main}
	Let $a$ be an integer vector with values in range $[0,\emax ]$. For some integer $k > 0$, let $\hat{a} = a^{\ttimes \lceil k/2 \rceil}$ and $\bar{a} = a^{\ttimes \lfloor k/2 \rfloor}$. Given $\hat{a}$ and $\bar{a}$ as input, Algorithm~\ref{algorithm:power:lemma} computes $a^{\ttimes k} = \hat{a} \ttimes \bar{a}$ in time $\tildorder(\emax |a^{\ttimes \lceil k \rceil}|)$.
\end{lemma}
\begin{proof}
	The correctness of Algorithm~\ref{algorithm:power:lemma} boils down to whether  intervals $[x_i, y_i]$ provided for the prediction technique meet the conditions of Theorem~\ref{theorem:prediction}. Before we prove that the conditions are met, we note that by Lemma~\ref{lemma:power:1}, the values of $\hat{b}$ and $\bar{b}$ are at most $\emax $ more than that of $\hat{a}$ and $\bar{a}$. Moreover, by definition, the vectors $\hat{b}$ and $\bar{b}$ are non-decreasing and lower bounded by the values of $\hat{a}$ and $\bar{a}$. 
	
	\textbf{First condition:} The first condition is that for every $0 \leq i < |a^{\ttimes \lceil k/2 \rceil}|$ and $x_i \leq j \leq y_i$ we have $\hat{a}_i + \bar{a}_j \geq (\hat{a} \ttimes \bar{a})_{i+j} - O(\emax )$. In what follows, we show that in fact $\hat{a}_i + \bar{a}_j \geq (\hat{a} \ttimes \bar{a})_{i+j} - 5\emax $ holds for such $i$'s and $j$'s. Due to Lemma~\ref{lemma:power:2}, for every such $i$ and $j$, there exist an $i'$ and a $j'$ such that $\hat{a}_{i'} + \bar{a}_{j'} = (\hat{a} \ttimes \bar{a})_{i'+j'}$, $i'+j' = i+j$, and $|\hat{a}_{i'} - \bar{a}_{j'}| \leq \emax $. Therefore, 
	\begin{equation*}
	\begin{split}
	(\hat{a} \ttimes \bar{a})_{i+j} &=(\hat{a} \ttimes \bar{a})_{i'+j'}\\
	& = \hat{a}_{i'} + \bar{a}_{j'} \\
	& \leq 2\min\{\hat{a}_{i'} + \bar{a}_{j'}\} + \emax \\
	& \leq 2\min\{\hat{b}_{i'} + \bar{b}_{j'}\} + \emax .
	\end{split}
	\end{equation*}
	In addition to this, we know that $i+j = i'+j'$ and thus either $i' \leq i$ or $j' \leq j$. In any case, since both $\hat{b}$ and $\bar{b}$ are non-decreasing, 
	$\max\{\hat{b}_i, \bar{b}_j\} \geq \min\{\hat{b}_{i'},\hat{b}_{j'}\}$ 
	and therefore, 
	\begin{equation*}
	\begin{split}
	(\hat{a} \ttimes \bar{a})_{i+j} &\leq 2\min\{\hat{b}_{i'}, \bar{b}_{j'}\} + \emax  \\
	&\leq 2\max\{\hat{b}_i, \bar{b}_j\} + \emax .
	\end{split}
	\end{equation*}
	Due to Algorithm~\ref{algorithm:power:lemma}, $\max\{\hat{b}_i, \hat{b}_j\} - \min\{\hat{b}_i, \bar{b}_j\} \leq 2\emax $ and hence 
	\begin{equation*}
	\begin{split}
	(\hat{a} \ttimes \bar{a})_{i+j} &\leq 2\max\{\hat{b}_i, \bar{b}_j\} + \emax \\
	& \leq \hat{b}_i + \bar{b}_j + 3\emax \\
	& \leq \hat{a}_i + \bar{a}_j + 5\emax .
	\end{split}
	\end{equation*}
	
	\textbf{Second condition:} The second condition is that for every $0 \leq i < |a^{\ttimes \lceil k/2 \rceil}|$, there exists a $0 \leq j \leq i$ such that $\hat{a}_j + \bar{a}_{i-j} = (\hat{a} \ttimes \bar{a})_i $ and that $x_j \leq i-j \leq y_j$. We prove this condition via Lemma~\ref{lemma:power:2}. Lemma~\ref{lemma:power:2} states that for every $|a^{\ttimes \lceil k/2 \rceil}|$ there exists a $0 \leq j \leq i$ such that satisfies $\hat{a}_j + \bar{a}_{i-j} = (\hat{a} \ttimes \bar{a})_i$ and also $|\hat{a}_{j} - \bar{a}_{i-j}| \leq \emax $.
	Since the values of $\bar{b}, \hat{b}$ differ from $\hat{a}, \hat{b}$ by an additive factor of at most $\emax $, the latter inequality implies $|\hat{b}_{j} - \bar{b}_{i-j}| \leq 2\emax $. Due to Algorithm~\ref{algorithm:power:lemma}, if $|\hat{b}_{j} - \bar{b}_{i-j}| \leq 2\emax $ then $i-j$ lies in the interval $[x_j, y_j]$.
	
	\textbf{Third condition:} The third condition is regarding the monotonicity of $x_i$'s and $y_i$'s. This condition directly follows from the fact that both vectors $\hat{b}$ and $\bar{b}$ are non-decreasing and as such, the computed intervals are also non-decreasing.
	
	Apart from an invocation of Algorithm~\ref{alg:prediction}, the rest of the operations in Algorithm~\ref{algorithm:power:lemma} run in time $\tildorder(n)$ and therefore the total running time of Algorithm~\ref{algorithm:power:lemma} is $\tildorder(\emax |a^{\ttimes \lceil k \rceil}|)$.
\end{proof}

Based on Lemma~\ref{lemma:power:main}, for an integer vector with values in range $[0,\emax ]$, we can compute $a^{\ttimes k}$ via $O(\log k)$ $\ttimes$ operations, each of which takes time $\tildorder(\emax |a^{\ttimes \lceil k \rceil}|)$. Moreover, we always need to make at most $O(\log k) = \tildorder(1)$ $\ttimes$ operations in order to compute $a^{\ttimes k}$.

\begin{theorem}\label{theorem:power}
	Let $a$ be an integer vector with values in range $[0,\emax ]$. For any integer $k \geq 1$, one can compute $a^{\ttimes k}$ in time $\tildorder(\emax |a^{\ttimes \lceil k \rceil}|)$.
\end{theorem}
\begin{proof}
	The proof follows from the correctness of Algorithm~\ref{algorithm:power:lemma} and the fact that it runs in time $\tildorder(\emax |a^{\ttimes \lceil k \rceil}|)$.
\end{proof}

Theorem~\ref{theorem:power} provides a strong tool for solving many combinatorial problems including the unbounded knapsack problem. In order to compute the solution of the unbounded knapsack problem, it only suffices to construct a vector $a$ of size $t$ wherein $a_i$ specifies the value of the heaviest items with size $i$. $a$ itself specifies the solution of the unbounded knapsack problem if we are only allowed to put one item in the bag. Similarly, $a^{\ttimes 2}$ denotes the solution of the unbounded knapsack problem when we can put up to two items in the knapsack. More generally, for every $1 \leq k$, $a^{\ttimes k}$ denotes the solution of the unbounded knapsack problem subject to using at most $k$ items. This way, $a^{\ttimes t}$ formulates the solution of the unbounded knapsack problem. Note that in order to solve the knapsack problem, we only need to compute a prefix of size $t+1$ of $a^{\ttimes t}$. This makes the running time of every $\ttimes$ operation $\tildorder(\emax t)$ and thus computing the first $t+1$ elements of $a^{\ttimes t}$ takes time $\tildorder(\emax t)$.

\begin{theorem}[a corollary of Theorem~\ref{theorem:power}]\label{theorem:unboundedknapsack}
	The unbounded knapsack problem can be solved in time $\tildorder(\vmax t+n)$ when the item values are integers in range $[0,\vmax ]$.
\end{theorem}
\newpage
\section{Knapsack for Items with Small Sizes}\label{sec:sizelimit}
We also consider the case where the size of the items is bounded by $\smax $. Note that in such a scenario, the values of the items can be large real values, however, each item has an integer size in range $[1,\smax ]$. We propose a randomized algorithm that solves the knapsack problem w.h.p. in time $\tildorder(\smax (n+t))$ in this case. Our algorithm is as follows: we randomly put the items in $t/\smax $ different buckets. Using the classic quadratic time knapsack algorithm we solve the problem for each bucket up to a knapsack size $\tildorder(\smax )$. Next, we merge the solutions in $\log (t/\smax )$ rounds. In the first round, we merge the solutions for buckets $1$ and $2$, buckets $3$ and $4$, and so on. This results in $t/2\smax $ different solutions for every pair of buckets at the end of the first round. In the second round, we do the same except that this time the number of buckets is divided by $2$. After $\log (t/\smax )$ rounds, we only have a single solution and based on that, we determine the maximum value of the solution with a size bounded by $t$ and report that value.

If we use the classic $(\max,+)$-convolution for merging the solutions of two buckets, it takes time $O(t^2)$ for merging two solutions and yields a slow algorithm. The main idea to improve the running time of the algorithm is to merge the solutions via a faster algorithm. We explain the idea by stating a randomized argument. Throughout this paper, every time we use the term w.h.p. we mean with a probability of at least $1-n^{-10}$.

\begin{lemma}\label{lemma:sizelimit}
	Let $(s_1, v_1), (s_2, v_2), \ldots, (s_n,v_n)$ be $n$ items with sizes in range $[1,\smax ]$. Let the total size of the items be $S$. For some $0 < p < 1/2$, we randomly select each item of this set with probability $p$ and denote their total size by $S'$. If $\smax  \leq 2pS$ then for some $\zarib = \tildorder(1)$ $|pS - S'| \leq \zarib \sqrt{\smax pS}$ holds w.h.p. (with probability at  least $1-n^{-10}$). 
\end{lemma}
\begin{proof}
	This lemma follows from the Bernstein's inequality~\cite{uspensky1937introduction}. Bernstein's inequality states that if $x_1, x_2, \ldots, x_n$ are $n$ independent random variables strictly bounded by the intervals $[a_i, b_i]$ and $\bar{x} = \sum x_i$ then we have:
	$$\pr[|\bar{x} - \mathbb{E}[\bar{x}]| > y] \leq 2 \exp (-\frac{y^2/2}{V+Zy/3})$$
	where $V = \sum \mathbb{E}[(x_i - \mathbb{E}[x_i])^2]$ and $Z = \max\{b_i - a_i\}$.
	
	To prove the lemma, we use Bernstein's inequality in the following way: for every item we put a variable $x_i$ which identifies whether item $(s_i, v_i)$ is selected in our set. If so, we set $x_i = s_i$, otherwise we set $x_i = 0$. As such, the value of every variable $x_i$ is in range $[0, s_i]$ and thus $a_i = 0$ and $b_i = s_i$ for all $1 \leq i \leq n$. This way we have
	$$\mathbb{E}[\bar{x}] = \sum \mathbb{E}[x_i] = \sum p s_i = p (\sum s_i) = pS.$$
	Moreover, $s_i \leq \smax $ holds for all $i$ and for each $x_i$ we have $\mathbb{E}[(x_i - \mathbb{E}[x_i])^2] = p(1-p) s_i^2 \leq p(1-p) \smax s_i$. Thus, $$V = \sum \mathbb{E}[(x_i - \mathbb{E}[x_i])^2] \leq \sum p(1-p)\smax s_i = (\sum s_i) p(1-p)\smax  = p(1-p)\smax S.$$ By replacing $\mathbb{E}[\bar{x}]$ by $pS$, $Z$ by $\smax $, and $V$ by $p(1-p)\smax S$ we get
	$$\pr[|\bar{x} - pS]| > y] \leq 2 \exp (-\frac{y^2/2}{p(1-p)\smax S+\smax y/3}).$$
	We set $\zarib = 40 \log n$ and $y = \zarib \sqrt{\smax pS}$ to bound the probability that $|\sum x_i - pS| > \zarib \sqrt{\smax pS}$ happens. Thus we obtain
	$$\pr[|\bar{x} - pS]| > \zarib \sqrt{\smax pS}] \leq 2 \exp (-\frac{\zarib^2 \smax pS/2}{p(1-p)\smax S+\zarib \smax \sqrt{\smax pS}/3}).$$
	Since $1-p \leq 1$ we have 
	\begin{equation}\label{equation:1}
	\frac{\zarib^2 \smax pS/2}{p(1-p)\smax S}  = \frac{\zarib^2 /2}{(1-p) } \geq  \zarib^2/2.
	\end{equation}
	Moreover, by the assumption of the lemma $p \leq 1/2$ holds. In addition to this, $\smax  \leq 2pS$ and therefore 
	\begin{equation}\label{equation:2}
	\frac{\zarib^2 \smax pS/2}{\zarib \smax \sqrt{\smax pS}/3} = \frac{\zarib pS/2}{ \sqrt{\smax pS}/3} = \frac{\zarib \sqrt{pS}/2}{ \sqrt{\smax} /3} = 3 \frac{\zarib \sqrt{pS}/2}{ \sqrt{\smax}} \geq 3 \frac{\zarib \sqrt{pS}/2}{ \sqrt{2pS}} = 3\frac{\zarib /2}{\sqrt{2}} \geq 3\zarib/(2\sqrt{2}).
	\end{equation}
	It follows from Inequalities \eqref{equation:1} and \eqref{equation:2} that
	\begin{equation*}
	\begin{split}
	\frac{\zarib^2 \smax pS/2}{p(1-p)\smax S+\zarib \smax \sqrt{\smax pS}/3} &= \frac{\zarib^2 \smax pS/2}{\big [p(1-p)\smax S\big]+ \big[\zarib \smax \sqrt{\smax pS}/3\big ]}\\
	& \geq \min\{\frac{\zarib^2 \smax pS/2}{p(1-p)\smax S} , \frac{\zarib^2 \smax pS/2}{\zarib \smax \sqrt{\smax pS}/3} \}/2\\
	& \geq \min\{\zarib^2/2 ,\frac{3 \zarib}{2\sqrt{2}}\} \geq \frac{3 \zarib}{4\sqrt{2}}\\
	& \geq \zarib /2\\
	& = 20 \log n.	
	\end{split}
	\end{equation*}
	 This implies that $\exp (-\frac{\zarib^2 \smax pS/2}{p(1-p)\smax S+\zarib \smax \sqrt{\smax pS}/3}) \leq \exp(-20 \log n) \leq n^{-10} $ and thus $|pS - S'| \leq \zarib \sqrt{\smax pS}$ holds w.h.p.
\end{proof}

In our analysis, we fix an arbitrary optimal solution of the problem and state our observations based on this solution.
Since the sizes of the items are bounded by $\smax $, then either our solution uses all items and has a total size of $\sum s_i$ (if $\sum s_i$ is not larger than $t$) or leaves some of the items outside the knapsack and therefore has a size in range $[t-\smax +1, t]$. One can verify in $O(n)$ if the total size of the items is bounded by $t$ and compute the solution in this case. Therefore, from now on, we assume that the total size of the items is at least $t$ and thus the solution size is in $[t-\smax +1, t]$.

Now, if we randomly distribute the items into $t/\smax $ buckets then the expected size of the solution in each bucket is $O(\smax )$ and thus we expect the size of the solution in each bucket to be in range $[0, \tildorder(\smax )]$ w.h.p.\ due to Lemma~\ref{lemma:sizelimit}. Therefore, it suffices to compute the solution for each bucket up to a size of $\tildorder(\smax )$. Next, we use Lemma~\ref{lemma:sizelimit} to merge the solutions in faster than quadratic time. Every time we plan to merge the solutions of two sets of items $S_1$ and $S_2$, we expect the size of the solutions in these two sets to be in ranges $[t|S_1|/n - \tildorder(\sqrt{t\smax |S_1|/n}), t|S_1|/n + \tildorder(\sqrt{t\smax |S_1|/n})]$ and $[t|S_2|/n - \tildorder(\sqrt{t\smax |S_2|/n}), t|S_2|/n + \tildorder(\sqrt{t\smax |S_2|/n})]$ w.h.p. Therefore, if we only consider the values within these ranges, we can merge the solutions correctly w.h.p. and thus one can compute the solution for $S_1 \cup S_2$ w.h.p. in time $\tildorder(\sqrt{t\smax (|S_1|+|S_2|)/n}^2) = \tildorder(t\smax (|S_1|+|S_2|)/n)$. This enables us to compute the solution w.h.p. in time $\tildorder(\smax (n+t))$.

\begin{algorithm}[H]
	\label{alg:sizelimit}
	\KwData{Knapsack size $t$ and $n$ items $(s_i,v_i)$ where $1 \leq s_i \leq \smax $ for all items.}
	\KwResult{Solution for knapsack size $t$}
	
	Randomly distribute the items into $t/\smax $ buckets\;    
	\For{$j \in [t/\smax ]$}{
		$x_{1,j} = $ solution of the problem for bucket $i$ up to size $(\zarib + 2) \smax $\;
	}    
	
	\For{$i \in [2, \lceil \log (t/\smax ) \rceil]$}{
		\For{$j \in \lceil t/\smax /2^i \rceil$}{
			Combine the solutions of $x_{i-1,2j-1}$ and $x_{i-1,2j}$ into $x_{i,j}$ (based on Lemma~\ref{lemma:sizelimit} )\;
		}
	}    
	\textbf{Return } $\max{x_{\lceil \log (t/\smax ) \rceil,1}}$\;        
	\caption{$\knapsackforsmallsizes$}\label{algorithm:alg1}
\end{algorithm}

\begin{theorem}\label{theorem:limitedsize}
	There exists a randomized algorithm that correctly computes the solution of the knapsack problem in time $\tildorder(\smax (n+t))$ w.h.p., if the item sizes are integers in range $[1,\smax ]$.
\end{theorem}
\begin{proof}
	We assume w.l.o.g. that the total size of the items is at least $t$ and thus the solution size is in range $[t-\smax +1,t]$.
	As outlined earlier, we randomly put the items into $t/\smax $ buckets. Based on Lemma~\ref{lemma:sizelimit}, the expected size of the solution in each bucket is in range $[\smax -1, \smax ]$. Therefore, by Lemma~\ref{lemma:sizelimit} w.h.p. the size of the solution in every bucket is at most $\smax +\tildorder(\smax ) = \tildorder(\smax )$. Therefore, for each bucket with $n_i$ items we can compute the solution up to size $\tildorder(\smax )$ in time $\tildorder(\smax  n_i )$. Since $\sum n_i = n$, the total running time of this step is $\tildorder(\smax n)$.
	
	We merge the solutions in $\log (t/\smax )$ rounds. In every round $i$, we make $t/\smax /2^i$ merges each corresponding to the solutions of $2^i$ buckets. By Lemma~\ref{lemma:sizelimit}, the range of the solution size in every merge is $[\smax 2^i - \tildorder(\sqrt{\smax ^22^i}), \smax 2^i + \tildorder(\sqrt{\smax ^22^i})]$ w.h.p. Thus, every merge takes time $\smax ^22^i$. Moreover, in every round $i$ the number of merges is $t/\smax /2^i$. Therefore, the total running time of each phase is $\tildorder(\smax t)$ and thus the algorithm runs in time $\tildorder(\smax (n+t))$. In order to show our solution is correct with probability at least $1-n^{-10}$, we argue that we make at most $n$ merges and therefore the total error of our solution is at most $nn^{-10} = n^{-9}$. Thus, if we run  Algorithm~\ref{algorithm:alg1} twice and output the better of the generated answers, our error is bounded by $2(n^{-9})^2 = n^{-18}/2 \leq n^{-10}$ and thus the output is correct with probability at least $1-n^{-10}$.
\end{proof}

As a corollary of Theorem~\ref{theorem:limitedsize}, we can also solve the unbounded knapsack problem in time $\tildorder(\smax (n+t))$ if the sizes of the items are bounded by $\smax $.

\begin{corollary}[of Theorem~\ref{theorem:limitedsize}]\label{corollary:unboundedlimitedsize}
	There exists a randomized  $\tildorder(\smax (n+t))$ time algorithm that solves the unbounded knapsack problem w.h.p. when the sizes are bounded by $\smax$.
\end{corollary}
\begin{proof}
The crux of the argument is that in an instance of the unbounded knapsack problem if the sizes of two items are equal, we never use the item with the smaller value in our solution. Thus, this leaves us with $\smax $ different items. We also know that we use each item of size $s_i$ at most $\lfloor t/s_i \rfloor$ times and thus if we copy the most profitable item of each size $s_i$, $\lfloor t/s_i \rfloor$ times, this gives us an instance of the 0/1 knapsack problem with $O(t \log \smax )$ items. Using the algorithm of Theorem~\ref{theorem:limitedsize} we can solve this problem in time $\tildorder(\smax t)$. Since the reduction takes time $O(n)$ the total running time is $\tildorder(\smax (n+t))$.
\end{proof}

Using the same idea, one can also solve the problem in time $\tildorder((n+t)\smax)$ when each item has a given multiplicity.
\newpage
 \section{Strongly Polynomial Time Algorithms for Knapsack with Multiplicities}\label{sec:knapsack-multiplicity}
In this section, we study the knapsack problem where items have multiplicities. We assume throughout this section that the sizes of the items are bounded by $\smax$. More precisely, for every item $(s_i, v_i)$, $m_i$ denotes the number of copies of this item that can appear in any solution. We show that when all the sizes are integers bounded by $\smax$, one can solve the problem in time $\tildorder(n \smax ^2 \min\{n,\smax\})$. Notice that this running time is independent of $t$ and thus our algorithm runs in strongly polynomial time. This result improves upon the $O(n^3 \smax ^2)$ time algorithm of~\cite{tamir2009new}. 

We begin, as a warm-up, by considering the case where $m_i = \infty$ for all items. We show that in this case, the $O(n^2 \smax ^2)$ time algorithm of~\cite{tamir2009new} can be improved to an $\tildorder(n\smax +  \smax^2 \min\{n,\smax\})$ time algorithm. Before we explain our algorithm, we state a mathematical lemma that will be later used in our proofs.

\begin{lemma}\label{lemma:math}
    Let $S$ be a subset of items with integer sizes. If $|S| \geq k$ then there exists a non-empty subset of $S$ whose total size is divisible by $k$.
\end{lemma}
\begin{proof}
    Select $k$ items of $S$ and give them an arbitrary ordering. Let $s_i$ be the total size of the first $i$ items in this order. Therefore, $0 = s_0 < s_1 < s_2 < \ldots < s_{k}$ holds. By pigeonhole principal, from set $\{s_0, s_1, \ldots, s_k\}$ two numbers have the same remainder when divided by $k$. Therefore, for some $i < j$ we have $s_i \textsf{ Mod } k = s_ j \textsf{ Mod } k$. This means that the total size of the items in positions $i+1$ to $j$ is divisible by $k$.
\end{proof}

When all multiplicities are infinity, our algorithm is as follows: define $\best := \arg \max v_i / s_i$ to be the index of an item with the highest ratio of $v_i / s_i$ or in other words, the most profitable item. We claim that there always exists an optimal solution for the knapsack problem in which the total size of all items except $(s_\best, v_\best)$ is bounded by $\smax ^2$.

\begin{lemma}\label{lemma:best}
    Let $\ii$ be an instance of the knapsack problem where the multiplicity of every item is equal to infinity and let $(s_\best, v_\best)$ be an item with the highest ratio of $v_i / s_i$. There exists an optimal solution for $\ii$ in which the number of items except $(s_\best, v_\best)$ is smaller than $s_\best$.
\end{lemma}
\begin{proof}
    We begin with an arbitrary optimal solution and modify the solution until the condition of the lemma is met. Due to Lemma~\ref{lemma:math}, every set $S$ with at least $s_\best$ items contains a subset whose total size is divisible by $s_\best$. Therefore, until the number of items other than $(s_\best, v_\best)$ drops below $s_\best$, we can always find a subset of such items whose total size is divisible by $s_\best$. Next, we replace this subset with multiple copies of $(s_\best, v_\best)$ with the same total size. Since $v_\best/s_\best$ is the highest ratio over all items, the objective value of the solution doesn't hurt, and thus it remains optimal.
\end{proof}

Since $s_i \leq \smax$ holds for all items, Lemma~\ref{lemma:best} implies that in such a solution, the total size of all items except $(s_\best, v_\best)$ is bounded by $\smax ^2$.  This implies that at least $\max\{0,\lfloor(t-\smax^2)/s_\best\rfloor\}$ copies of item $(s_\best, v_\best)$ appear in an optimal solution. Thus, one can put these items into the knapsack and solve the problem for the remaining space of the knapsack. Let the remaining space be $t'$ which is bounded by $\smax ^2 + \smax$. Therefore, the classic $O(nt')$ time algorithm for knapsack finds the solution in time $O(n \smax ^2)$. Also, by Theorem~\ref{theorem:limitedsize}, one can solve the problem in time $\tildorder((n+t')\smax) = \tildorder(n \smax + \smax^3)$. Thus, the better of two algorithms runs in time $\tildorder(n\smax + \smax^2 \min\{n,\smax\})$. This procedure is shown in Algorithm~\ref{alg:infinitymultiplicitiesalgorithm}.
 \begin{algorithm}
     \KwData{A knapsack size $t$ and $n$ items with sizes and values $(s_i,v_i)$. $m_i = \infty$ and $s_i \leq \smax$ hold for all $1 \leq i \leq n$}
     \KwResult{The solution of the knapsack problem for knapsack size $t$}

        $\best \leftarrow \arg \max v_i / s_i$\;
	    $\mathsf{cnt} \leftarrow \max\{0, \lfloor (t-\smax^2)/s_\best\}\rfloor$\;
        $t' \leftarrow t - \mathsf{cnt} \cdot s_\best$\;
        \If{$n \leq \smax$}{
	        \textbf{Report} $\mathsf{cnt} \cdot v_\best + \classicknapsackalgorithm(t',n,\{(s_1,t_1),(s_2,t_2),\ldots,(s_n,t_n)\},\{m_1,m_2,\ldots,m_n\})$\;
	    }\Else{
			\textbf{Report} $\mathsf{cnt} \cdot v_\best + \knapsackforsmallsizes(t',n,\{(s_1,t_1),(s_2,t_2),\ldots,(s_n,t_n)\},\{m_1,m_2,\ldots,m_n\})$\;
		}

     \caption{$\infinitymultiplicitiesalgorithm$}\label{alg:infinitymultiplicitiesalgorithm}
 \end{algorithm}

\begin{theorem}\label{theorem:strong1}
    When $s_i \in  [\smax]$ and $m_i = \infty$ hold for every item, Algorithm~\ref{alg:infinitymultiplicitiesalgorithm} computes the solution of the knapsack problem in time $\tildorder(n\smax + \smax^2 \min\{n,\smax\})$.
\end{theorem}
\begin{proof}
    The main ingredient of this proof is Lemma~\ref{lemma:best}. According to Lemma~\ref{lemma:best}, there exists a solution in which apart from $(s_\best, v_\best)$ type items, the total size of the remaining items is bounded by $\smax ^2$. Therefore, we are guaranteed that at least $\mathsf{cnt}$ copies of item $(s_\best, v_\best)$ appear in an optimal solution of the problem. Thus, the remaining space of the knapsack ($t'$) is at most $\smax ^2 + \smax$ and therefore Algorithm~\ref{alg:infinitymultiplicitiesalgorithm} solves the problem in time $\tildorder(n \smax + \smax^2 \min\{n, \smax\})$.
\end{proof}

Next, we present our algorithm for the general case where every multiplicity $m_i \geq 1$ is a given integer number. Our solution for this case runs in time $\tildorder(n \smax ^2 \min\{n, \smax\})$. We assume w.l.o.g. that $t \geq \smax ^2$, otherwise the better of the classic knapsack algorithm and our limited size knapsack algorithm solves the problem in time $\tildorder(n \smax + \smax ^2 \min\{n, \smax\})$. In addition to this, we assume that the items are sorted in decreasing order of $v_i / s_i$, that is
$$v_1 / s_1 \geq v_2 / s_2 \geq \ldots \geq v_n / s_n.$$
We define $t' = t-\smax ^2$ to be a smaller knapsack size which is less than $t$ by an additive factor of $\smax^2$. We construct a pseudo solution for the smaller knapsack problem, by putting the items one by one into the smaller knapsack (of size $t'$) greedily. We stop when the next item does not fit into the knapsack. Let $b_i$ be the number of copies of item $(s_i, v_i)$ in our pseudo solution for the smaller knapsack problem. In what follows, we show that there exists an optimal solution for the original knapsack problem such that if $b_i \geq \smax$ holds for some item $(s_i, v_i)$, then at least $b_i - \smax$ copies of $(s_i, v_i)$ appear in this solution.

\begin{lemma}\label{lemma:stronggood}
	Let $b_i$ denote the number of copies of item $(s_i, v_i)$ in our pseudo solution for the smaller knapsack problem. There exists an optimal solution for the original knapsack problem that contains at least $b_i - \smax$ copies of each item $(s_i, t_i)$ such that $b_i \geq \smax$.
\end{lemma}
\begin{proof}
	To show this lemma, we start with an optimal solution and modify it step by step to make sure the condition of the lemma is met.
	We denote the number of copies of item $(s_i, v_i)$ in our solution by $a_i$.
	In every step, we find the smallest index $i$ such that $a_i < b_i - \smax$. Notice that due to the greedy nature of our algorithm for constructing the pseudo solution and the fact that $b_i > 0$  then $b_j = m_j$ for every $j < i$. Hence, $a_j \leq m_j =  b_j$ holds for all $j \leq i$. Since at least one copy of item $(s_i, v_i)$ is not used in the optimal solution, then the unused space in the optimal solution is smaller than $s_i$. Recall that the total size of the pseudo solution is bounded by $t' = t-\smax ^2$ and since $a_j \leq b_j$ for all $j \leq i$, then the first $i$ items contribute to at most $t-\smax^2 - \smax s_i$ space units of the solution. Moreover, as we discussed above, the total size of the solution is at least $t-\smax$ and thus the rest of the items have a size of at least $\smax^2$ in our optimal solution. Therefore we have
	$$\sum_{j=i+1}^n a_j s_j \geq \smax ^2$$ and since $s_j \leq \smax$ holds, we have $\sum_{j=i+1}^n a_j \geq \smax \geq s_i$. Based on Lemma~\ref{lemma:math} there exists a subset of these items whose total size is divisible by $s_i$ and thus we can replace them with enough (and at most $\smax$) copies of item $(s_i, v_i)$ without hurting the solution. At the end of this step $a_i$ increases and all $a_j$ for $j < i$ remain intact. Therefore after at most $\sum b_i$ steps, our solution has the desired property.
\end{proof}

What Lemma~\ref{lemma:stronggood} suggests is that although our pseudo solution may be far from the optimal, it gives us important information about the optimal solution of our problem. If our pseudo solution uses all copies of items, it means that all items fit into the knapsack and therefore the solution is trivial. Otherwise, we know that the total size of the pseudo solution is at least $t' - \smax = t - \smax ^2 - \smax$. Based on Lemma~\ref{lemma:stronggood}, for any item with $b_i \geq \smax$ we know that at least $b_i - \smax$ copies of this item appear in an optimal solution of our problem. Therefore, we can decrease the multiplicity of such items by $b_i - \smax$ and decrease the knapsack size by $(b_i - \smax) s_i$. We argue that after such modifications, the remaining size of the knapsack is at most $\smax + \smax^2 + n \smax ^2$. Recall that the total size of the pseudo solution is at least $t - \smax ^2 - \smax$ and therefore
$\sum b_i s_i \geq t - \smax ^2 - \smax$. This implies that
\begin{equation*}
\begin{split}
\sum \max\{0, b_i - \smax\} s_i & \geq \sum (b_i - \smax) s_i \\
&= \sum b_i s_i - \sum \smax s_i \\
& \geq [t - \smax ^2 - \smax] - \sum \smax s_i \\
&\geq  [t - \smax ^2 - \smax] - \sum \smax ^2 \\
& = [t - \smax ^2 - \smax] - n \smax ^2\\
& = t - \smax - (n+1) \smax ^2
\end{split}
\end{equation*}
Therefore, after the above modifications, the remaining size of the knapsack is at most $\smax + (n+1) \smax ^2$. Thus, we can solve the problem in time $\tildorder(n \smax^3)$ using Lemma~\ref{theorem:limitedsize} and solve the problem in time $\tildorder(n^2 \smax^2)$ using the classic knapsack algorithm. This procedure is explained in details in Algorithm~\ref{alg:givenmultiplicitiesalgorithm}.

 \begin{algorithm}
	\KwData{A knapsack size $t$ and $n$ items with sizes and values $(s_i,v_i)$. $n$ multiplicities $m_1, m_2, \ldots, m_n$. $s_i \leq \smax$ holds for all $1 \leq i \leq n$}
	\KwResult{The solution of the knapsack problem for knapsack size $t$}
	$t' \leftarrow \max\{0, t - \smax ^2 \}$\;
	\For {$i \in [1,n]$}{
		$b_i \leftarrow \min\{m_i, \lfloor t' / s_i \rfloor\}$\;
		$t' \leftarrow t' - b_i s_i$\;
		\If{$b_i \neq m_i$}{
			\textbf{break}\;
		}
	}
	$t'' \leftarrow t$\;
	$\mathsf{surplus} \leftarrow 0$\;
	\For {$i \in [1,n]$}{
		$t'' \leftarrow t'' - \max\{0,b_i - \smax\} s_i$\;
		$m'_i \leftarrow m_i - \max\{0,b_i - \smax\}$ \;
		$\mathsf{surplus} \leftarrow \mathsf{surplus} + \max\{0,b_i - \smax\} v_i$\;
	}
	\If{$n \leq \smax$}{
		\textbf{Report} $\mathsf{surplus} + \classicknapsackalgorithm(t'',n,\{(s_1,t_1),(s_2,t_2),\ldots,(s_n,t_n)\},\{m'_1,m'_2,\ldots,m'_n\})$\;
	}\Else{
		\textbf{Report} $\mathsf{surplus} + \knapsackforsmallsizes(t'',n,\{(s_1,t_1),(s_2,t_2),\ldots,(s_n,t_n)\},\{m'_1,m'_2,\ldots,m'_n\})$\;
	}

	\caption{$\givenmultiplicitiesalgorithm$}\label{alg:givenmultiplicitiesalgorithm}
\end{algorithm}
  
\begin{theorem}\label{theorem:strong2}
Algorithm~\ref{alg:givenmultiplicitiesalgorithm} solves the knapsack problem in time $\tildorder(n  \smax^2 \min\{n,\smax\})$ when the sizes of the items are integers in range $[1,\smax]$ and each item has a given integer multiplicity.
\end{theorem}
\begin{proof}
The proof is based on Lemma~\ref{lemma:stronggood}. After determining the values of vector $b'$, we know that for each item $(s_i, t_i)$ at least $b_i - \smax$ copies appear in the solution. Thus, we can remove the space required by these items and reduce the knapsack size. As we discussed before, after all these modifications, the new knapsack size ($t''$) is bounded by $\smax +  (n+1) \smax ^2$ and thus the better of the classic knapsack algorithm and the algorithm of Section~\ref{sec:sizelimit} solve the problem in time $\tildorder(n \smax^2 \min\{n,\smax\})$.
\end{proof}
\newpage
\section{Related Convolution Problems}\label{sec:related}
In Sections~\ref{sec:01sparsity} and~\ref{sec:treeseparability} we discuss other convolution type problems that are similar to knapsack. We mentioned in the introduction that
several problems seem to be closely related to the 
\Problem{knapsack} and \Problem{convolution} problems.
We define and mention some previous results for some of those problems here.
The \FProblem{tree sparsity} problem asks for a maximum-value subtree of size $k$
from a given node-valued tree.
The best-known algorithm for the problem runs in $O(kn)$ time, which is
quadratic for $k=\Theta(n)$.  Backurs \etal~\cite{backurs2017better} show that
it is unlikely to obtain a strongly subquadratic-time algorithm for
this problem, since it implies the same for the 
\Problem{$(\min, +)$ convolution} problem.  They provide the first
single-criterion $(1+\eps)$-approximation for \Problem{tree sparsity} that
runs in near-linear time, and works on arbitrary trees.
Given a set of integers and a target, the \FProblem{subset sum} problem looks for a subset
whose sum matches the target.  To find all the \emph{realizable} integers up to $u$ takes
time $\tilde O(\min\{\sqrt nu, u^{4/3}, \sigma\})$, where $\sigma$ is
the sum of the given input integers~\cite{koiliaris2017faster},  improving upon the
simple $O(nu)$ dynamic-programming solution~\cite{Bellman:1957}.
Finding out whether a specific $t$ is realizable may be done in $\tilde O(n+t)$ randomized
time, matching certain conditional lower bounds~\cite{bringmann2017near}.
%
%
The \FProblem{least-value sequence} problem is studied in
K\"unnemann \etal~\cite{KPS17}: given a sequence of $n$ items and
a (perhaps succinctly represented) not necessarily positive
value function for every pair,
find a subsequence that minimizes the sum of values of adjacent pairs.
Several problems such as \Problem{longest chain of nested boxes}, \Problem{vector domination},
and a \Problem{coin change} problem fit in this category and are considered.
For each of these, the authors identify ``core'' problems, which help to either demonstrate hardness or
design fast algorithms.
The fastest algorithms for \FProblem{language edit distance} is based on
computing the $(\min, +)$ product of two $n\times n$ matrices, which
also solves \Problem{all pairs shortest paths}.  Bringmann \etal~\cite{BGSW16}
show that the \Problem{matrix product} can be computed in subcubic time
if one matrix has bounded differences in either rows or columns.

While \FProblem{minimum convolution}\footnote{It is also called
    $(\min,+)$ convolution, min-sum convolution, inf-convolution, infimal convolution or
    the epigraphical sum in the literature.} admits a near linear-time
$(1+\eps)$-approximation~\cite{CGIMPS16,backurs2017better},
we do not know of a strongly subquadratic-time exact algorithm for it.
The best-known algorithm runs in time $O(n^2(\log\log n)^3/\log^2n)$~\cite{BCDEHILPT12}.
Some special cases have faster algorithms, though:
$O(n)$ time for convex sequences,
and $O(n\log n)$ time for randomly permuted sequences~\cite{maragos:book,Bussieck:1994}.
Moreover additive combinatorics allows us to solve the \Problem{convolution} problem
for increasing integers bounded by $O(n)$ in randomized time $O(n^{1.859})$ and
deterministic $O(n^{1.864})$~\cite{chan2015clustered}.

Cygan \etal~\cite{cygan2017problems} study \Problem{minimum convolution} as a hardness
assumption, and identify several problems that are as hard.
First of all, \Problem{minimum convolution} is known to reduce to either
\Problem{three-sum} or \Problem{all pairs shortest paths} problem,
though no reduction in the other direction is known, and the relation
of the latter two is not known.  (The \FProblem{three-sum} problem asks whether
three elements of a given set of $n$ numbers sum to zero.)
Despite the recent progress on the subset sum problem, which is a
special case of the \Problem{$0/1$ knapsack} problem, the latter
is shown to be equivalent to \Problem{minimum convolution}.  (The former reduces
to \Problem{$(\vee,\wedge)$ convolution} that can be solved via FFT.)
A similar reduction exists for the \Problem{unbounded knapsack} problem.
\newpage
\section{Tree Separability}\label{sec:treeseparability}
0/1 knapsack, unbounded knapsack, and tree sparsity along with a few other combinatorial optimization problems have been shown to be computationally equivalent with respect to subquadratic algorithms~\cite{cygan2017problems}. In other words, a subquadratic algorithm for any problem in this list yields a subquadratic algorithm for the rest of the problems. In this section, we introduce the tree separability problem and show that this problem is indeed computationally equivalent to the rest of the problems of the list. Next, in Section~\ref{sec:treeseparability:fast}, we show that in some cases, a bounded weight tree separability problem can be solved in better than subquadratic time. This result in nature is similar to the algorithms we provide for knapsack problems.

In the tree separability problem, we are given a tree $T$ with $n$ nodes and $n-1$ edges. Every edge $e = (i,j)$ is associated with a weight $w_e$. The goal of this problem is to partition the vertices of $T$ into two (not necessarily connected) partitions of size $m$ and $n-m$ in a way that the total weight of the crossing edges is minimized. A special case of the problem where $|m - (n-m)| \leq 1$ is known as tree bisection.

\subsection{Equivalence with $(\max,+)$ Convolution}\label{sec:treeseparability:reduction}
To show a subquadratic equivalence, we first present an indirect reduction from $(\max,+)$ convolution to tree separability. We use the \textit{\textsf{MaxCov-UpperBound}} as an intermediary problem in our reduction. Cygan \etal~\cite{cygan2017problems} show that any subquadratic algorithm for \textsf{MaxCov-UpperBound} yields a subquadratic solution for $(\max,+)$ convolution. 

\begin{definition}
In the \textsf{MaxCov-UpperBound} problem, we are given two vectors $a$ and $b$ of size $n$ and a vector $c$ of size $2n-1$. The goal is to find out whether there exists an $i$ such that $(a \ttimes b)_i > c_i$.
\end{definition}

\begin{lemma}[proven in~\cite{cygan2017problems}]
    Any subquadratic solution for \textsf{MaxCov-UpperBound} yields a subquadratic solution for the $(\max,+)$ convolution.
\end{lemma}

The main idea of our reduction is as follows: Given three vectors $a$, $b$, and $c$ with sizes $n$, $n$, and $2n-1$ one can construct a tree consisted of three paths joining at a vertex $r$. We show that based on the solution of the tree separability on this tree, one can determine if $(a \ttimes b)_{i} \geq c_i$ for some $0 \leq i < 2n-1$.

\begin{lemma}\label{lemma:red1}
    Any subquadratic algorithm for tree separability results in a subquadratic algorithm for the $(\max,+)$ convolution.
\end{lemma}
\begin{proof}
As we mentioned earlier, we prove this reduction through \textsf{MaxCov-UpperBound}. Suppose we are given two vectors $a$ and $b$ of size $n$ and a vector $c$ of size $2n-1$ and are asked if $(a \ttimes b)_i > c_i$ for some $0 \leq i < |c|$. We answer this question by constructing a tree of size $8n$ as follows: Let $M = 10 \max\{1,|a_0|,|a_1|,\ldots,|a_{n-1}|,|b_0|,|b_1|,\ldots,|b_{n-1}|\}$ be a large enough number. The root of the tree is a vertex $r$ and three paths are connected to vertex $r$. The vertices of each path correspond to the elements of one vector. Thus, we denote the vertices of the paths by $a'_i$, $b'_i$, and $c'_i$ respectively. For every $c'_i$ we set the weight of the edge between $c'_i$ and $c'_{i-1}$ (or $r$ in case of $i = 0$) equal to $M + c_{2n-2-i}$. Similar to this, for every $0 \leq i < n$, we set the weight of the edge between $a'_{i+n}$ and $a'_{i+n-1}$ equal to $M - a_i$ and the weight of the edge between $b'_{i+n}$ and $b'_{i+n-1}$ equal to $M - b_i$. The rest of the edges have weight $\infty$. This construction is illustrated in Figure~\ref{fig:treeseparability}. Our claim is that for $m = 4n-1$ the solution of the tree separability problem is at least $3M$ if and only if $a \ttimes b$ is bounded by $c$.

We first prove that if for some $a_i$ and $b_j$ we have $a_i + b_j = (a \ttimes b)_{i+j} > c_{i+j}$ then the solution of the tree separability problem for $m = 4n-1$ is smaller than $3M$. To this end, we put the following vertices in one partition and the rest of the vertices in the second partition: $$\{r,a'_0,a_1,\ldots,a'_{n+i-1},b'_0,b_1,\ldots,b'_{n+j-1},c'_0,c'_1,\ldots,c'_{2n-i-j-3}\}.$$
Notice that the above list contains exactly $4n-1$ vertices. Moreover, the only crossing edges of this solution are the ones connected to $a'_{n+i}$, $b'_{n+j}$, and $c'_{2n-i-j-2}$ with weights $M-a_i$, $M-b_j$, and $M+c_{i+j}$. Since $a_i + b_j > c_{i+j}$ we have $M-a_i + M-b_j + M+c_{i+j} < 3M$ and thus the solution of the tree separability problem is smaller than $3M$.

Finally, we show that if $a_i + b_j \geq c_{i+j}$ holds for all $i$ and $j$, then the solution of the tree separability problem is at least $3M$. Notice that since $M$ is large enough, in order for a solution to have a weight smaller than $3M$ it has to meet the following constraints:
\begin{itemize}
    \item The solution should not contain an edge with weight $\infty$.
    \item The number of crossing edges between the partitions should be at most $3$.
\end{itemize}
In any solution that meets the above constraints, the partition with size $4n-1$ contains vertex $r$. Moreover, none of the crossing edges in parts $a'$ and $b'$ have weight $\infty$ and thus the crossing edges correspond to two vertices $a'_{n+i}$ and $b'_{n+j}$ with $0 \leq i,j <n$ and therefore their weights are $a_i$ and $b_j$. Since the size of the partition is $4n-1$, the third crossing edge has a weight of $c_{i+j}$. Recall that we assume $c_{i+j} \geq a_i + b_j$ and therefore $M-a_i + M-b_j +M+c_{i+j} \geq 3M$.
\tikzstyle{H-node}=[rectangle,draw=black,fill=white!30,inner sep=1.3mm]
\tikzstyle{B-node}=[circle,draw=blue,fill=blue!20,inner sep=2.5mm]
\tikzstyle{G-node}=[circle,draw=black,fill=white!30,inner sep=3.3mm]
\tikzstyle{R-node}=[rectangle,draw=red,fill=red!20,inner sep=2.6mm]
\tikzstyle{W-node}=[rectangle,draw=white,fill=white!30,inner sep=0.2mm]
\tikzstyle{test-node}=[circle,draw=black,fill=black,inner sep=.2mm]

\tikzstyle{bl0} = [draw=black, thick, dashed]   
\tikzstyle{b9} = [draw=red, thick]   
\tikzstyle{b8} = [draw=blue, thick, dotted]   
\tikzstyle{bl1} = [->, draw=black]   
\tikzstyle{bl2} = [draw=black!70,thick]   
\tikzstyle{bl3} = [draw=black,thick, dotted]   

\tikzstyle{br0} = [draw=brown, dashed]   
\tikzstyle{br1} = [->, draw=brown]   
\tikzstyle{br2} = [->, draw=brown,thick]   

\tikzstyle{red0} = [draw=red, thick, dashed]   
\tikzstyle{red1} = [draw=red]   
\tikzstyle{red2} = [draw=red,thick]   

\tikzstyle{gr0} = [draw=green, thick, dashed]   
\tikzstyle{gr1} = [draw=green]   
\tikzstyle{gr2} = [draw=green,thick]   
\tikzstyle{gr4} = [draw=green,semithick,rounded corners]   

\begin{figure}
\begin{center}
\begin{tikzpicture}[scale=0.5][domain=0:8]
\draw (0,0) node[G-node,label=center:$r$,label=below:] (r) {};
\draw (0,3) node[G-node,label=center:$c'_{0}$,label=below:] (c'_0) {};
\draw (0,6) node[G-node,label=center:$c'_{1}$,label=below:] (c'_1) {};
\draw (0,9) node[G-node,label=center:$c'_{2}$,label=below:] (c'_2) {};
\draw (0,13) node[G-node,label=center:$c'_{2n-2}$,label=below:$\vdots$] (c'_{2n-2}) {};

\draw (-3,0) node[G-node,label=center:$a'_{0}$,label=below:] (a'_0) {};
\draw (-6,0) node[G-node,label=center:$a'_{1}$,label=below:] (a'_1) {};
\draw (-9.3,0) node[G-node,label=center:$a'_{n-1}$,label=right:$\ldots$] (a'_{n-1}) {};
\draw (-12.3,0) node[G-node,label=center:$a'_{n}$,label=below:] (a'_{n}) {};
\draw (-15.6,0) node[G-node,label=center:$a'_{2n-1}$,label=right:$\ldots$] (a'_{2n-1}) {};
\draw (-15.6,3) node[G-node,label=center:$a'_{2n}$] (a'_{2n}) {};
\draw (-15.6,6) node[G-node,label=center:$a'_{2n+1}$] (a'_{2n+1}) {};
\draw (-15.6,9) node[G-node,label=center:$a'_{2n+2}$] (a'_{2n+2}) {};
\draw (-15.6,13) node[G-node,label=center:$a'_{3n-1}$,label=below:$\vdots$] (a'_{3n-1}) {};

\draw (3,0) node[G-node,label=center:$b'_{0}$,label=below:] (b'_0) {};
\draw (6,0) node[G-node,label=center:$b'_{1}$,label=below:] (b'_1) {};
\draw (9.3,0) node[G-node,label=center:$b'_{n-1}$,label=left:$\ldots$] (b'_{n-1}) {};
\draw (12.3,0) node[G-node,label=center:$b'_{n}$,label=below:] (b'_{n}) {};
\draw (15.6,0) node[G-node,label=center:$b'_{2n-1}$,label=left:$\ldots$] (b'_{2n-1}) {};
\draw (15.6,3) node[G-node,label=center:$b'_{2n}$] (b'_{2n}) {};
\draw (15.6,6) node[G-node,label=center:$b'_{2n+1}$] (b'_{2n+1}) {};
\draw (15.6,9) node[G-node,label=center:$b'_{2n+2}$] (b'_{2n+2}) {};
\draw (15.6,13) node[G-node,label=center:$b'_{3n-1}$,label=below:$\vdots$] (b'_{3n-1}) {};

\draw[b9] (r)  to node [below,label=right:$M+c_{2n-2}$] {} (c'_0) ;
\draw[b9] (c'_0)  to node [below,label=right:$M+c_{2n-3}$] {} (c'_1) ;
\draw[b9] (c'_1)  to node [below,label=right:$M+c_{2n-4}$] {} (c'_2) ;
\draw[b8] (r)  to node [below,label=below:$\infty$] {} (a'_0) ;
\draw[b8] (a'_0)  to node [below,label=below:$\infty$] {} (a'_1) ;
\draw[b9] (a'_{n-1})  to node [below,label=below:$M-a_0$] {} (a'_{n}) ;
\draw[b8] (a'_{2n-1})  to node [below,label=right:$\infty$] {} (a'_{2n}) ;
\draw[b8] (a'_{2n})  to node [below,label=right:$\infty$] {} (a'_{2n+1}) ;
\draw[b8] (a'_{2n+1})  to node [below,label=right:$\infty$] {} (a'_{2n+2}) ;

\draw[b8] (b'_{2n-1})  to node [below,label=left:$\infty$] {} (b'_{2n}) ;
\draw[b8] (b'_{2n})  to node [below,label=left:$\infty$] {} (b'_{2n+1}) ;
\draw[b8] (b'_{2n+1})  to node [below,label=left:$\infty$] {} (b'_{2n+2}) ;

\draw[b8] (r)  to node [below,label=below:$\infty$] {} (b'_0) ;
\draw[b8] (b'_0)  to node [below,label=below:$\infty$] {} (b'_1) ;
\draw[b9] (b'_{n-1})  to node [below,label=below:$M-b_0$] {} (b'_{n}) ;






\end{tikzpicture}
\end{center}
\caption{Dashed edges have weight $\infty$ while the weight of the solid edges is based on the value of the vectors $a$, $b$, and $c$.}
\label{fig:treeseparability}
\end{figure}
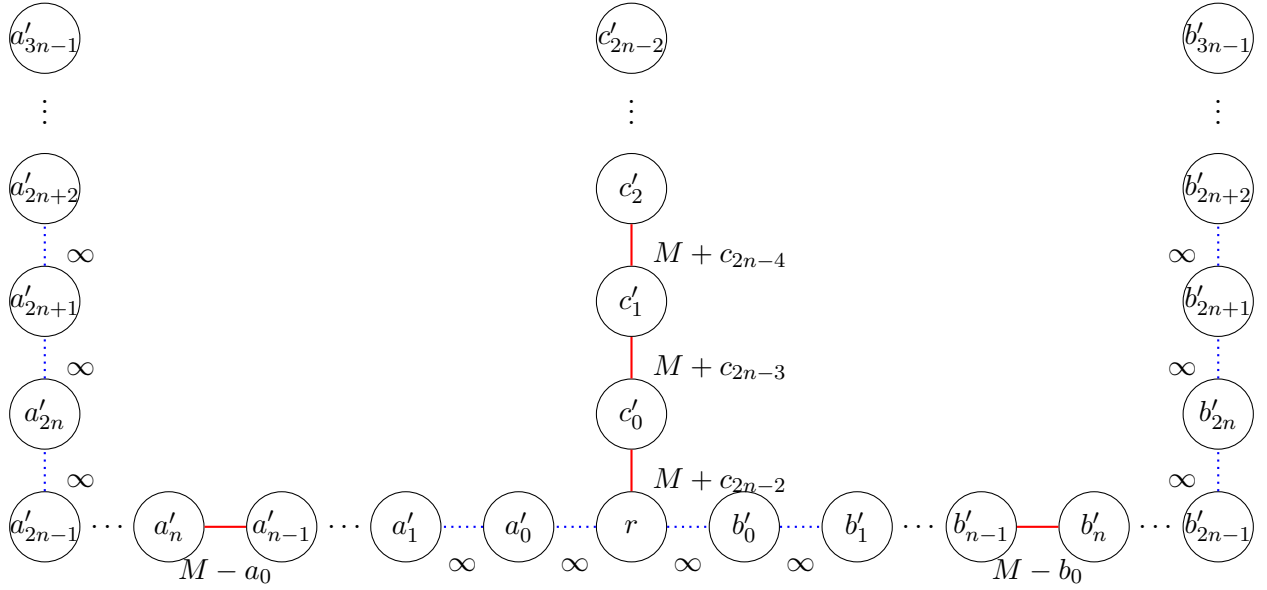
\end{proof}

We also show in Appendix~\ref{appendix:treeseparability} that a $T(n)$ time algorithm for computing the $(\max,+)$ convolution of two vectors yields an $\tildorder(T(n))$ time algorithm for solving tree separability. The proof is very similar to the works of Cygan \etal~\cite{cygan2017problems} and Backurs \etal~\cite{backurs2017better}. If the height of the tree is small, the classic dynamic program yields a running time of $T(n)$. We use the spine decomposition of~\cite{sleator1983data} to deal with cases where the height of the tree is large.

\begin{lemma}\label{lemma:red2}
    Any $T(n)$ time algorithm for solving $(\max,+)$ convolution yields an $\tildorder(T(n))$ time algorithm for tree separability.
\end{lemma}

Lemmas~\ref{lemma:red1} and~\ref{lemma:red2} imply that $(\max,+)$ convolution and tree separability are computationally equivalent.

\begin{theorem}[A corollary of Lemmas~\ref{lemma:red1} and~\ref{lemma:red2}]\label{theorem:equivalence}
    $(\max,+)$ convolution and tree separability are computationally equivalent with respect to subquadratic algorithms.
\end{theorem}
\subsection{Fast Algorithm for Special Cases}\label{sec:treeseparability:fast}
We show that when the maximum degree of the tree and the weights of the edges are bounded by $\dmax$ and $\wmax$, one can solve the problem in time $\tildorder(\dmax \wmax n)$. Note that our algorithm only works when the edge weights are integers. In particular, when both $\wmax$ and $\dmax$ are $O(1)$ our algorithm runs in (almost) linear time.

Our main observation is the following: for any $1 \leq m < n$, there exists a partitioning of a given tree $T$ into two partitions of sizes $m$ and $n-m$ such that the number of crossing edges is bounded by $2\dmax \log n$.

\begin{lemma}\label{lemma:bound}
    Let $T$ be a tree of size $n$ and $1 \leq m < n$ be an integer number. There exists a partitioning of $T$ into two partitions of sizes $m$ and $n-m$ with at most $2 \dmax \log n$ crossing edges where $\dmax$ is the maximum degree of a vertex in $T$.
\end{lemma}
\begin{proof}
It is a well-known fact that every tree of size $n$ has a vertex $v$ such that if we remove $v$ from the tree, the size of each connected component of the tree is bounded by $(2/3)n$~\cite{west2001introduction}. Based on this observation, we inductively construct a solution with no more than $2 \dmax \log n$ crossing edges for any tree with $n$ vertices and a given partition size $m$. The base case is $n \leq 2$ for which the lemma holds trivially. Now, for a given tree $T$ of size $n$, we find its center $v$ with the above property. Next, we remove $v$ from the tree to obtain $d_v$  connected components $T_1, T_2, \ldots, T_{d_v}$. To construct a solution, we first put vertex $v$ in the partition that is to be of size $m$. We then continue growing this partition by adding the subtrees to it one by one. We stop when adding any subtree to the solution increases the size of the partition to more than $m$. Let $m'$ be the size of the partition and $T_i$ be the first subtree that cannot be entirely added to the solution. If $m' = m$ our solution is valid, otherwise we recursively partition $T_i$ into two partitions of sizes $m-m'$ and $|T_i|-(m-m')$ and update the solution by adding the $m-m'$ part to it. Note that we have at most $\dmax$ crossing edges for vertex $v$ and also based on the induction hypothesis, the number of crossing edges in partitions of $T_i$ is at most $2\dmax \log |T_i|$. In addition to this, since $v$ is a center of the tree, we have $|T_i| \leq n2/3$ and therefore $2\dmax \log|T_i| \leq (2\log n -1)\dmax$. Thus, the total number of crossing edges in our solution is bounded by $2\dmax \log n$.
\end{proof}

It follows from Lemma~\ref{lemma:bound} that when the maximum degree of a tree is bounded by $\dmax$ and the maximum weight of the edges is bounded by $\wmax$ then the solution of the tree separability problem is bounded by $2\dmax \wmax \log n$. Therefore, the values of the solutions for every subproblem of the tree separability problem are bounded by $2\dmax \wmax \log n$. Thus, every time we wish to compute the convolution of two vectors $a$ and $b$ corresponding to the solutions of the subproblems, the $(\max, +)$ convolution can be computed in time $\tildorder(\dmax \wmax (|a| + |b|))$ (Lemma~\ref{lemma:mult1}). Thus, based on Lemma~\ref{lemma:red2}, we can compute the solution of the tree separability problem in time $\tildorder(\dmax \wmax n)$.

\begin{theorem}[A corollary of Lemmas~\ref{lemma:bound} and~\ref{lemma:red2}]
Given a tree $T$ with $n$ nodes whose maximum degree is bounded by $\dmax$. If the weights of the edges are integers bounded by $\wmax$, one can compute the solution of the tree separability problem for $T$ in time $\tildorder(\dmax \wmax n)$.
\end{theorem}

\newpage
\section{0/1 Tree Sparsity}\label{sec:01sparsity}
We show in Sections~\ref{sec:maxplusconvolution},~\ref{sec:knapsack}, and~\ref{sec:treeseparability} that convolution, knapsack, and tree separability problems can be solved in almost linear time in special cases. One of the important problems that lies in the same computational category with these problems is the tree sparsity problem. Therefore, an important question that remains open is whether an almost linear time algorithm can solve the tree sparsity problem when the weights of the vertices are small integers. In particular, is it possible to solve the 0/1 tree sparsity problem (in which the weight of every vertex is either $0$ or $1$) in $\tildorder(n)$ time? We show in this section that tree sparsity is the hardest problem of this category when it comes to small weights. More precisely, we show that an $\tildorder(n)$ time algorithm for 0/1 tree sparsity immediately implies linear time solutions for the rest of the problems when the input values are small. We assume that the goal of the tree sparsity problem is to find for every $i$ what is the weight of the heaviest connected component of the tree of size $i$.

To this end, we define the $\dmax$-distance bounded $(\max,+)$ convolution problem as follows: given two vectors $a$ and $b$ with the condition that $\max |a_i - a_{i-1}| \leq \dmax$ and $\max |b_i - b_{i-1}| \leq \dmax$ hold for every $i$. The goal is to compute $a \ttimes b$. We show that a $T(n)$ time algorithm for 0/1 tree sparsity yields an $\tildorder(T(\dmax n))$ time algorithm for $\dmax$-distance bounded $(\max, +)$ convolution of two integer vectors  $a$ and $b$ where $n = |a| + |b|$. This yields fast algorithms for convolution, knapsack, tree separability, and tree separability when the input values are small integers. Indeed we already know that convolution and knapsack problems admit almost linear time algorithms for such special cases, nonetheless such a reduction sheds light on the connection between these problems.

We begin by showing that a $T(n)$ time algorithm for $0/1$ tree sparsity yields an $\tildorder(T(\dmax n))$ time algorithm for $\dmax$-distance bounded convolution.
\begin{lemma}
    Given a $T(n)$ time algorithm for $0/1$ tree sparsity, one can solve the $\dmax$-distance bounded convolution for integer vectors in time $O(T(\dmax n))$.
\end{lemma}
\begin{proof}
    Suppose we are given two $\dmax$-distance bounded vectors $a$ and $b$ and wish to compute $a \ttimes b$. We assume w.l.o.g. that both $a$ and $b$ are of size $n$. Also, we can assume w.l.o.g. that both vectors are increasing because of the following fact: If we add $i (\dmax+1)$ to every element $i$ of both vectors $a$ and $b$, they both become increasing since $|a_i - a_{i-1}| \leq \dmax$ and $|b_i - b_{i-1}| \leq \dmax$ hold for the original vectors. Moreover, if we compute $c = a \ttimes b$ for the new vectors, one can compute the solution for the convolution of the original vectors by just subtracting $i (\dmax+1)$ from every element $i$ of vector $c$. In addition to this, since the original vectors are $\dmax$-distance bounded, after adding $i(\dmax+1)$ to every element $i$ of the vectors, the resulting vectors are $(2\dmax+1)$-distance bounded.
    
    For the rest of the proof, we assume both vectors $a$ and $b$ are increasing and of size $n$. Moreover, both vectors are $\dmax$-distance bounded and thus both $|a_i - a_{i-1}| \leq \dmax$ and $|b_i - b_{i-1}| \leq \dmax$ hold. We also assume $a_0=b_0 = 0$ since one can ensure that constraint by shifting the values. To compute $a \ttimes b$, we construct an instance of the $0/1$ tree sparsity problem with $n+1+a_{n-1}+b_{n-1}$ vertices. The underlying tree has a root $r$ connected to three different paths $a',b',c'$. Path $c'$ contains $n$ consecutive vertices, each with weight $1$. Paths $a'$ and $b'$ correspond to vectors $a$ and $b$. Vertices of path $a'$ are denoted by $a'_1, a'_2, \ldots, a'_{a_{n-1}}$. The weight of each $a'_i$ is $1$ if and only if there exists a $j$ such that $a_j = i$. Similar to this, the vertices of path $b'$ are denoted by $b'_1, b_2, \ldots, b'_{b_{n-1}}$ and the weight of a vertex $b'_i$ is equal to 1 if and only if $b_j = i$ for some $j$. An example of such construction is shown in Figure~\ref{fig:treesparsity}.
    \tikzstyle{H-node}=[rectangle,draw=black,fill=white!30,inner sep=1.3mm]
\tikzstyle{B-node}=[circle,draw=blue,fill=blue!20,inner sep=2.5mm]
\tikzstyle{G-node}=[circle,draw=black,fill=white!30,inner sep=3.3mm]
\tikzstyle{R-node}=[rectangle,draw=red,fill=red!20,inner sep=2.6mm]
\tikzstyle{W-node}=[rectangle,draw=white,fill=white!30,inner sep=0.2mm]
\tikzstyle{test-node}=[circle,draw=black,fill=black,inner sep=.2mm]

\tikzstyle{bl0} = [draw=black, thick, dashed]   
\tikzstyle{b9} = [draw=red, thick]   
\tikzstyle{b8} = [draw=blue, thick, dotted]   
\tikzstyle{bl1} = [->, draw=black]   
\tikzstyle{bl2} = [draw=black!70,thick]   
\tikzstyle{bl3} = [draw=black,thick, dotted]   

\tikzstyle{br0} = [draw=brown, dashed]   
\tikzstyle{br1} = [->, draw=brown]   
\tikzstyle{br2} = [->, draw=brown,thick]   

\tikzstyle{red0} = [draw=red, thick, dashed]   
\tikzstyle{red1} = [draw=red]   
\tikzstyle{red2} = [draw=red,thick]   

\tikzstyle{gr0} = [draw=green, thick, dashed]   
\tikzstyle{gr1} = [draw=green]   
\tikzstyle{gr2} = [draw=green,thick]   
\tikzstyle{gr4} = [draw=green,semithick,rounded corners]   

\begin{figure}
\begin{center}
\begin{tikzpicture}[scale=0.5][domain=0:8]
\draw (0,0) node[G-node,label=center:$1$,label=below:,label=below:$r$] (r) {};
\draw (0,3) node[G-node,label=center:$1$,label=below:,label=right:$c'_1$] (c'_1) {};
\draw (0,6) node[G-node,label=center:$1$,label=below:,label=right:$c'_2$] (c'_2) {};
\draw (0,10) node[G-node,label=center:$1$,label=below:$\vdots$,label=right:$c'_8$] (c'_8) {};
\draw (0,13) node[G-node,label=center:$1$,label=right:$c'_9$] (c'_9) {};

\draw (-3,0) node[G-node,label=center:$1$,label=below:$a'_1$] (a'_1) {};
\draw (-6,0) node[G-node,label=center:$0$,label=below:$a'_2$] (a'_2) {};
\draw (-9.3,0) node[G-node,label=center:$0$,label=below:$a'_3$] (a'_3) {};
\draw (-12.3,0) node[G-node,label=center:$1$,label=below:$a'_4$] (a'_4) {};

\draw (3,0) node[G-node,label=center:$0$,label=below:$b'_1$] (b'_1) {};
\draw (6,0) node[G-node,label=center:$1$,label=below:$b'_2$] (b'_2) {};
\draw (9.3,0) node[G-node,label=center:$0$,label=below:$b'_3$] (b'_3) {};
\draw (12.3,0) node[G-node,label=center:$1$,label=below:$b'_4$] (b'_4) {};
\draw (15.6,0) node[G-node,label=center:$1$,label=below:$b'_5$] (b'_5) {};

\draw[b9] (r)  to node [below,label=right:] {} (a'_1) ;
\draw[b9] (a'_1)  to node [below,label=right:] {} (a'_2) ;
\draw[b9] (a'_2)  to node [below,label=right:] {} (a'_3) ;
\draw[b9] (a'_3)  to node [below,label=below:] {} (a'_4) ;

\draw[b9] (r)  to node [below,label=right:] {} (c'_1) ;
\draw[b9] (c'_1)  to node [below,label=right:] {} (c'_2) ;
\draw[b9] (c'_8)  to node [below,label=below:] {} (c'_9) ;

\draw[b9] (r)  to node [below,label=right:] {} (b'_1) ;
\draw[b9] (b'_1)  to node [below,label=right:] {} (b'_2) ;
\draw[b9] (b'_2)  to node [below,label=right:] {} (b'_3) ;
\draw[b9] (b'_3)  to node [below,label=below:] {} (b'_4) ;
\draw[b9] (b'_4)  to node [below,label=below:] {} (b'_5) ;






\end{tikzpicture}
\end{center}
\caption{The above tree corresponds to vectors $a = \langle 0,1,3,4 \rangle$ and $b = \langle 0,2,4,5\rangle$.}
\label{fig:treesparsity}
\end{figure}
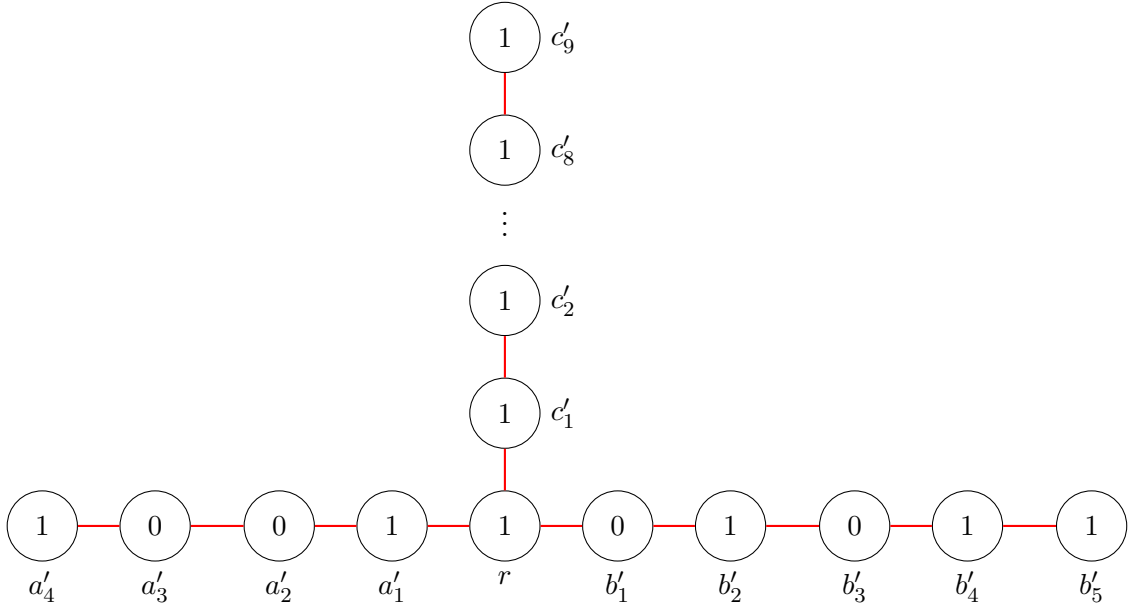
    Let $s$ be the solution vector to the tree sparsity problem explained above. In other words, $s$ is a vector of size $n+2+a_{n-1}+b_{n-1}$ where $s_i$ is the heaviest connected subtree of size $i$. We argue that for every $0 \leq i < |(a \ttimes b)|$, $(a \ttimes b)_i$ is equal to $j-n-1$ for the smallest $j$ such that $s_j = n+i+1$. If this claim is correct then computing $(a \ttimes b)$ can be trivially done provided that the solution vector $s$ is given.  
    
    In order to prove the above statement, we first start with an observation.
    \begin{observation}\label{observation:hi1}
    For any $1 \leq i \leq n+1+a_{n-1}+b_{n-1}$ there exists an optimal solution for size $i$ that contains vertex $r$.
    \end{observation}
    \begin{proof}
    For $i \leq n+1$, one can simply start with vertex $r$ and go along path $c'$ to collect $i$ vertices with weight $1$. Obviously, this is the best we can achieve since the weight of every vertex is bounded by $1$. For $i > n+1$ we argue that the weight of the solution is at least $n+1$ since path $c'$ along with vertex $r$ and some subpath of $a'$ and $b'$ suffice to have a solution weight at least $n+1$. In addition to this, if we remove vertex $r$ from the tree, any connected subtree contains at most $n$ vertices with weight $1$ and thus any solution not containing vertex $r$ has a weight of at most $n$. Thus, one can obtain an optimal solution containing path $c'$ and vertex $r$ from any optimal solution.
    \end{proof}

    \begin{observation}\label{observation:hi2}
    For any $i \geq n+1$, there always exists an optimal solution of size $i$ that contains both the entire path $c'$ and vertex $r$.    
    \end{observation}
    \begin{proof}
    By Observation~\ref{observation:hi1} we already know that there always exists an optimal solution of size $i$ that contains vertex $r$. Notice that the weight of all vertices in path $c'$ is equal to $1$. Thus, if the solution doesn't contain all vertices of path $c'$, one can iteratively remove a leaf from the solution and instead add the next vertex in path $c`$ to the solution. This does not hurt the solution since the weight of all vertices in path $c'$ is equal to $1$.
    \end{proof}

    Now, for an $i \geq n+1$ consider such a solution of the tree sparsity problem for size $i$. Apart from $n+1$ vertices of path $c'$ and $r$, such a solution contains a prefix of path $a'$ and a prefix of path $b'$. The number of vertices with weight one in each of these paths indicates how many indices in the corresponding vectors have a value equal to or smaller than the number of vertices of the solution in that path. This implies that if vector $s$ is the solution of the tree sparsity problem then $(a \ttimes b)_i$ is equal to $j-n-1$ for the smallest $j$ such that $s_j = n+1+i$. 
    
    Notice that since $a$ and $b$ are $\dmax$-distance bounded, both $a_{n-1}$ and $b_{n-1}$ are bounded by $\dmax n$ and thus one can solve the corresponding tree sparsity problem in time $T((\dmax+1)n+1)$. Since $T$ is at most quadratic, the running time is $O(T(\dmax n))$.
\end{proof}

In Section~\ref{sec:maxplusconvolution}, we discussed that when input values for convolution are small integers bounded by $\emax$, an $\tildorder(\emax n)$ time algorithm can solve the problem exactly. Notice that when the input values are bounded by $\emax$, the input vectors are also $\emax$-distance bounded and thus an $\tildorder(T(\emax n))$ is also possible via a reduction to 0/1 tree sparsity. We showed in Section~\ref{sec:knapsack} that the knapsack problem reduces to knapsack convolution and used the prediction technique to solve the knapsack convolution in time $\tildorder(\vmax n)$ when the item values are bounded by $\vmax$. However, it follows from the definition that if the item values are integers bounded by $\vmax$ then knapsack convolution is a special case of $\vmax$-distance bounded convolution. Thus, a $T(n)$ time solution for 0/1 tree sparsity implies a $\tildorder(T(\vmax n))$ time solution for bounded value knapsack. It is interesting to observe that the solutions of both tree sparsity and tree separability are $\wmax$-distance bounded if the weights are integers bounded by $\wmax$. Therefore, a $T(n)$ time algorithm for 0/1 tree sparsity also implies a $\tildorder(T(\wmax n))$ time algorithm for tree sparsity and tree separability when the weights are integers bounded by $\wmax$. In particular, an $\tildorder(n)$ time algorithm for 0/1 tree sparsity yields $\tildorder(\wmax n)$ time algorithms for both tree sparsity and tree separability when the weights are integers bounded by $\wmax$.

\begin{theorem}
    A $T(n)$ time solution for 0/1 tree sparsity implies the following:
    \begin{itemize}
        \item A $\tildorder(T(\emax n))$ time algorithm for bounded knapsack convolution when the values are integers in range $[0,\emax]$.
        \item A $\tildorder(T(\vmax n))$ time algorithm for 0/1 knapsack when the item values are integers in range $[0,\vmax]$.
        \item A $\tildorder(T(\wmax n))$ time algorithm for tree sparsity when the vertex weights are integers in range $[0,\wmax]$.
        \item A $\tildorder(T(\wmax n))$ time algorithm for tree separability when the vertex weights are integers in range $[0,\wmax]$.
    \end{itemize}
\end{theorem}
\newpage
\section{Reduction to  Polynomial Convolution}\label{sec:maxplusconvolution:reduction}
Given two integer vectors $a$ and $b$ with the condition that all $a_i$'s and $b_i$'s are in range $[0,\emax ]$ we wish to find the convolution of the two vectors in the $(\max,+)$ setting. We denote this by $a \ttimes b$. Let $n = |a| + |b|$ and define $a'$ as a vector with the same size as $a$ as follows: $$\forall 0 \leq i < |a| \hspace{2cm}a'_i = (n+1)^{a_i}.$$ Similarly, we assume $b'$ is a vector with the same size as $b$ such that $$\forall 0 \leq i < |b| \hspace{2cm}b'_i = (n+1)^{b_i}.$$
This way, for all $i$ and $j$ we have $a'_i b'_j = (n+1)^{a_i} (n+1)^{b_i} = (n+1)^{a_i + b_j}$. Let $c = a \ttimes b$ be the solution of the problem and $c' = a' \times b'$ be the polynomial multiplication of $a'$ and $b'$, thus, for every $0 \leq i < |c'|$ we have
\begin{equation*}
c'_{i} = \sum_{j=0}^i a'_i b'_{i-j}
= \sum_{j=0}^i (n+1)^{a_j+b_{i-j}}
\leq i \max _{j=0}^i (n+1)^{a_j+b_{i-j}}
\leq i (n+1)^{\max _{j=0}^i  {a_j+b_{i-j}}}
\leq i (n+1)^{c_i}
\leq n (n+1)^{c_i}.
\end{equation*}

Similarly one can show that 
\begin{equation*}
c'_{i} = \sum_{j=0}^i a'_i b'_{i-j}
= \sum_{j=0}^i (n+1)^{a_j+b_{i-j}}
\geq  \max _{j=0}^i (n+1)^{a_j+b_{i-j}}
\geq  (n+1)^{\max _{j=0}^i  {a_j+b_{i-j}}}
\geq  (n+1)^{c_i},
\end{equation*}
and thus $(n+1)^{c_i} \leq c'_i \leq n(n+1)^{c_i}$. Therefore, $c_i = \lfloor \log_{n+1} c'_i \rfloor$ and thus one can determine vector $c$ from $c'$. This reduction shows that any algorithm for computing $c'$ in time $\tildorder(\emax n)$ yields a similar running time for computing $c$. Polynomial multiplication can be computed in time $\tildorder(n)$ via FFT when the running time of each arithmetic operation is $O(1)$~\cite{thomas2001introduction}. However, here, every element of $a'$ and $b'$ can be as large as $(n+1)^\emax $ and thus every arithmetic operation for such values takes time $\tildorder(\emax )$. Thus, the total running time of the FFT method for computing $c'$ is $\tildorder(\emax n)$. This yields an $\tildorder(\emax n)$ time algorithm for computing $a \ttimes b$.

\begin{algorithm}[H]
	\label{alg:mult1}
	\KwData{Two vectors $a$ and $b$}
	\KwResult{$a \ttimes b$}
	
	$a' \leftarrow \text{ a vector of size }|a| \text{ s.t } a'_i = (n+1)^{a_i}$\;	
	$b' \leftarrow \text{ a vector of size }|b| \text{ s.t } b'_i = (n+1)^{b_i}$\;	
	
	$c' = \mathsf{FFTPolynomialConvolution}(a',b')$\;
	$c \leftarrow \text{ a vector of size }|c'| \text{ s.t } c_i := \lfloor \log_{n+1}c'_i \rfloor$\;	
	\textbf{Return } $c$\;		
	\caption{\textsf{BoundedRangeConvolution}($a,  b$)}
\end{algorithm}

\begin{lemma}[also used in~\cite{zwick2002all,chan2015clustered,bringmann2017near,backurs2017better,zwick1998all}]\label{lemma:mult1}
	Given two vectors $a$ and $b$ whose values are integers in the range $[0,\emax ]$, there exists an algorithm to compute $a \ttimes b$ in time $\tildorder(\emax n)$ where $n = |a| + |b|$.
\end{lemma}

Of course Lemma~\ref{lemma:mult1} holds whenever the range of the values of the vectors is an interval of length $\emax$. It has been also shown that Lemma~\ref{lemma:mult1} holds even when the input values are in set $\{0,1,\ldots,\emax,-\infty\}$. The reason behind this is that since the numbers are small, one can replace $-\infty$ by $-2\emax$ and solve the problem in time $\tildorder(n \emax)$ using the same procedure. Then, we replace every negative value of the solution by $-\infty$. The same idea can solve the problem when input values are allowed to be $\infty$ as well as $-\infty$. 
An immediate consequence of Lemma~\ref{lemma:mult1} is that even if the vectors do not have integer values, still one can use this method to approximate the solution within a small additive error. We show this in Lemma~\ref{lemma:mult2}.

\begin{lemma}\label{lemma:mult2}
	Let $a$ and $b$ be two given vectors whose values are real numbers in range $[0,\emax ]$. One can compute in time $\tildorder(\emax n)$ a vector $c$ with the same size as $|a \ttimes b|$ such that $(a \ttimes b)_i -1 < c_i \leq (a \ttimes b)_i$ holds for all $0 \leq i < |c|$ where $n = |a| + |b|$.
\end{lemma}
\begin{proof}
For a vector $x$, let $\lfloor x \rfloor$ be an integer vector of the same size where $\lfloor x \rfloor_i  = \lfloor x_i \rfloor$. Moreover, for a vector $x$ we define $\alpha x$ as a vector of the same size where $(\alpha x)_i = \alpha x_i$. Note that for a given vector $x$, both $\lfloor x \rfloor$ and $\alpha x$ can be computed in time $O(n)$ from $x$. We argue that $c = 1/2(\lfloor 2a \rfloor \ttimes \lfloor 2b \rfloor)$ meets the conditions of the lemma. This observation proves the lemma since all values of $\lfloor 2a \rfloor$ and $\lfloor 2b \rfloor$ are integers in range $[0,2\emax ]$ and thus one can compute $\lfloor 2a \rfloor \ttimes \lfloor 2b \rfloor$ in time $\tildorder(\emax n)$. With additional $O(n)$ operations we compute $c$ from $\lfloor 2a \rfloor \ttimes \lfloor 2b \rfloor$.

Notice that for every $i$ we have $2a_i -1 < \lfloor 2a_i \rfloor \leq 2a_i$ and similarly $2b_i -1 < \lfloor 2b_i \rfloor \leq 2b_i$. Therefore, for every $0 \leq i < |a \ttimes b|$ we have $(2a \ttimes 2b)_i - 2 < (\lfloor 2a \rfloor \ttimes \lfloor 2b \rfloor)_i \leq (2a \ttimes 2b)_i$. This implies that $(a \ttimes b)_i - 1 < (1/2(\lfloor 2a \rfloor \ttimes \lfloor 2b \rfloor))_i \leq (a \ttimes b)_i$ holds for all $0 \leq i < |a \ttimes b|$ and thus the proof is complete.
\end{proof}

Similar to Lemma~\ref{lemma:mult1}, Lemma~\ref{lemma:mult2} also holds when $-\infty$ and $\infty$ are allowed in the input.
\begin{algorithm}[H]
	\label{alg:mult2}
	\KwData{Two vectors $a$ and $b$}
	\KwResult{An approximate solution to $a \ttimes b$}
	
	$a' \leftarrow \lfloor 2a \rfloor$\;	
	$b' \leftarrow \lfloor 2b \rfloor$\;	
	
	$c' = \mathsf{BoundedRangeConvolution}(a',b')$\;
	$c = c' / 2$\;	
	\textbf{Return } $c$\;		
	\caption{\textsf{ApproximateConvolution}($a,  b$)}
\end{algorithm}
\newpage
\section{Omitted Proofs of Section~\ref{sec:maxplusconvolution:solutionrange}}\label{sec:step2-omitted}

\begin{proof}[of Lemma~\ref{lemma:rangegood}]

\textbf{first condition:} Suppose for the sake of contradiction that the  condition of the lemma doesn't hold for some $i$ and $j$. We assume w.l.o.g. that $a_i - b_i \geq a_j - b_j$ and thus $a_i - b_i > a_j - b_j + \emax $. This implies that $a_i + b_j > a_j + b_i +\emax $ and hence $(a \ttimes b)_{i+j} \geq a_i + b_j > a_j + b_i +\emax $ which contradicts the first assumption of the lemma.

\textbf{second condition:} Based on the assumption of the lemma we have $(a \ttimes b)_{i+k} - \emax  \leq a_i+b_k \leq (a \ttimes b)_{i+k}$. Similarly, $(a \ttimes b)_{2j} - \emax  \leq a_j + b_j \leq (a \ttimes b)_{2j}$. Since $j-i = k-j$ then $2j = i+k$ and thus both $a_j + b_j$ and $a_i + b_k$ are lower bounded by  $(a \ttimes b)_{i+k} - \emax $ and upper bounded by $(a \ttimes b)_{i+k}$. Hence we have $|(a_i + b_k) - (a_j + b_j)| \leq \emax $. If we add the term  $[(b_k-a_k) - (b_j-a_j)]$ to the expression $(a_i + a_k) - 2a_j$ we obtain
\begin{equation*} 
\begin{split}
|(a_i + a_k) - 2a_j + [(b_k-a_k) - (b_j-a_j)]| &= |(a_i + (b_k - a_k) + a_k) - (2a_j + (b_j - a_j))| \\
& = |(a_i + b_k) - (a_j + b_j)|\\
& \leq \emax,
\end{split}
\end{equation*}
which implies $|(a_i + a_k) - 2a_j| \leq \emax  + |(b_k-a_k) - (b_j-a_j)|$. Recall that we proved $|(b_k-a_k) - (b_j-a_j)| \leq \emax $ and thus $|(a_i + a_k) - 2a_j| \leq 2\emax $. 
\end{proof}

\begin{proof}[of Lemma~\ref{lemma:range3}]
	We present a simple algorithm and show that (i) it provides a correct solution for the problem and (ii) its running time is $\tildorder(\emax n)$. In this algorithm, we construct two vectors $a'$ and $b'$ from $a$ and $b$ such that all values of $a'$ and $b'$ are in range $[0,6\emax ]$ and $a \ttimes b$ can be computed from $a' \ttimes b'$. The key idea here is that if we add a constant $C$ to all components of either $a$ or $b$, this value is added to  all elements of $(a \ttimes b)$. Moreover, for a fixed $C$, if we add  a value of $iC$ to every element $i$ of both $a$ and $b$, then every $(a \ttimes b)_{i}$ is increased by exactly $iC$. Based on these observations, we construct two vectors $a'$ and $b'$ of size $n$ from $a$ and $b$ as follows:
	\begin{align*}
	\hspace{2cm}a'_i &:= a_i + [3\emax  - a_0]&+ i[(a_0-a_{n-1})/(n-1)]&\hspace{2cm}\\
	\hspace{2cm}b'_i &:= b_i + [3\emax +a_{n-1} - a_0 -b_{n-1}]&+ i[(a_0-a_{n-1})/(n-1)].&\hspace{2cm}
	\end{align*}
	The transformation formulas are basically the application of the above operations to $a$ and $b$ which are delicately chosen to make sure the values of $a'$ and $b'$ fall in range $[0,6\emax ]$. Notice that vectors $a'$ and $b'$ might have fractional values. However, we show that all the values of these vectors are in range $[0,6\emax ]$. Thus, we can use the algorithm of Lemma~\ref{lemma:mult2} to compute in time $\tildorder(\emax n)$ an approximate solution $c'$ to $a' \ttimes b'$ within an error less than $1$. Next, based on vector $c'$ we construct a solution $c$ as follows:
	$$c_i := \lceil c'_i - [6\emax +a_{n-1} - 2a_0 -b_{n-1}] - i[(a_0-a_{n-1})/(n-1)] \rceil.$$
	Finally, we report $c$ as the solution to $a \ttimes b$. This procedure is shown in Algorithm~\ref{alg:solutionrangegood1}.
	
	\begin{algorithm}
		\KwData{Two integer vectors $a$ and $b$ of size $n$ meeting the condition of Lemma~\ref{lemma:range3}}
		\KwResult{$a \ttimes b$}
		
		$a' \leftarrow \text{ a vector of size }n \text{ s.t } a'_i = a_i + [3\emax  - a_0]+ i[(a_0-a_{n-1})/(n-1)]$\;	
		$b' \leftarrow \text{ a vector of size }n \text{ s.t } b'_i= b_i + [3\emax +a_{n-1} - a_0 -b_{n-1}]+ i[(a_0-a_{n-1})/(n-1)]$\;	
		
		$c' \leftarrow \mathsf{ApproximateConvolution}(a',b')$\;
		$c \leftarrow \text{ a vector of size }2n-1 \text{ s.t } c_i := \lceil c'_i - [6\emax +a_{n-1} - 2a_0 -b_{n-1}] - i[(a_0-a_{n-1})/(n-1)]\rceil$\;	
		\textbf{Return } $c$\;		
		\caption{\textsf{DistortedNTimesNConvolution}($a,  b$)}
		\label{alg:solutionrangegood1}
	\end{algorithm}

	In what follows, we show that $c$ is indeed equal to $a \ttimes b$ and that our algorithm runs in time $\tildorder(\emax n)$. We first point out a few observations regarding $a'$ and $b'$:
	
	\begin{observation}\label{observation:1}
		$a'_0 = a'_{n-1}  = b'_{n-1} = 3\emax $.
	\end{observation}
	\begin{proof}
		According to the formula, 
		\begin{equation*}
		\begin{split}
		a'_{n-1} &= a_{n-1} + [3\emax  - a_0]+ (n-1)[(a_0-a_{n-1})/(n-1)]\\
		& = a_{n-1} + [3\emax  - a_0]+ (a_0-a_{n-1}) \\
		& = 3\emax  + [a_0 - a_0]+ (a_{n-1}-a_{n-1}) \\ 
		&=  3\emax .\\
		\end{split}
		\end{equation*}
		Moreover,
		\begin{equation*}
		\begin{split}
		a'_0 &= a_0 + [3\emax  - a_0]+ 0[(a_0-a_{n-1})/(n-1)]\\
		& = a_0 + [3\emax  - a_0] \\
		&= 3\emax  + [a_0 - a_0] \\
		&= 3\emax .\\
		\end{split}
		\end{equation*} 
		Finally, for $b'_{n-1}$ we have
		\begin{align*}
		b'_{n-1} &= b_{n-1} + [3\emax +a_{n-1} - a_0 -b_{n-1}]+ (n-1)[(a_0-a_{n-1})/(n-1)]\\
		&= b_{n-1} + [3\emax +a_{n-1} - a_0 -b_{n-1}]+ (a_0-a_{n-1})\\
		&= 3\emax  + [b_{n-1} +a_0 - a_0 -b_{n-1}]+ (a_{n-1}-a_{n-1})\\
		&= 3\emax . 
		\qedhere
		\end{align*}
	\end{proof}
	
	\begin{observation}\label{observation:2}
		For every $0 \leq i,j < n$ we have $a'_i + b'_j \geq (a' \ttimes b')_{i+j} - \emax $ and $$(a' \ttimes b')_{i+j} = (a \ttimes b)_{i+j} + [6\emax +a_{n-1} - 2a_0 -b_{n-1}] + (i+j)[(a_0-a_{n-1})/(n-1)].$$ 
	\end{observation}
	\begin{proof}
		Based on the construction of $a'$ and $b'$ we have 
		\begin{equation*}
		\begin{split}
		a'_i + b'_j &= a_i + [3\emax  - a_0]+ i[(a_0-a_{n-1})/(n-1)] \\
		&+ b_j + [3\emax +a_{n-1} - a_0 -b_{n-1}]+ j[(a_0-a_{n-1})/(n-1)]\\
		&= a_i + b_j + [6\emax +a_{n-1} - 2a_0 -b_{n-1}] + i[(a_0-a_{n-1})/(n-1)] + j[(a_0-a_{n-1})/(n-1)]\\
		&= a_i + b_j + [6\emax +a_{n-1} - 2a_0 -b_{n-1}] +(i+j)[(a_0-a_{n-1})/(n-1)].
		\end{split}
		\end{equation*}	
		Notice that $[6\emax +a_{n-1} - 2a_0 -b_{n-1}] + (i+j)[(a_0-a_{n-1})/(n-1)]$ is the same for all pairs of indices that sum up to $i+j$ and thus $$(a' \ttimes b')_{i+j} = (a \ttimes b)_{i+j} + [6\emax +a_{n-1} - 2a_0 -b_{n-1}] + (i+j)[(a_0-a_{n-1})/(n-1)].$$ Since we have $a_i + b_j \geq (a \ttimes b)_{i+j} - \emax $, by adding the $[6\emax +a_{n-1} - 2a_0 -b_{n-1}] + (i+j)[(a_0-a_{n-1})/(n-1)]$ part to both sides of the inequality we get $a'_i + a'_j \geq  (a' \ttimes b')_{i+j} - \emax $.
	\end{proof}
	
	Next, we use Observations~\ref{observation:1} and~\ref{observation:2} to show that the values of both vectors $a'$ and $b'$ are in the interval $[0,6\emax ]$. Observation~\ref{observation:2} states that both $a'$ and $b'$ meet the conditions of Lemma~\ref{lemma:rangegood} and thus for every $0 \leq i < j <k < n$ subject to $j-i$ = $k-j$ we have $|a'_i + a'_k - 2a_j| \leq 2\emax $. Moreover, Observations~\ref{observation:1} implies that both $a'_0$ and $a'_{n-1}$ are equal to $3\emax $. Now, let $j = \arg\max a'_i$ and suppose for the sake of contradiction that $a'_j > 5\emax $. If $2j < n$, then by setting $i = 0$ and $k = 2j$ Lemma~\ref{lemma:rangegood} implies that $a'_0 + a'_{2j} \geq 2a'_j - 2\emax $. Therefore, since $a'_0 = 3\emax $ this implies $a'_{2j} - a'_j \geq a'_j - 5\emax $ and thus if $a'_j > 5\emax $ then $a'_{2j} -a'_j > 0$ contradicts the maximality of $a'_j$. If $2j \geq n$, then by setting $i = 2j-n-1$ and $k = n-1$ one could show that if $a'_j > 5\emax $, then $a'_{2j-n-1} > a'_{j}$ holds which again contradicts the maximality of $a'_j$. One can make a similar argument and show that if $j = \arg\min a'_i$ and $a'_j < \emax $, then either of $a'_{2j}$ or $a'_{2j-n-1}$ should be less than $a'_j$ which contradicts the minimality of $a'_j$. Therefore, all the values of vector $a'$ lie in the interval $[\emax ,5\emax ]$.
	
	Recall that by Observation~\ref{observation:2}, $a'$ and $b'$ meet the condition of Lemma~\ref{lemma:rangegood} and thus for all $0 \leq i < n$ we have 
	\begin{equation*}
	\begin{split}
	|(a'_i - b'_i) - (a'_{n-1} - b'_{n-1})| & = |(a'_i - b'_i) - (3\emax  - 3\emax )| \\
	& = |(a'_i - b'_i)|\\
	& \leq \emax .\\
	\end{split}
	\end{equation*}
	Since for all $0 \leq i < n$, $\emax  \leq a'_i \leq 5\emax $ holds, we have $0 \leq b'_i  \leq 6\emax $ which shows that the values of both $a'$ and $b'$ lie in the interval $[0,6\emax ]$. However, since the values are not integer, still we cannot compute $a' \ttimes b'$ in time $\tildorder(\emax n)$. Instead, we can compute in time $\tildorder(\emax n)$ a vector $c'$ such that $(a' \ttimes b')_i -1 < c'_i \leq (a' \ttimes b')_i$ holds for all $0 \leq i < |c'|$. Observation~\ref{observation:2} implies that for all $0 \leq i < |c'|$ we have 
	$$(a \ttimes b)_i -1 < c'_{i} - [6\emax +a_{n-1} - 2a_0 -b_{n-1}] - i[(a_0-a_{n-1})/(n-1)] \leq (a \ttimes b)_i.$$
	Notice that since both $a$ and $b$ are integer vectors, $a \ttimes b$ is also an integer vector and thus all its elements have integer values. Therefore, $\lceil c'_i - [6\emax +a_{n-1} - 2a_0 -b_{n-1}] - i[(a_0-a_{n-1})/(n-1)]\rceil = (a \ttimes b)_i$ for all $0 \leq i < |c'|$ and hence our algorithm computes $a \ttimes b$ correctly.
	
	With regard to the running time, all steps of the algorithm run in time $O(n)$, except where we approximate $c'$ from $a'$ and $b'$ which takes time $\tildorder(\emax n)$ since $0 \leq a'_i,b'_i \leq 6\emax $ (Lemma~\ref{lemma:mult2}). Thus, the total running time of our algorithm is $\tildorder(\emax n)$.
\end{proof}

\begin{proof}[of Lemma~\ref{lemma:range4}]
	As mentioned earlier, we show this lemma by a direct reduction to Lemma~\ref{lemma:range3}. We assume w.l.o.g. that $|b| \geq |a|$. Let $l = \lceil |b| / |a| \rceil$  and construct $l$ intervals  $(x_i, y_i)$ of length $|a|$ (i.e. $y_i = x_i + |a| - 1 $ for all $i$) as follows: for $1 \leq  i < l$ set $x_i = (i-1) |a|$ and $y_i = i |a|-1$. Also, set $x_l = |b| - |a|$ and $y_l = |b|-1$. This way, every $0 \leq i < |b|$ appears in at least one interval.
	
	Next, we construct $l$ vectors $b^1$, $b^2$, $\ldots$, $b^l$ from $b$ where every $b^i$ is a vector of length $|a|$ and $b^{i}_j = b_{x_i+j}$. We next compute $c^i = a \ttimes b^i$ for all $1 \leq i \leq l$, each in time $\tildorder(\emax |a|)$ using Lemma~\ref{lemma:range3}. Thus, the total running time of this step is $\tildorder(\emax |a|l) = \tildorder(\emax (|a|+|b|))$.
	
	Finally, we construct a solution of size $|a|+|b|-1$ initially containing $-\infty$ for all elements and for every vector $0 \leq j < 2|a|-1$ and $c^i$, we set $c_{x_i+j} := \max\{c_{x_i+j}, c^i_j\}$. This takes a total time of $O(|a| l) = O(|a| + |b|)$ and therefore the total running time of the algorithm is $\tildorder(\emax (|a|+|b|))$.
	
	\begin{algorithm}
		\KwData{Two integer vectors $a$ and $b$ meeting the condition of Lemma~\ref{lemma:range4}}
		\KwResult{$a \ttimes b$}
		
		$l \leftarrow \lceil |b| / |a| \rceil$\;
		\For {$i \in [1,l-1]$}{
			$x_i \leftarrow (i-1) |a|$\;
			$y_i \leftarrow i |a|-1$\;
		} 	
		$x_l \leftarrow |b| - |a| $\;
		$y_l \leftarrow |b| -1$\;
		
		\For {$i \in [1,i]$}{
			$b^i \leftarrow $\text{ a vector of size $|a|$ s.t. }$b^i_j = b_{x_i+j}$\;
			$c^i \leftarrow  \mathsf{DistortedNTimesNConvolution}(a,b^i)$\;
		} 
		
		$c \leftarrow $\text{ a vector of size $|a| + |b| -1 $ with values set to $-\infty$ initially}\;
		
		\For {$i \in [1,l]$}{
			\For {$j \in [0,2|a|-2]$}{
				$c_{x_i+j} \leftarrow \max\{c_{x_i+j}, c^i_j\}$\;
			}
		}
		\textbf{Return } $c$\;		
		\caption{\textsf{DistortedConvolution}($a,  b$)}\label{alg:solutionrangegood2}
	\end{algorithm}

	We argue that for all $0 \leq i < |a| + |b| -1$, $c_i \leq (a \ttimes b)_i$ holds. To this end, suppose for the sake of contradiction that $c_i > (a \ttimes b)_i$ for an $0 \leq i < |c|$. Due to our algorithm, $c_i = c^k_j$ for some $1 \leq k \leq l$ and $0 \leq j < 2|a|-1$ such that $x_k + j = i$. Hence, $c_i = (a \ttimes b^k)_j \leq (a \ttimes b)_{x_k+j}$ and since $x_k + j = i$ we have $c_i \leq (a \ttimes b)_i$ which contradicts $c_i > (a \ttimes b)_i$.
	Also, if $c_i < (a \ttimes b)_i$ for some $i$, then we argue that by definition $(a \ttimes b)_i = b_j + a_{i-j}$ for some $j$. Since every element of $b$ appears in at least one interval, there exists a $k$ such that $x_k \leq j \leq y_k$. Since $b^k_{j-x_k} = b_j$ we have $c^k_{i-x_k} = c^k_{j-x_k+(i-j)} \geq b_j + a_{i-j} =  (a \ttimes b)_i$. Note that $c_i \geq c^k_{i-x_k}$ and thus $c_i \geq (a \ttimes b)_i$ which contradicts $c_i < (a \ttimes b)_i$.
\end{proof}
\newpage
\section{Omitted Proofs of Section~\ref{sec:maxplusconvolution:prediction}}\label{sec:step3-omitted}

\begin{proof}[of Observation~\ref{observation:simple1}]
	Suppose for the sake of contradiction that $\projection(\alpha,\beta)$ is not an interval of $a$. This means that there are three integers $i < j < k$ such that $i,k \in \projection(\alpha,\beta)$ but $j \notin \projection(\alpha,\beta)$. Therefore, we have $x_i \leq x_k \leq \alpha$ and $y_k \geq y_i \geq \beta$ but either $x_j > \alpha$ or $y_j < \beta$. However, since the intervals are monotone, $x_j \leq x_k$ and also $y_j \geq y_i$ and thus $x_j \leq \alpha$ and $y_j \geq \beta$ which is a contradiction.
\end{proof}

\begin{proof}[of Observation~\ref{observation:simple2}]
	Similar to Observation~\ref{observation:simple1}, suppose for the sake of contradiction that $\projection(\alpha_1, \beta_1) \setminus \projection(\alpha_2, \beta_2)$ is not an interval. Since both of $\projection(\alpha_1, \beta_1)$ and $\projection(\alpha_2, \beta_2)$ are intervals (see Observation~\ref{observation:simple1}), this implies that there exist $i < j < k$ such that $i, j, k \in \projection(\alpha_1, \beta_1)$, $i, k \notin \projection(\alpha_2, \beta_2)$, and $j \in \setminus \projection(\alpha_2, \beta_2)$. In other words, $[\alpha_1, \beta_1]$ is a subset of $[x_i, y_i]$, $[x_j, y_j]$, and $[x_k, y_k]$. Moreover, none of $[x_i, y_i]$ and $[x_k, y_k]$ entirely contain $[\alpha_2, \beta_2]$ but $[x_j, y_j]$ contains $[\alpha_2, \beta_2]$.  Notice that if $[x_i, y_i]$ contains $\beta_2$ the monotonicity of the intervals implies that $[x_i, y_i]$ contains $[\alpha_2, \beta_2]$. Similarly, we can imply that $\alpha_2 \notin [x_k, y_k]$ and thus $[x_i, y_i] \cap [x_k, y_k] \subseteq [\alpha_2, \beta_2]$. Since $[\alpha_1, \beta_1]$ and $[\alpha_2, \beta_2]$ are disjoint, this implies $[\alpha_1, \beta_1]$ is empty and thus the solution is empty and contradicts  our assumption.
\end{proof}
\newpage
\section{An $\tildorder(\vmax t+n)$ Time Algorithm for Knapsack}\label{appendix:knapsack}
The result of this section follows from the reduction of~\cite{cygan2017problems} from knapsack to $(\max,+)$ convolution. However, for the sake of completeness we restate the reduction of~\cite{cygan2017problems} and use Theorem~\ref{theorem:knapsackconvolution} to solve knapsack in time $\tildorder(\vmax t+n)$, in case the item values are integers in range $[0,\vmax]$.

\begin{proof}[of Theorem~\ref{theorem:knapsack}] 
	For simplicity, we assume that the goal of the knapsack problem is to find the solution for all knapsack sizes $0 \leq i \leq t$.
	The blueprint of the reduction is as follows: We first divide the items into $\lceil \log t \rceil$ buckets in a way that the sizes of the items in every bucket differ by at most a multiplicative factor of two. Next, for each of the buckets, we solve the problem with respect to the items of that bucket. More precisely, for every bucket $i$ we compute a vector $c^i$ of size $t+1$ where $(c^i)_j$ is the solution to the knapsack problem for bucket $i$ and knapsack size $j$. Once we have these solutions, it only suffices to compute $c^1 \ttimes c^2 \ttimes \ldots \ttimes c^{\lceil \log t \rceil}$ and report the first $t+1$ elements as the solution. Therefore, the problem boils down to finding the solution for each of the buckets.
	
	In every bucket $i$, the size of the items is in range $[2^{i-1},2^i-1]$.  Now, if we fix a range $[r_1, r_2]$ for the item sizes, the maximum number of items used in any solution is $t / r_1$. Therefore, if we randomly put the items in $t / r_1$ categories, any fixed solution will have no more than $\polylog(t)$ items in every category. Based on this, we propose the following algorithm to solve the problem for item sizes in range $[r_1,r_2]$: randomly put the items into $t / r_1$ categories. For every category, solve the problem up to a knapsack size $r_2 \polylog(t)$ and merge the solutions. We show that merging the solutions can be done via some convolution invocations of total size $\tildorder(t)$. Thus, the only non-trivial part is to solve the problem up to a knapsack size $r_2 \polylog(t)$ for every category of items. Since $r_2/r_1 \leq 2$, each of these solutions consists of at most $\polylog(t)$ items. Cygan \etal~\cite{cygan2017problems} show that if we again put these items in $\polylog(t)$ random groups, then using $(\max, +)$ convolution one can solve the problem in almost linear time.
	We bring the pseudocode of the algorithms below. For correctness, we refer the reader to~\cite{cygan2017problems}. Here, we just show that the algorithm runs in time $\tildorder(\vmax t + n)$ if we use the knapsack convolution. Since the algorithm is probabilistic, in order to bring the success probability close to $1$, we use a factor $\zarib = \polylog(t)$ in our algorithm and run the procedures $\zarib$ times and take the best solution found in these runs. We do not specify the exact value of $\zarib$, however, since $\zarib$ is logarithmically small, it does not have an impact on the running time of our algorithms since we use the $\tildorder$ notation.
	
	Algorithm~\ref{alg:avali} solves the problem when the solution consists of at most $\zarib$ items.
	
	\begin{algorithm}[H]
		\KwData{Knapsack size $t$ and items $(s_1,v_1), (s_2, v_2), \ldots, (s_n,v_n)$}
		\KwResult{A solution vector $c$}
		
		$c \leftarrow $A vector of size $t+1$ with all $0$'s initially\;
		\For {$cnt \in [1,\zarib]$}{
			Randomly put the items in $\zarib^2$ lists $l_1, l_2, \ldots, l_{\zarib}$\;
			\For {$i \in [1,\zarib^2]$}{
				$c^i \leftarrow$ A vector of size $t+1$ where $c^i_j$ is the highest value of an item in $l_i$ with size at most $j$;
			}
			$c' \leftarrow c^1 \ttimes c^2 \ttimes \ldots c^{\zarib^2}$\;  \label{line:avali:1}
			\For {$i \in [0,t]$}{
				$c_i = \max\{c_i, c'_i\}$\;
			}
		}
		\Return c;
		\caption{\textsf{BoundedSolutionKnapsackAlgorithm}($t, \{(s_1,v_1), (s_2,v_2),\ldots\}$)}\label{alg:avali}
	\end{algorithm}
	Notice that $\zarib$ is logarithmically small in size of the original knapsack.    The only time consuming operation of the algorithm is Line~\ref{line:avali:1} which takes time $\tildorder(\vmax t)$ due to Lemma~\ref{lemma:mult1} since the item values are bounded by $\vmax$. Moreover, Algorithm~\ref{alg:avali} iterates over all items at least once. Thus, the total running time of Algorithm~\ref{alg:avali} is $\tildorder(\vmax t+n)$ for a given knapsack size $t$ and $n$ items. Algorithm~\ref{alg:dovomi} uses Algorithm~\ref{alg:avali} to solve the knapsack problem when all the item sizes are in range $[r_1, r_2]$ and $r_2 \leq 2r_1$.
	
	\begin{algorithm}[H]
		\KwData{Knapsack size $t$, items $(s_1,v_1), (s_2, v_2), \ldots, (s_n,v_n)$, and range $[r_1, r_2]$}
		\KwResult{A solution vector $c$}
		
		$c \leftarrow $A vector of size $t+1$ with all $0$'s initially\;
		\For {$cnt \in [1,\zarib]$}{
			Randomly put the items in $\lceil t/r_1 \rceil$ lists $l_1, l_2, \ldots, l_{\lceil t/r_1 \rceil}$\;
			\For {$i \in [1,\lceil t/r_1 \rceil]$}{
				$c^i \leftarrow$ \textsf{BoundedSolutionKnapsackAlgorithm}($\zarib r_2, l_i$);
			}
			$c' \leftarrow \textsf{Merge}(\{c^1, c^2, \ldots, c^{\lceil t/r_1 \rceil}\})$\; 
			\For {$i \in [0,t]$}{
				$c_i = \max\{c_i, c'_i\}$\;
			}
		}
		\Return c;
		\caption{\textsf{BoundedRangeKnapsackAlgorithm}($t, \{(s_1,v_1), (s_2,v_2),\ldots\},[r_1,r_2]$)}\label{alg:dovomi}
	\end{algorithm}
	Algorithm~\ref{alg:dovomi} puts the items into $\lceil t/r_1 \rceil$ different categories and solves each category using Algorithm~\ref{alg:avali}. Since the running time of Algorithm~\ref{alg:avali} is $\tildorder(\vmax t + n)$, except the part where we merge the solutions. In the following, we describe the algorithm for merging the solutions and show that its running time is $\tildorder(\vmax t)$ where $t$ is the original knapsack size. 
	
	\begin{algorithm}[H]
		\KwData{$k$ vectors $c^1, c^2, \ldots, c^k$ with total size $t$}
		\KwResult{$c^1 \ttimes c^2 \ttimes c^3 \ldots c^k$}
		\If{$k=1$}{
			\Return $c^1$
		}\Else{
			$a \leftarrow \textsf{Merge}(c^1,c^2,\ldots,c^{\lfloor k/2 \rfloor})$\;
			$b \leftarrow \textsf{Merge}(c^{\lfloor k/2 \rfloor +1},c^{\lfloor k/2 \rfloor +2},\ldots,c^{k})$\;
			\Return \textsf{KnapsackConvolution($a,b$)}\;
		}
		\Return c;
		\caption{\textsf{Merge}($\{c^1, c^2, \ldots, c^k\}$)}\label{alg:sevomi}
	\end{algorithm}
	Notice that Algorithm~\ref{alg:sevomi} uses the knapsack convolution to merge the vectors. Every merge for vectors with total size $n$ takes time $\tildorder(\vmax n)$. Moreover, the total size of the vectors is $\tildorder(t)$ and due to Algorithm~\ref{alg:sevomi}, the total length of the vectors in all convolutions is $\tildorder(t)$. Thus, Algorithm~\ref{alg:sevomi} runs in time $\tildorder(\vmax t)$.

	Finally, in Algorithm~\ref{alg:charomi} we merge the solutions of different buckets and report the result.
	
	\begin{algorithm}[H]
		$l_1, l_2, \ldots, l_{\lceil \log t \rceil+1} \leftarrow \lceil \log t \rceil+1$ lists of items initially empty\;
		\For{$i \in [1,\lceil \log t \rceil+1]$}{
			Put all items with size in range $[2^{i-1},2^i-1]$ in $l_i$\;
			$c^i \leftarrow \textsf{BoundedRangeKnapsackAlgorithm}(t,l_i, [2^{i-1},2^i-1])$\;
		}
		$c \leftarrow c^1 \ttimes c^2 \ttimes \ldots \ttimes c^{\lceil \log t \rceil+1}$\;
		\Return the first $t+1$ elements of $c$\;
		\caption{\textsf{KnapsackViaConvolution}($t,\{(s_1,v_1), (s_2,v_2), \ldots \}$)}\label{alg:charomi}
	\end{algorithm}
	
	Since we use the knapsack convolution for merging the solutions of different buckets, the running time of Algorithm~\ref{alg:charomi} is also $\tildorder(\vmax t+n)$;
\end{proof}
\newpage
\section{Omitted Proofs of Section~\ref{sec:knapsack:convolution}}\label{appendix:knapsack:convolution}

\begin{proof}[of Observation~\ref{observation:knapsack1}] 
	We argue that in any optimal solution, if for two items $i$ and $j$ we have $w_i/s_i > w_j/s_j$ then either $f_i = 1$ or $f_j = 0$. If not, one can increase $f_i$ by $\epsilon$ and decrease $f_j$ by $s_j\epsilon/s_i$ and obtain a better solution. Notice that for items with the same ratio of $w_i/s_i$ it doesn't matter which items are put in the knapsack so long as the total size of these items in the knapsack is fixed. Thus, the greedy algorithm provides an optimal solution. The running time of the algorithm is $O(n \log n)$ since after sorting the items we only make an iteration over the items in time $O(n)$.
\end{proof}

\begin{proof}[of Observation~\ref{observation:knapsack2}] 
	Similar to Observation~\ref{observation:knapsack1}, in order to maximize the weight we always add the item with the highest ratio of $w_i/s_i$ to the knapsack. Therefore, this yields the maximum total weight for any knapsack size $t$. The running time of the algorithm is $O(n \log n)$ since it sorts the items and puts them in the knapsack one by one.
\end{proof}

\begin{proof}[of Observation~\ref{observation:fas}] 
	This observation follows from 
        the greedy algorithm for knapsack.  Notice that
        we add the items to the knapsack greedily and the
        solution consists of two types of items: items of knapsack
        $\k_a$ and items of knapsack $\k_b$. Since the algorithm
        greedily adds the items to the solution, the order of items is
        based on $w_i/s_i$ and thus the order of items added to the
        solution for each type is also based on $w_i/s_i$. Thus, if
        for some knapsack size $x$ we define $\fa(x)$ to be the total
        size of the items in the solution of $x$ that belong to $\k_a$
        and set $\fb(x)$ equal to the size of the solution for items
        of knapsack $\k_b$, then $c'(x) = a'(\fa(x)) +
        b'(\fb(x))$. The monotonicity of $\fa$ and $\fb$ follow from
        the fact that in order to update the solution we only add
        items and we never remove any item from the solution.
\end{proof}

In the proofs of Observations~\ref{observation:inc} and
\ref{observation:dec} we refer to the solution of $a'(x)$ and $b'(x)$
as the solution that the greedy algorithm for fractional knapsack
provides for knapsack size $x$ and knapsack problems $\k_a$ and
$\k_b$, respectively. Similarly, we denote by the solution of $c'(x)$
the solution that Algorithm
\ref{alg:greedyalgorithmforfractionalknapsackconvolution} provides for
the fractional convolution of $a \ttimes b$ with respect to knapsack
size $x$. We say two solutions differ in at most one item, if they are
the same except for one item.

\begin{proof}[of Observation~\ref{observation:inc}] 
	We assume w.l.o.g. that $y$ and $y'$ are close enough to make sure the solution of $c'(x+y)$ differs from the solution of $c'(x+y')$ by at most one item. Similarly, we assume w.l.o.g. that the solutions of $b'(y)$  and $b'(y')$ differ by at most one item.
	If the statement of Observation~\ref{observation:inc} is correct for such $y$ and $y'$ then it extends to all $y < y'$ in range $[0,\fa^{-1}(x)-x]$ since for every $y < y'$ one can write $y < y_1 < y_2 < \ldots< y'$ such that every two consecutive elements are close enough. Therefore, the statement holds for any pairs of consecutive elements and thus holds for $y$ and $y'$. In order to compare $c'(x+y) - a'(x) - b'(y)$ with $c(x+y') - a'(x) - b'(y')$ it only suffices to compare $c'(x+y') - c'(x+y)$ with $b'(y') - b'(y)$.
	
	It follows from the monotonicity of $\fa$ that since $y < y' < \fa^{-1}(x)-x$, then $\fa(x+y) \leq \fa(x+y') \leq x$ and therefore $\fb(x+y') \geq y'$. Let $(s_i, w_i)$ be the last item in the solution of knapsack problem $\k_b$ for knapsack size $y'$. Hence, due to Algorithm~\ref{alg:greedyalgorithmforfractionalknapsackconvolution}, any item not included in the solution of $c'(x+y)$ has a ratio of weight over size which is upper bounded by $w_i/s_i$. Therefore, $c'(x+y') - c'(x+y) \leq (y'-y)(w_i/s_i)$. Since $b'(y)$ and $b'(y')$ differ in at most one item we have $b'(y) - b'(y) = (y'-y)(w_i/s_i)$. Thus, $c'(x+y') - c'(x+y) \leq b'(y) - b'(y)$ and therefore $c'(x+y') - a'(x) - b'(y') \leq c'(x+y) - a'(x) - b'(y)$.
\end{proof}

\begin{proof}[of Observation~\ref{observation:dec}] 
	The proof is similar to that of Observation~\ref{observation:inc}. We assume w.l.o.g. that the solutions of $c'(x+y)$ and $c'(x+y')$ differ in at most one item and also the solutions of $b'(y)$ and $b'(y')$ differ in at most one item. By monotonicity of $\fa$ we have $\fa(x+y') \geq x$ and thus $\fb(x+y') \leq y'$. This means that if $(s_i, w_i)$ is the last item of $c'(x+y')$ then any item not included in $b'(y)$ has a ratio of weight over size of at most $w_i/s_i$. This implies that $b'(y') - b'(y) \leq (y'-y)w_i/s_i = c'(x+y') - c'(x+y)$ which implies $c'(x+y) - a'(x)+b'(y) \leq c'(x+y') - a'(x) - b'(y')$.
\end{proof}

\begin{proof}[of Observation~\ref{observation:mon1}] 
	Since $x < x'$ and $y \leq \fa^{-1}(x) - x$ then we have $\fb^{-1}(y) -y \leq x < x'$. Since $a'$ and $b'$ are symmetric, this observation follows from Observation~\ref{observation:dec}.
\end{proof}

\begin{proof}[of Observation~\ref{observation:mon2}] 
	Similar to Observation~\ref{observation:mon1}, we have $x < x' \leq \fb^{-1}(y)-y$ and the observation reduces to Observation~\ref{observation:dec} by switching $a'$ and $b'$.
\end{proof}
\newpage
\section{Omitted Proof of Section~\ref{sec:treeseparability:reduction}}\label{appendix:treeseparability}
\begin{proof}[of Lemma~\ref{lemma:red2}]
As aforementioned, the proof follows from the ideas of Cygan \etal~\cite{cygan2017problems} and Backurs \etal~\cite{backurs2017better}. We assume that the tree is rooted at some arbitrary vertex and for a vertex $v$ we refer to the subtree rooted at $v$ by $T(v)$. We say the solution of a subtree $T(v)$ is a vector $a$ of size $|T(v)|+1$, where every $a_i$ denotes the minimum cost for putting the vertices of $T(v)$ into two disjoint components of sizes $i$ and $|T(v)|-i$ where vertex $v$ itself in the part with size $i$.

Let $u$ and $v$ be two disjoint subtrees of the graph and denote by $a$ and $b$ the solutions of these subtrees. If we add an edge from $u$ to $v$ and wish to compute the solution for the combined tree, one can derive the solution vector from $a$ and $b$. To this end, there are two possibilities to consider: either $u$ and $v$ are in the same component in which case the solution is equal to $a \ttimes b$. Otherwise, we construct a vector $b'$ where $b'_i = 1+b_{|T(v)|-i}$ and compute $a \ttimes b'$ to find the answer Therefore, merging the solutions of two subtrees reduces to computing the convolution of two solution vectors.

If the tree is balanced, and therefore the height of the tree is $O(\log n)$, the standard dynamic program yields a running time of $T(n)$ since the total lengths of the convolutions we make is $O(n \log n)$. However, if the tree is not balanced, in the worst case, the height of the tree also appears in the running time. A classic tool to overcome this challenge is the spine decomposition of~\cite{sleator1983data} to break the tree into a number of spines. Every spine is a path starting from a vertex and ending at some leaf. We say a spine $x$ is above a spine $y$, if there exists a vertex in $y$ such at least one parent of that vertex appears in $x$. We denote this relation with $x \prec y$. Such a spine decomposition satisfies the property that every sequence of spines $x_1 \prec x_2 \prec \ldots \prec x_k$ has a length of at most $O(\log n)$. This enables us to solve the dynamic program in time $\tildorder(T(n))$. The overall idea is that instead of updating every vertex at each stage, we update the solution of a spine.  Since every spine is a path, we can again reduce the problem of combining the solutions of a spine to convolution. This way, the height of the updates reduces to $O(\log n)$ and thus our algorithm runs in time $\tildorder(T(n))$. We refer the reader to \cite{cygan2017problems} and~\cite{backurs2017better} for a formal proof.
\end{proof}

\end{document}